\documentclass[10pt,conference,letterpaper]{IEEEtran}
\usepackage{amsmath,epsfig}
\usepackage{booktabs} % For formal tables
\usepackage{graphicx}
\usepackage{svg}
\usepackage{epstopdf}
\usepackage{verbatim}
\usepackage{lipsum}
\usepackage{mathtools}
\usepackage{float}
\usepackage{multirow} 
\usepackage{import}
\usepackage{tabu}
\usepackage[linesnumbered,ruled,vlined,boxed]{algorithm2e}
\usepackage[export]{adjustbox}
\usepackage{quoting}
\usepackage{array}
\usepackage{tabularx}
\usepackage{listings}
\usepackage{color}
\usepackage{frcursive}
\usepackage{calligra}
\usepackage[T1]{fontenc}

\usepackage{balance}
\usepackage{listings}

\newcommand{\dscaler}{{\textsc{Dscaler}}} 

\newcommand{\gscaler}{{\textsc{Gscaler}}} 
\newcommand{\oldd} {{\mathcal{D}}}
\newcommand{\scaled} {{\widetilde{\mathcal{D}}}}
\newcommand{ \oldf }{{\mathcal{F}}}
\newcommand{\scalef }{{\widetilde{\mathcal{F}}}} 

\newcommand{\tweakf} {{\widetilde{\mathcal{F}}}} 
\newcommand{ \tweakd} {{\widetilde{\mathcal{D}}}}

\newcommand{\scalepo}{\po}
\newcommand{\tweakpo}{{\widetilde{\po}}}

\newcommand{\Xiami} {{\tt Xiami}}
\newcommand{\DoubanMovie }{{\tt DoubanMovie}}
\newcommand{\DoubanMusic }{{\tt DoubanMusic}}
\newcommand{\DoubanBook }{{\tt DoubanBook}}
\newcommand{ \Douban }{{\tt Douban}}

\newcommand{ \ReX} {{ReX}}

\newcommand{\Tlinear }{{\mathcal T}_{\tt linear}} 
\newcommand{ \Tpairwise}{{\mathcal T}_{\tt pairwise}} 
\newcommand{\Tcoappear }{{\mathcal T}_{\tt coappear}} 
\newcommand{\T }{{\mathcal T}}
\newcommand{ \linear} {{linear}}
\newcommand{\coappear}{{coappear}}
\newcommand{\pairwise}{{pairwise}}

\newcommand{\rex}{{ReX}}
\newcommand{\rand}{{Rand}}

\newcommand{ \response }{{\tt response2post}}
\newcommand{ \post }{{\tt post}}

\newcommand{\tweakH}{{\widetilde{H}}}
\newcommand{\scaleH}{{H}}
\newcommand{\tweakh}{{\widetilde{h}}}
\newcommand{\scaleh}{{{h}}}

\newcommand{ \scalexi}{\xi}
\newcommand{ \tweakxi}{{\widetilde{\xi}}}

\newcommand{ \po} {{\rho}}

\newcommand{ \sF} {{\widetilde{F}}}

\newcommand{\F}{\mathcal{F}}

\newcommand{ \upsizer }{{UpSizeR}}

\newtheorem{definition}{Definition}
\newtheorem{theorem}{Theorem}
\newtheorem{lemma}{Lemma}

\title{A tool framework for tweaking features in synthetic datasets}
\author{
	{J.W. Zhang{\small $~^{\#1}$},  Y.C. Tay{\small $~^{\#2}$} }%
	\vspace{1.6mm}\\
	\fontsize{10}{10}\selectfont\itshape
	$^{\#}$\,School of Computing, National University of Singapore\\
	\fontsize{9}{9}\selectfont\ttfamily\upshape
	$^{1}$\,jiangwei@u.nus.edu\\
	$^{2}$\,dcstayyc@nus.edu.sg%
	\vspace{1.2mm}\\
	\fontsize{10}{10}\selectfont\rmfamily\itshape
}

\begin{document}	
	\maketitle
	
	\begin{abstract}
		Researchers and developers use benchmarks to 
		compare their algorithms and products.
		A database benchmark must have a dataset $\oldd$.
		To be application-specific, 
		this dataset $\oldd$ should be empirical.
		However, $\oldd$ may be too small, or too large,
		for the benchmarking experiments.
		$\oldd$ must, therefore, be scaled to the desired size.
		
		To ensure the scaled $\scaled$ is similar to $\oldd$,
		previous work typically specifies or extracts 
		a fixed set of features $\F = \{\F_1, \F_2, \dots, \F_n\}$ from $\oldd$,
		then uses $\F$ to generate synthetic data for $\scaled$.
		However, this approach ($\oldd \rightarrow$ $\F$ $\rightarrow \scaled$ ) 
		becomes increasingly intractable as $\F$ gets larger,
		so a new solution is necessary.
		
		Different from existing approaches, 
		this paper proposes ASPECT to scale $\oldd$ to enforce similarity.
		ASPECT first uses a size-scaler ($S_0$) to scale 
		$\oldd$ to $\scaled$. 
		Then the user selects a set of desired features $\scalef_1,\ldots,\scalef_n$.
		For each desired feature $\scalef_k$, 
		there is a tweaking tool $\T_k$ that tweaks $\scaled$ to 
		make sure $\scaled$ has the required feature $\scalef_k$.
		ASPECT coordinates the tweaking of $\T_1,\ldots,\T_n$ to $\scaled$,
		so $\T_n(\cdots(\T_1(\scaled))\cdots)$ has the required features 
		$\scalef_1,\ldots,\scalef_n$.
		
		By shifting from $\oldd \rightarrow \F \rightarrow \scaled$  
		to $\oldd \rightarrow  \scaled \rightarrow \scalef$, 
		data scaling becomes flexible. 
		The user can customise the scaled dataset with their own interested features.
		Extensive experiments on real datasets show that ASPECT
		can enforce similarity in the dataset effectively and efficiently. 
	\end{abstract}

	\section{Introduction}
	We have two motivations for introducing ASPECT:
	
	\label{sec:intro}
	{\bf Motivation 1}:
	Benchmarks are ubiquitous in the computing industry and academia.
	Developers use benchmarks to compare products and algorithms,
	while researchers use them similarly in research.
	
	For 20-odd years, 
	the popular benchmarks for database management systems
	were the ones defined by the Transaction Processing Council 
	(TPC)~\footnote{http://www.tpc.org/}.
	However, the small number of TPC benchmarks are increasingly irrelevant
	to the myriad of diverse applications,
	and the TPC standardization process is too slow~\cite{stonebreaker}.
	This led to a proposal for a paradigm shift,
	from a top-down design of domain-specific benchmarks 
	by committee consensus,
	to a bottom-up collaboration to develop tools 
	for application-specific benchmarking~\cite{vision}.
	
	A database benchmark must have a dataset.
	For the benchmark to be application-specific,
	it must start with an empirical dataset $\oldd$.
	This $\oldd$ may be too small or too large for the benchmarking experiment,
	so the first tool to develop would be for scaling $\oldd$ to a desired size.
	
	{\bf Motivation 2}:  
	Apart from benchmarking, 
	dataset scaling plays important roles in other fields as well.
	A start-up company with a small dataset 
	may want a larger dataset for testing the scalability of their system architecture. 
	On the other hand, an enterprise with a large dataset
	may want a scaled down version to provide quick answers
	to aggregation queries (averages, count, etc.).

	\begin{figure}[t]
		\centering
		\includegraphics[height=0.38in]{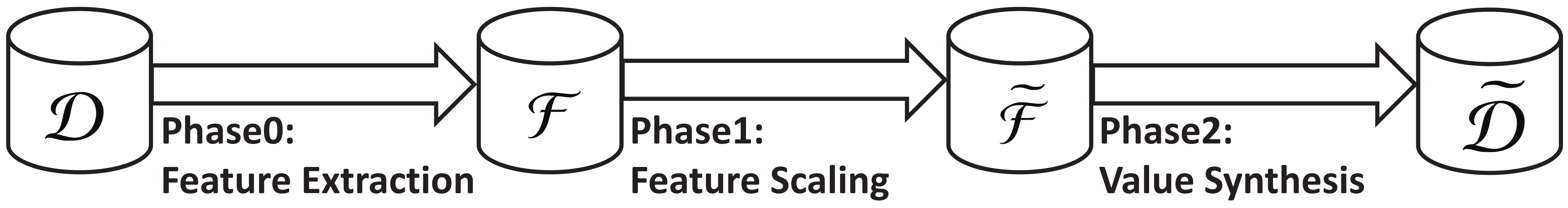}
		\caption{Existing scaling framework}
		\label{fig:oldframework}
		\vspace{-4mm}
	\end{figure}
	
	Given this outlook, 
	a tool that scales an {\it empirical} dataset $\oldd$ to 
	a {\it synthetic} and {\it similar} $\scaled$ will be very appealing.
	This generation of artificial data is necessary if $\scaled$ is larger,
	and helpful if $\scaled$ is smaller or equal in size~\cite{sdv,upsizer}.
	For all cases: 
	$\scaled$  must be {\it similar} to $\oldd$.
	Moreover, the similarity definition should be application-specific. 
	It can be measured by graph properties, query results, etc., 
	depends on the application.
	
	\subsection{Existing approach and the limitations}
	\label{sec:existingLimitation}
	To ensure $\scaled$ is similar to $\oldd$,
	previous work~\cite{chronos, sdv, upsizer, dscaler} 
	typically follows the framework  in Fig.\ref{fig:oldframework}.
	Each algorithm extracts a fixed set of features 
	$\F = \{\F_1, \F_2, \dots, \F_n\}$  from $\oldd$, 
	then scales  $\F$ to $\bf \widetilde{\F}$ as a predicting feature 
	for the scaled dataset $\scaled$.
	$\scaled$ is finally synthesized based on $\bf \widetilde{\F}$.
	$\F$ here defines the similarity between $\oldd$ and $\scaled$:
	the more features in $\F$, 
	the greater the similarity between $\oldd$ and $\scaled$. 
	
	For example, 
	if $\oldd$ is a graph and $\F = \{ \F_1, \F_2\}$, 
	where $\F_1$ is \textit{density} 
	and $\F_2$ is \textit{number of triangles}, 
	then we would expect $\oldd$ and $\scaled$ are similar 
	in terms of \textit{density} and \textit{triangles}. 
	However, there are some limitations from the perspective of a developer and a user:
	
	\subsubsection{The developer faces the implementation reusability and scalability issue} 
	{\bf Algorithm Implementation Reusability}: 
	Consider the scenario where one application developer 
	implements an algorithm $A_1$ using the feature set $\{\F_1, \F_2\}$. 
	Later, another developer may find it more important 
	for her application to preserve $\{\F_2, \F_3\}$,
	where $\F_3$ is the \textit{number of rectangles}.
	So she implements another algorithm $A_2$ to preserve $\{\F_2, \F_3\}$.
	However, a third developer might want to preserve $\{\F_1, \F_2, \F_3\}$;
	what should he do?
	In this case, $A_1$ or $A_2$ only preserves part of $\{\F_1, \F_2, \F_3\}$.
	To preserve  $\{\F_1, \F_2, \F_3\}$, 
	one has to modify $A_1$ to preserve the extra $\F_3$, 
	or modify $A_2$ to preserve the extra $F_1$, 
	or write a new algorithm $A_3$ from scratch.
	This is a waste of effort, 
	since $\F_1$, $\F_2$, $\F_3$ were already preserved by $A_1$ and $A_2$.
	{\bf Algorithm Implementation Scalability}: 
	As mentioned previously,  the more features in $\F$, 
	the greater the similarity between $\oldd$ and $\scaled$. 
	However, a large feature set dramatically increases the difficulty of
	designing an algorithm that maintains the features simultaneously. 
	For example, if $\F = \{\F_1, \F_2, \F_3, \F_4 \}$, 
	where $\F_4$ is  {\it the fraction of nodes with degree} $1$,
	then it is less likely one can design a single algorithm 
	which preserves all $4$ features.
	If we only consider degree distributions as features, 
	then it is already NP-hard to decide whether 
	there exists a graph satisfying certain degree distributions~\cite{gmark}.
	
	\subsubsection{The user does not have a choice of the features} 
	In the current framework, 
	once an algorithm is implemented, the features are fixed.
	Consider the same example used above, 
	$A_1$ is implemented to preserve $\F=\{\F_1, \F_2\}$, 
	and $A_2$ is implemented to preserve $\F=\{\F_2, \F_3\}$. 
	The user can only choose to preserve $\{\F_1, \F_2\}$ 
	or $\{\F_2, \F_3\}$, but not the union.
	
	\subsection{Overcoming the limitations}
	In this paper, we propose ASPECT, 
	a flexible framework for synthetic data scaling.
	Unlike existing approaches, 
	ASPECT takes the following two steps as illustrated in Fig.\ref{fig:featuretweaking}:
	
	\noindent
	{\bf Step1:}  
	Use a size-scaler $S_0$ to scale $\oldd$ to $\scaled_0$ of desired size.
	
	\noindent
	{\bf Step2:} 
	For the desired feature set $\{\scalef_1, \scalef_2,\dots,\scalef_k\}$, 
	apply \textit{independently} developed tools 
	$\T_1, \T_2, \dots, \T_k$ on $\scaled_0$ in order.
	Each tool $\T_i$ generates a dataset $\scaled_i$ 
	by adjusting $\scaled_{i-1}$.
	After the adjustment of $\T_i$, 
	$\scaled_i$ satisfies \{$\scalef_1, \scalef_{2}, \dots, \scalef_i$\}.
	Note that Step2 does not depend on the size-scaler in Step1.
	We call this {\bf tweaking} $\scaled_{i-1}$ by tool $\T_i$.
	The final dataset is $\scaled$.
	
	For the above-mentioned limitations in Sec.\ref{sec:existingLimitation}, 
	ASPECT resolves them with ease: 
	For \textbf{implementation reusability}, 
	each feature tweaking tool is independently developed. 
	Once a tweaking tool $\T_i$ for feature $\scalef_i$ is implemented, 
	then the user can apply $\T_i$ together with other tweaking tools 
	whenever it is needed. No re-coding! 
	For \textbf{implementation scalability}, 
	to preserve the feature set with $n$ features, 
	the developer just needs to implement $n$ tweaking tools, 
	instead of hardcoding all $n$ features into a single piece of software.
	And each tweaking tool $\T_i$ tweaks the feature $\scalef_i$. 
	For the issue of \textbf{feature choice}, 
	once tweaking tools $\T_1$, $\T_2$, $\T_3$ and $\T_4$
	are implemented for $\sF_1$, $\sF_2$, $\sF_3$ and $\sF_4$ respectively, 
	the user can choose $\{\T_2, \T_4\}$ to get $\{\sF_2, \sF_4\}$, 
	or $\{\T_1, \T_3, \T_4\}$ to get $\{\sF_1, \sF_3, \sF_4\}$, etc.
	
	Hence, to enforce greater similarity in the scaled dataset $\scaled$, 
	we just need to apply more tweaking tools.
	We envision having developers from the database community 
	contributing tools $\T_i$ to a repository for tweaking synthetic datasets. 
	Then ASPECT will have more tools for the user to customise the scaled datasets.
	This would go some way towards realising the suggested paradigm shift
	to a bottom-up collaboration for application-specific benchmarking.
	
	\begin{figure}[t]
		\centering
		\includegraphics[height=0.55in]{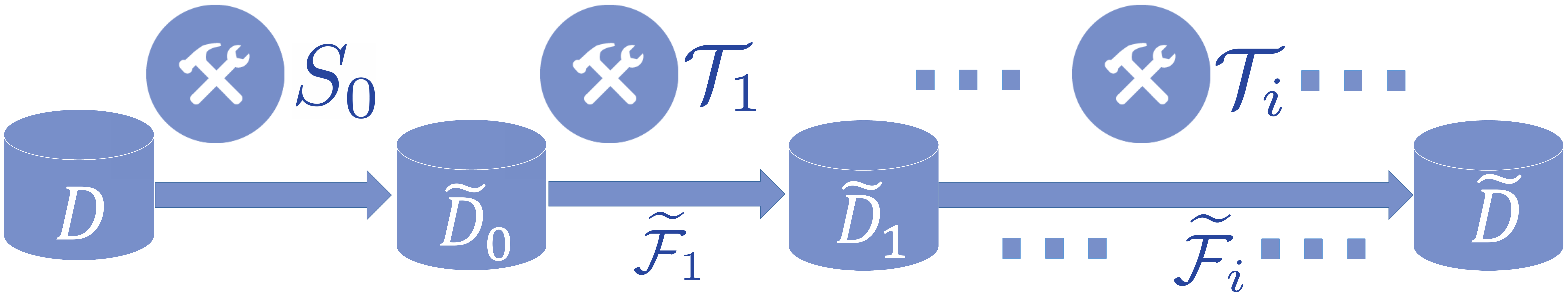}
		\caption{The ASPECT abstract -- 
			The empirical dataset $\oldd$ 
			is first scaled to $\widetilde{D}_0$ by a size-scaler $S_0$.
			After that, tools 
			$\T_1,\T_2,\dots, \T_n$ sequentially tweaks dataset for 
			features $\scalef_1,\scalef_2,\dots, \scalef_n$.
			After $\T_i$ tweaks $\scaled_{i-1}$, this results in $\scaled_i$. 
			\{$\scalef_1, \scalef_{2}, \dots, \scalef_i$\} is reflected in $\scaled_i$.
		}
		\label{fig:featuretweaking}
		\vspace{-4mm}
	\end{figure}
	
	However, the tools require some coordination,
	since some changes to $\scaled_{i-1}$ by one tool $\T_i$ 
	may be undone by another tool $\T_j$.
	Moreover, tools are developed independently by different developers. 
	Different developers might have different tweaking techniques. 
	To ensure a tool is compatible with ASPECT,  
	ASPECT must explicitly standardise types of modifications 
	that could be made on a dataset while tweaking. 
	
	\subsection{Overview}
	To summarize, our contribution in this paper are fourfold:
	\begin{itemize}
		\item[1.]
		We propose ASPECT, a framework for flexible application of tweaking tools to 
		enforce target features in synthetic dataset.	
		
		\item[2.]
		We present results from extensive experiments on real datasets,
		to verify that ASPECT can enforce similarity in the dataset 
		effectively and efficiently. 	
		
		\item[3.]
		We present necessary and sufficient conditions,
		and tweaking algorithms, for three new complex features.
		
		\item[4.] 
		We state Feature Tweaking Bound and Order Problems 
		that offer a rewarding challenge for research on dataset tweaking.
	\end{itemize}

	We first introduce ASPECT architecture in Sec.\ref{sec:framework}, 
	followed by three new complex features that serve to 
	illustrate the ASPECT framework in Sec.\ref{sec:featuretweaker}.
	One of them concerns inter-column and inter-row correlation induced by
	implicit relationships in a social network dataset;
	we thus provide here a solution to a problem highlighted  previously~\cite{vision}.
	Sec.\ref{sec:experiment} describes the datasets and similarity measures used in the experiments,
	and the results are presented in Sec.\ref{sec:results}.
	Sec.\ref{sec:limit} points out some limitations and insights of ASPECT.  
	Related work is surveyed in Sec.\ref{sec:related}, 
	before Sec.~\ref{sec:conclusion} concludes with a summary.

	\section{ASPECT Architecture}
	\label{sec:framework}
	
	As shown in Fig.\ref{fig:featuretweaking},
	an input dataset is first scaled by the size-scaler $S_0$ 
	which returns a scaled dataset $\scaled_0$ of the desired size.
	Note that $S_0$ could be any tool 
	which guarantees the number of tuples in each table generated is as expected 
	and there are no invalid foreign key values.
	For example, $S_0$ may be $\dscaler$~\cite{dscaler}, 
	or it could be $\rex$~\cite{rex};
	we will show that ASPECT is able to preserve the features well 
	for both $\dscaler$ and $\ReX$ in Sec.\ref{sec:results}.
	The choice of $S_0$ is outside the scope of this paper.
	After the dataset is resized, 
	ASPECT then coordinates the application of tools $\T_i$ on $\scaled_{i-1}$, 
	to make sure the feature $\tweakf_i$ is reflected 
	in the tweaked dataset $\scaled_i = \T_i(\scaled_{i-1})$.
	In the tweaking process, there are a few issues:
	
	{\bf I1.} 
	How do we get the target feature $\tweakf_i$?
	
	{\bf I2.} 
	Given an target feature $\tweakf_i$,
	how can we tweak $\scaled_{i-1}$ to ensure that
	the tweaked dataset $\tweakd_{i}$ contains $\tweakf_i$?
	
	{\bf I3.} 
	Given $\scaled_{n-1}$ already contains 
	$\tweakf_1, \tweakf_2,\dots, \tweakf_{n-1}$,
	how can we maintain  $\tweakf_1, \tweakf_2,\dots,\tweakf_{n-1}$ 
	while tweaking $\tweakf_n$? 
	
	{\bf I4.} 
	Tools are developed independently by different developers.
	How can we make sure these independently developed tools are compatible with ASPECT?

	\subsection{ASPECT flow}
	To address the above 4 issues, 
	we illustrate the tweaking process for, 
	say, $\T_4$.
	
	{\bf Step1.} 
	ASPECT first calls the tool $\T_4$ to start tweaking, 
	then calls previously applied tools $\T_1, \T_2, \T_3$ to start preparation. 
	
	{\bf Step2.1.} 
	$\T_4$ then finds the target feature $\tweakf_4$ 
	by calling its \textit{Feature Generator}. 
	
	{\bf Step2.2.} 
	$\T_1, \T_2, \T_3$ call their respective \textit{Feature Calculators} to 
	calculate the corresponding features $\scalef_1, \scalef_2, \scalef_3$. 
	This step is concurrent with Step 2.1.
	
	{\bf Step3.} 
	$\T_4$ starts the \textit{Tweaking Algorithm}. 
	Every time $\T_4$ needs to modify $\scaled_3$, 
	$\T_4$ sends the intended modification to ASPECT for validation.
	
	{\bf Step4.} 
	ASPECT calls  $\T_1, \T_2, \T_3$ to confirm the modification 
	with their own \textit{Feature Validators}. 
	ASPECT summarizes the feedback from $\T_1, \T_2, \T_3$ 
	and replies ``yes/no'' to  $\T_4$.
	
	{\bf Step5.1.} 
	If the reply is ``yes'',  
	ASPECT modifies $\scaled_3$ and tells $T_1, \T_2, \T_3$ to 
	update the feature statistics by using their \textit{Feature Updators}.
	
	{\bf Step5.2.} 
	If the reply is ``no'',  
	ASPECT tells $\T_4$ to find an alternative modification.
	
	{\bf Step6.} Repeat from Step3 until $\T_4$ halts.
	
	In the flow above, 
	I1 is addressed by Step2.1 (\textit{Feature Generator}), 
	I2 is addressed by Step3 (\textit{Tweaking Algorithm}) and 
	I3 is handled in Step2.2, Step4 and Step5 
	(\textit{Feature Calculator, Feature Validator, Feature Updator}). 
	We will explain I1, I2, and I3 
	in Sec.\ref{sec:toolComponent}, 
	and I4 in Sec.\ref{sec:protocol} in detail.
	
	\subsection{Tweaking tool component}
	\label{sec:toolComponent}
	We now explain how each individual tool $\T_i$ should be implemented. 
	Each $\T_i$ must have at least 5 components.
	
	{\textit{\textbf{Feature Generator}}}: 
	This module generates the target feature statistics for the tweaked dataset.
	Such a generation can be done in 3 ways:
	(i) \textit{\textbf{User input}}: 
	The user might have their own target feature statistics for the scaled dataset. 
	For example, the user might want to specify
	the number of {males} in the population.	
	Hence, the user can manually input target features.
	(ii) \textit{\textbf{Developer generation}}: 
	When a developer implements the tweaking tool for a specific feature, 
	the developer has a better understanding of 
	how the feature changes while the dataset scales. 
	Therefore, the developer can provide the feature generation tool
	for his/her own developed feature. 
	(iii) \textit{\textbf{Generate through historical data}}:
	Apart from the previous two methods, 
	statistical tools can be developed for certain features for general purposes, 
	e.g. frequency distribution for attribute values.
	One can first take chronological snapshots of the dataset (if applicable), 
	$\oldd_1, \oldd_2,\dots, \oldd_k$,
	then extract the feature $\oldf_i$ from each snapshot dataset. 
	Next, apply data fitting methods on $\oldf_i$ to  
	fit $\oldf_i$  into different statistical models, 
	e.g. Poisson distribution. 
	Once the best matching model is learned, say, Poisson distribution,  
	we can learn how the shape parameter $\lambda$ varies as the dataset grows.
	Hence, we can get the target feature $\scalef_i$.
	Such an approach is orthogonal to this paper, 
	and will be elaborated in a separate paper.

	\textit{\textbf{Tweaking Algorithm}}: 
	It tweaks the dataset $\tweakd_{i-1}$ to make sure that 
	$\tweakd_{i}$ has the target feature $\scalef_i$ at the end of tweaking. 
	Note that it is not trivial to provide a tweaking algorithm for a complex feature, 
	e.g. $\linear$ feature as presented in Sec.~\ref{sec:featuretweaker}.
	The developer has to code the tweaking algorithm. 
	Moreover, the tweaking algorithm can only modify the dataset through 
	similar operations presented in Fig. \ref{fig:interface}
	
	\textit{\textbf{{Feature Calculator}}}: 
	It calculates the feature statistics for $\scalef_i$ from a given dataset.

	\textit{\textbf{Feature Validator}}: 
	It checks whether a proposed tuple insertion/deletion/replacement 
	affects some existing feature.
	Assuming 1 modification on a tuple is needed when 
	tweaking feature $\scalef_n$ on dataset $\scaled_{n-1}$.
	Modifying either $t_1$ or $t_2$ will satisfy $\scalef_n$,
	and modifying $t_1$ changes a 
	previously tweaked feature $\tweakf_1$,
	but modifying $t_2$ does not, 
	Then, modify $t_2$ instead of $t_1$. 
	
	At times, it is too strict when a tuple modification is allowed 
	only if no previously tweaked feature affected.  
	For example, 
	$\F_1$ = \{more than half of the customers are men\} 
	and $\F_2$ = \{more than half of the customers are women\}.
	These $2$ features are contradictory, 
	so one of it has to be violated.
	Hence, the validation needs to be relaxed. 
	In this paper, a tuple modification is allowed if 
	the resulting errors for all previously tweaked features are less than 5\%, 
	which is the threshold $e_{threshold}$.
	Consider the example used previously: 
	if error $< e_{threshold}$, 
	then modification on $t_1$ is allowed as well; 
	however, if error $\geq e_{threshold}$,
	then only $t_2$ can be chosen.
	
	In the worst case, if no tuple modification can satisfy
	previously tweaked feature's $e_{threshold}$,
	ASPECT allows more tuple modification by relaxing the validation on fewer features.
	So, some of the feature's error might be larger than $e_{threshold}$.
	
	\textit{\textbf{Feature Updater}}:
	After each modification of the tuples, 
	the\textit{ Feature Updater} 
	updates the tweaked features' statistics.
	
	Under ASPECT, it is the tool developers' responsibility
	to ensure that the above requirements are correctly implemented
	and adhere to ASPECT's structure. If the developer 
	does not, say, properly validate the modifications for
	\textit{Feature Validator}, then it is highly likely
	the corresponding feature will be affected by subsequent tweaks. 
	
	\begin{figure}[t]
		
		%	\begin{lstlisting}
		%	
		%public interface Updator {
		%	
		%	deleteValues(tableID, colIndexes, tupleIDs){}
		%		
		%	insertValues(tableID, colIndexes, insertingTupleIDs, colValues) {}    
		%	
		%	replaceValues(tableID, colIndexes, replacingTupleIDs, colValues) {}  
		%}
		%	\end{lstlisting}
		\centering
		\includegraphics[height=0.7in]{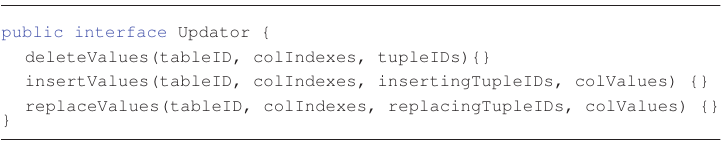}
		\vspace{-2mm}
		\caption{Feature Updator Interface}
		\label{fig:interface}
		\vspace{-4mm}
	\end{figure}
	
	\subsection{ASPECT compatibility guarantee}
	\label{sec:protocol}
	As mentioned in I4, 
	we need to guarantee each individually developed tool 
	can be used in ASPECT. 
	So, we standardise interface for the \textit{Feature Updator} and  \textit{Feature Validator}. 
	By having the same structure, 
	all tweaking tools will be compatible with ASPECT.
	Fig.\ref{fig:interface} presents the common functions 
	to be implemented for any \textit{Feature Updator}. 
	\textit{Feature Validator} uses a similar interface, 
	so it is omitted here.
	
	We classify three types of modifications that 
	can be made on a table with $n$ tuples and $k$ columns excluding primary key.

	{\tt \bf deleteValues}: 
	Given a list of tuples ({\tt tupleIDs}) in a table ({\tt tableID}),
	this operation erases some columns ({\tt colIndexes}) of these tuples. 
	Note that the deleted entries are temporarily empty, 
	new values will be added back via {\tt insertValues}.
	In Fig.\ref{fig:demoImplementation}, Step1 is a {\tt deleteValues} operation. 
	It deletes the first and third column's values of first and second tuple.
	The erased entries are empty after Step1.
	
	{\tt \bf insertValues}: 
	Given a table  ({\tt tableID}), 
	a list of values $\langle v_1,\dots, v_d \rangle$ ({\tt colValues}), 
	some columns $\langle c_1,\dots,c_d \rangle$,
	this operation adds this 
	$\langle v_1,\dots, v_d \rangle$ 
	into the tuples $t_1,t_2,\dots,t_m$ ({\tt insertingTupleIDs}), 
	where $v_1$ is the value for column $c_1$. 
	These values can only be inserted to the empty entries 
	resulting from {\tt deleteValues}.
	Moreover, the total number of inserted values 
	is the same as the total number of deleted values.
	In Fig.\ref{fig:demoImplementation}, 
	Step2 is an {\tt insertValues} operation. 
	It inserts  [4,4] into the first and third column of second tuple.

	{\tt \bf replaceValues}: 
	Given a table ({\tt tableID}) and a list of its tuples ({\tt replaceingTupleIDs}), 
	this operation replaces some columns ({\tt colIndexes}) 
	of these tuples with new values ({\tt newValues}).
	All these tuples will have the same values 
	({\tt newValues}) for the replaced attributes.
	\noindent {\tt replaceValues} is different from {\tt insertValues} where 
	the replacing entries must not be empty entries. 
	In Fig.\ref{fig:demoImplementation}, 
	Step3 is a {\tt replaceValues} operation.
	It replaces the first, second and third columns of second and third tuples with  [7,7,7].

	\begin{figure}
		\centering
		\includegraphics[height=1.1in]{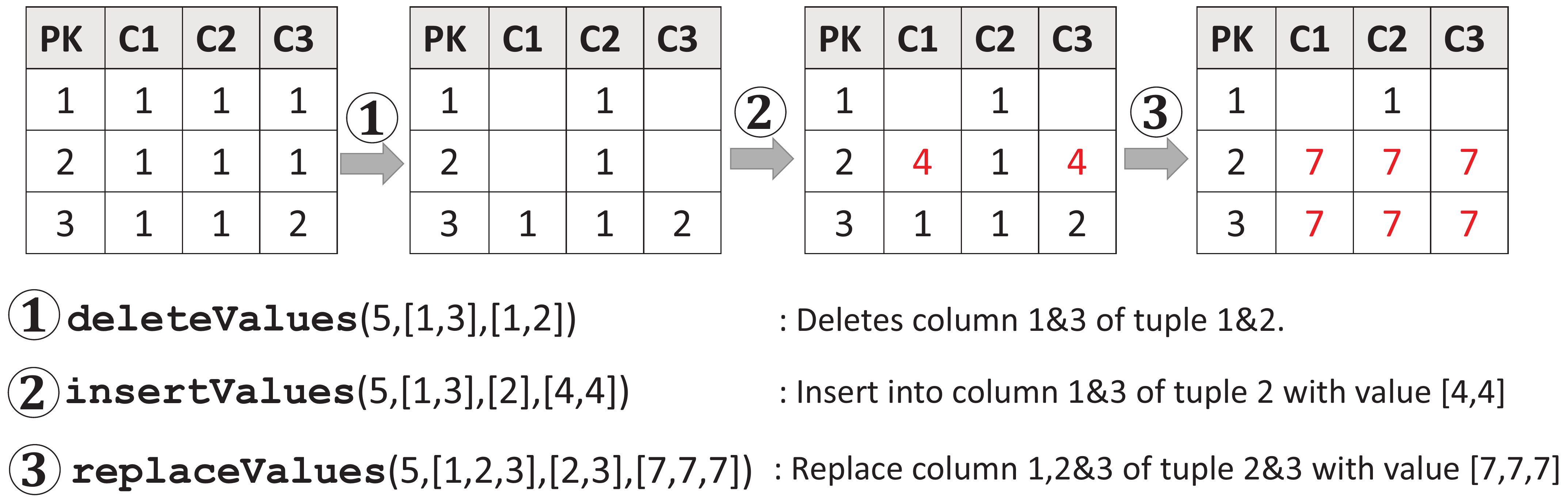}
		\caption{Demonstration for Updator Interface: 
			The update is on table with tableID 5.
			Step1 demonstrates  {\tt \bf deleteValues}; 
			Step2 demonstrates  {\tt \bf insertValues}; 
			Step3 demonstrates  {\tt \bf replaceValues}.
			The affected values are highlighted using red colour, 
			and the operation for each step is listed at the bottom.}
		\label{fig:demoImplementation}
		\vspace{-4mm}
	\end{figure}

	Any feature tweaking tool developer has to implement 
	the functions in Fig.\ref{fig:interface}
	to update the corresponding features. 
	The types of modifications in ASPECT are not exhaustive. 
	There can be other types of modifications that one wants to make. 
	The developer should transform other modifications to 
	the three basic modifications in Fig.\ref{fig:interface}.
	In this case,
	we sacrifice some accuracy to 
	favour generality for ASPECT and ease of programming for the developers. 
	If one wants to validate a modification on two tables A and B,
	one can validate it on table A and then validate it on table B, 
	then accept the modifications if both tables' validations are successful.
	We might have false-positive/false-negative cases for these types of validations,
	but we believe the effect is minor.
	In Sec.~\ref{sec:Tcoappear},
	we need to modify $k$ tables simultaneously.
	To fit into ASPECT, we validate/modify the tables one by one.
	Experiments in Sec.~\ref{sec:results} show this approach is effective.

	\section{Example Features For Demonstration}
	\label{sec:featuretweaker}	
	In ASPECT, the more tools we apply, the more features we can preserve.
	As pointed out previously, for some feature $\sF_j$, 
	it is inevitable that some previously tweaked feature $\sF_i$ 
	may be affected when tweaking $\sF_n$.
	The concern is how much $\sF_i$ is affected while tweaking $\sF_j$ under ASPECT.
	
	In this paper, 
	we run experiments to demonstrate 
	how $\sF_i$ is affected while tweaking $\sF_j$ empirically. 
	We will apply $3$ tools $\T_1, \T_2, \T_3$ sequentially 
	on $\scaled_0$ to fix the features $\sF_1, \sF_2, \sF_3$.  
	At the end of the tweaking,
	we will examine how well the three features are preserved.
	We propose three important and complex features: 
	$\linear$, $\coappear$, $\pairwise$ as example features.  
	These features are selected based on two criteria: 
	
	\textbf{Popularity}: 
	Our ultimate goal is to build an application-specific system for dataset scaling.
	Hence, it only makes sense if the features are widely used.
	The features we consider are used widely in the literature 
	~\cite{joinsyn,linkedSyn,recommendation,volunteer}. 
	
	\textbf{Complexity}: 
	Since we want to check how well ASPECT 
	can maintain features if tools can undo previously applied tools.
	Hence, the features we use for demonstration should be complex and affect each other. 
	Simple features such as 
	``\# of null values in each table'', 
	``\# of tuples in each table'' are easy to tweak. 
	To avoid presenting a strawman test, 
	we skip such simple features in this paper.
	Nevertheless, 
	these simple feature tweaking tools are already implemented in ASPECT.
	The user can apply such simple features if needed.
	
	To the best of our knowledge, 
	we are the first to publish these 3 features.
	Moreover, we are the first to preserve these 3 features within one dataset.
	
	In this section, we make the following assumption to avoid confusion.
	If $T'$ references $T$ (denoted $T' \rightarrow T$), 
	it does so via one foreign key constraint only. 
	This assumption can be easily relaxed.
	Also, $t_1 \rightarrow t_2, t_1 \in T_1, t_2 \in T_2$ 
	means tuple $t_1$ references $t_2$.  
	
	\subsection{Linear feature}	
	\label{sec:Tlinear}

	Applications are often interested in computing $T_k\bowtie\dots\bowtie T_1$
	for some reference chain $T_k\rightarrow \dots  \rightarrow T_1$.
	For example, to count the number of (distinct) movies with reviews
	that are commented on by users, 
	one may need to take the join of a reference chain
	from {\tt comments} to {\tt reviews} to {\tt movies}. 
	This is what we call a $\linear$ feature.
	
	Fig.\ref{fig:lineardemo} illustrates the concept of a $\linear$ feature.
	For any reference chain $T_k \rightarrow \dots \rightarrow T_1$, 
	the $\linear$ feature describes how one tuple $t_i\in T_i$ 
	is {\it transitively} referenced by other tuples $t_j\in T_j$, 
	for any $j>i$. 
	In Fig.\ref{fig:lineardemo}, 
	$a_2\in T_A$ is directly referenced by $b_2, b_3\in T_B$,
	and indirectly referenced by $c_1, c_2, c_3\in T_C$. 
	However, $a_2$ is not indirectly referenced by any tuple in $T_D$. 
	
	\begin{figure}
		\centering
		\includegraphics[height = 1.2in]{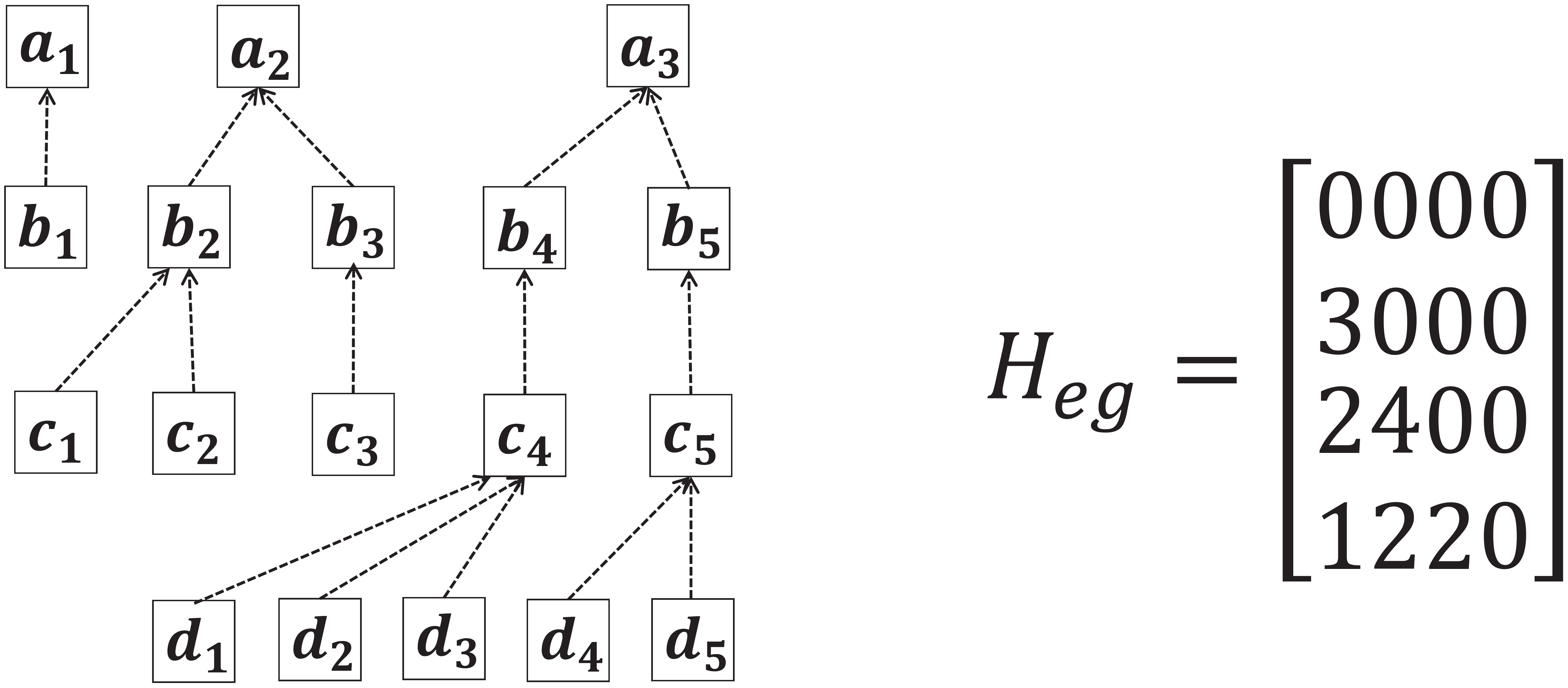}
		\caption{{\bf Linear feature}. There are $4$ tables 
			$T_D\rightarrow T_C \rightarrow T_B \rightarrow T_A$. 
			Each node in the tree is a tuple in the table.
			$a_i$ represents the $ith$ tuple in $T_A$, $b_i,c_i,d_i$ are defined similarly.
			$a_3$ is referenced by $T_B, T_C, T_D$ and
			$a_2$ is only referenced by $T_B, T_C$. 
			$H_{eg}$ is the corresponding linear join matrix.
		}
		\label{fig:lineardemo}
		\vspace{-4mm}
	\end{figure}
	
	Variants of linear joins are widely used to generate query result approximations, 
	such as database sampling~\cite{linkedSyn} and database generation~\cite{joinsyn}. 
	This paper presents $\Tlinear$, 
	an algorithm to tweak a dataset so it accurately scales the size of linear joins.
	\newline
	
	\begin{definition}	
		A tuple $t_1 \in T_1$ is a \textbf{root} of 
		$T_k  \rightarrow \dots \rightarrow T_1 $ 
		if there are tuples 
		$t_2\in T_2$, $\dots$, $t_k \in T_k$ such that 
		$t_k\rightarrow \dots \rightarrow t_1$.
		Let $S_{j,i}$ be the set of roots of
		$T_j \rightarrow \dots \rightarrow T_i$,
		and $h_{j,i}=|S_{j,i}|$.
	\end{definition}
	
	In Fig.\ref{fig:lineardemo}, $a_2$ is a root of 
	$T_C\rightarrow T_B\rightarrow T_A$,
	but not a root of  $T_D\rightarrow T_C\rightarrow T_B\rightarrow T_A$.
	If $T_1=T_A$, $T_2=T_B$, $T_3=T_C$ and $T_4=T_D$,
	then $S_{4,2}=\{b_4,b_5\}$, so $h_{4,2}=2$.
	The $h_{j,i}$ values form a matrix, as follows:
	\newline
	
	\begin{definition}
		For a maximal chain 
		$T_k \rightarrow \dots \rightarrow T_1 $, 
		define its \textbf{linear join matrix} as a lower triangular matrix 
		$$ H= \begin{bmatrix}
		0        &          &         &            & 0 \\
		h_{2,1}  & 0        &         &            &  \\
		h_{3,1}  & h_{3,2}  & \ddots  &            &  \\
		\vdots   & \vdots   & \ddots  & \ddots     &  \\
		h_{k,1}  & h_{k,2}  & \dots   & h_{k,k-1}  & 0 
		\end{bmatrix} 
		$$
		$T_k\rightarrow\dots\rightarrow  T_1$ is {\bf maximal}
		if there is no $T_{k+1}$ such that 
		$T_{k+1}\rightarrow T_k\rightarrow \dots \rightarrow T_1$ and
		$T_k\rightarrow\dots \rightarrow T_1 \rightarrow T_{k+1}$.
	\end{definition}
	
	In Fig.\ref{fig:lineardemo},  
	there are 2 roots $a_2, a_3$ for  $T_C\rightarrow T_B\rightarrow T_A$, 
	and 3 roots $a_1, a_2, a_3$ for  $T_B\rightarrow T_A$, 
	so $h_{3,1} = 2, h_{2,1} = 3$. 
	\newline
	
	Let $H$ be a linear join matrix in some $\scaled_i$ before tweaked by $\Tlinear$
	and $\tweakH$ the target linear join matrix.
	$\Tlinear$ tweaks $H$ to become $\tweakH$.
	There are two concerns: 
	\begin{enumerate}
		\item Is it possible to tweak $H$ to $\tweakH$?  (necessary conditions)
		\item How to tweak $H$ to $\tweakH$? (sufficient conditions)
	\end{enumerate}

	We first address concern 1 using the following theorem.
	\newline
	
	\begin{theorem}\text{\normalfont [{\bf necessity}]}
		Let $\scaleH$ be the linear join matrix of 
		$T_k\rightarrow \cdots \rightarrow T_1$
		before tweaked by $\Tlinear$,
		and $\tweakH$ be the target linear join matrix. 
		$\scaleH$ can be tweaked to  $\tweakH$  only if
		
		\begin{enumerate}
			\item[(L1)] $\tweakh_{j,i} \leq \min_{\substack{i \leq n \leq j}} |T_n|$
			for all $1 \leq i < j \leq k$.
			\item[(L2)] $\tweakh_{i+1,i} \geq \tweakh_{i+2,i} \geq \dots \geq \tweakh_{k,i}$
			for all $1\leq i\leq k-1$.
			\item[(L3)] $\tweakh_{j,1} \leq \tweakh_{j,2} \leq \dots \leq \tweakh_{j,j-1}$
			for all $2\leq j\leq k$.
			\item[(L4)] $\tweakh_{j,i+1} - \tweakh_{j+1,i+1} 
			\geq  \tweakh_{j,i} - \tweakh_{j+1,i}$ for all $1\leq i <j-1 < k-1$.
		\end{enumerate}
		\label{theorem:linearnec}
	\end{theorem}
	
	\begin{proof}
		\noindent
		{\it (L1)} The condition says the number of roots is not more than any table size
		along the linear join. 
		Consider the directed trees defined by the tuple references,
		like in Fig.\ref{fig:lineardemo}.
		Since each tuple has at most one parent, 
		$|T_n|$ is at least the number of roots $\tweakh_{j,i}$ for any $i\le n\le j$.
		
		\noindent
		{\it (L2)} The condition says the elements in $\tweakH$ for a column are
		non-increasing.
		It follows from observing that every path from $T_{j+1}$ to $T_i$
		contains a path from $T_j$ to $T_i$.
		
		\noindent
		{\it (L3)}  This condition says the elements in $\tweakH$ for a row
		are non-decreasing.
		It follows from observing that every path from $T_j$ to $T_i$
		contains a path from $T_j$ to $T_{i+1}$.
		
		\noindent
		{\it (L4)}  Since $S_{j+1,i}\subseteq S_{j,i}$, then
		$\tweakh_{j,i}-\tweakh_{j+1,i}=|S_{j,i}|-|S_{j+1,i}|$= $|S_{j,i}-S_{j+1,i}|$.
		Then, 
		$\tweakh_{j,i+1}-\tweakh_{j+1,i+1}=|S_{j,i+1}-S_{j+1,i+1}|$.
		Any $t\in S_{j,i}-S_{j+1,i}$ has a child $t^\prime$ that has a path
		from $T_j$ but not from $T_{j+1}$,
		so $t^\prime\in S_{j,i+1}-S_{j+1,i+1}$.
		Thus $|S_{j,i}-S_{j+1,i}|\leq |S_{j,i+1}-S_{j+1,i+1}|$ and {\it (L4)} follows.
		\newline
	\end{proof}
	
	\begin{figure}
		\centering
		\includegraphics[height=2.6in]{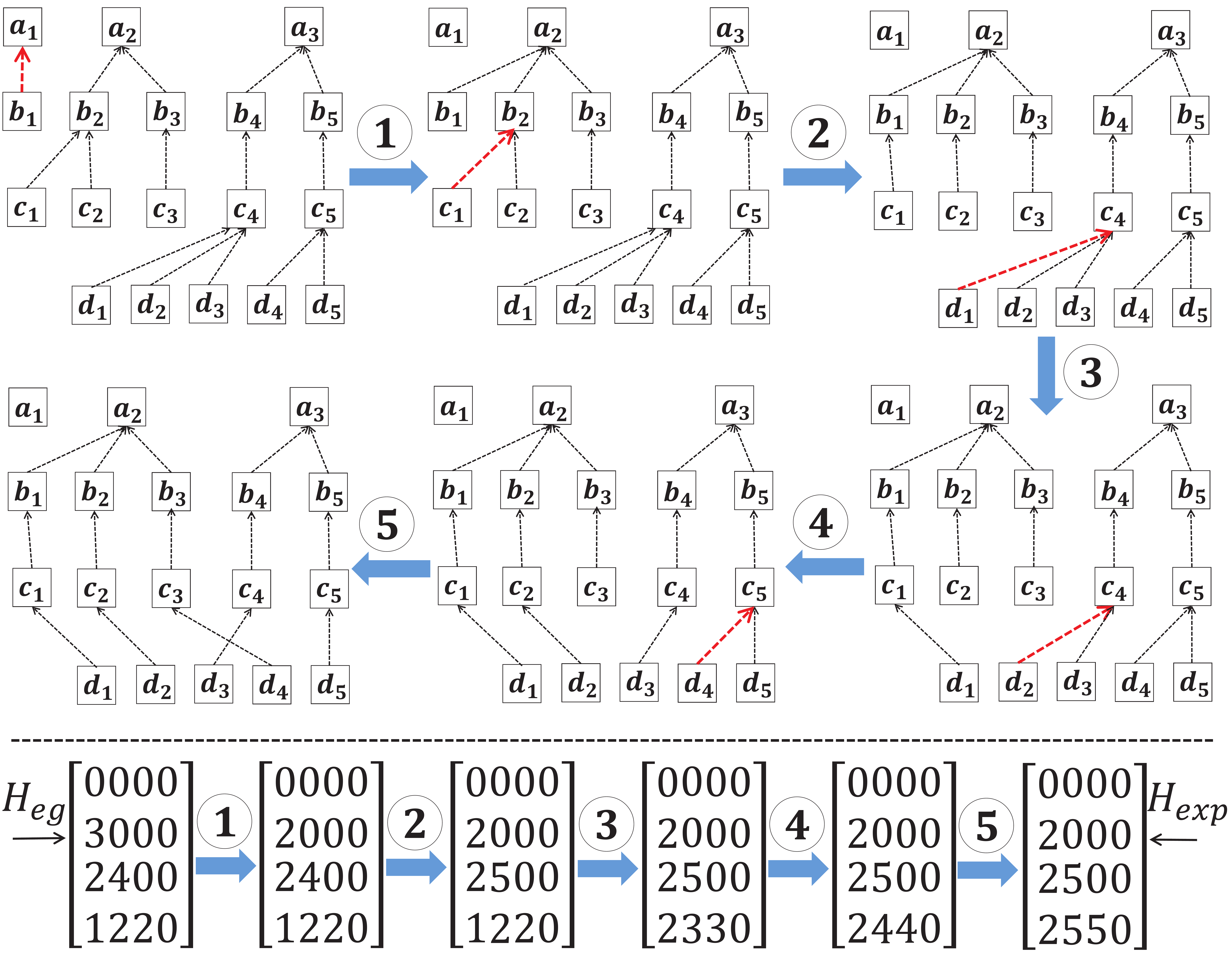}
		\caption{Tweaking demonstration for $\Tlinear$. 
			The red dotted lines are the modifications made on the dataset.
			The linear join matrix after each step is presented at the bottom.}
		\label{fig:linearExp}
		\vspace{-4mm}
	\end{figure}

	Next, we describe how $\Tlinear$ tweaks $\scaleH$ to $\tweakH$.
	$\Tlinear$ tweaks $\scaleH$  row by row.
	For $i$th row, $\Tlinear$ then tweaks the row entry by entry.
	$\Tlinear$ first does {\tt leadingAdjust}, 
	tweaking $(\scaleh_{i,1},\scaleh_{i,2},\dots, \scaleh_{i,i-1})$ to $(\tweakh_{i,1},\dots)$.
	$\Tlinear$ then does {\tt nonLeadingAdjust}:
	it tweaks $(\tweakh_{i,1},\dots)$ to $(\tweakh_{i,1},\tweakh_{i,2},\dots)$  
	$ \rightarrow (\tweakh_{i,1},\tweakh_{i,2},$
	$\tweakh_{i,3}, \dots)$  $\rightarrow \dots \rightarrow$ 
	$(\tweakh_{i,1},$ $\tweakh_{i,2},\tweakh_{i,3}, \dots,\tweakh_{i,i-1} )$. 
	Instead of providing a formal proof,
	we present an example of tweaking from 
	$H_{eg}$ to $H_{exp}$ in Fig.\ref{fig:linearExp} and 
	attach the proofs in the appendix.

	{\bf Second Row:} 
	We are expecting one less root for $T_B\rightarrow T_A$. 
	$\Tlinear$ chooses an existing root, say $a_1$, 
	and plucks all its descendants ($b_1$) and 
	attach them to some other root, say $a_2$.
	After such modification, 
	we will have 2 roots $a_2, a_3$ for $T_B\rightarrow T_A$.
	This is reflected in step 1.
	
	{\bf Third Row:} 
	No modifications are needed for the first entry. 
	For the second entry, 
	one more root for $T_C\rightarrow T_B$ is expected.
	Hence, we pluck $c_1$ from $b_2$ and attach $c_1$ to $b_1$. 
	This completes the tweaking for 
	the second row and it is reflected in step 2.
	
	{\bf Fourth Row:} 
	For the first entry, 
	we expect 1 more root for 
	$T_D\rightarrow T_C\rightarrow T_B\rightarrow T_A$, say $a_2$. 
	Hence, we pluck $d_1$ from $c_4$ and attach $d_1$ to $c_1$. 
	Now we have 2 roots $a_2, a_3$ for 
	$T_D\rightarrow T_C\rightarrow T_B\rightarrow T_A$, 
	and the last row becomes $(2,3,3,0)$.
	This is reflected in step 3.
	For the second entry, 
	we expect 2 more roots for $T_D\rightarrow T_C\rightarrow T_B$. 
	Hence, we pluck $d_2$ from $c_4$ and attach it to $c_2$. 
	This is reflected in step 4, and the last row is $(2,4,4,0)$ now.
	Lastly, 
	we pluck $d_4$ from $c_5$ and attach it to $c_3$ which ends the tweaking.
	This is reflected in step 5.

	When tweaking the $i$th row, 
	$\Tlinear$ always plucks the tuples in $i$th table 
	and attaches them to the $(i-1)$th table. 
	$\Tlinear$ never re-modify the entries in previously tweaked rows, 
	which gives some intuition that the tweaking is always possible.

	So far, we only considered tweaking one linear join matrix.
	In general, a dataset can have multiple overlapping reference chains.
	Suppose we have already tweaked the matrix for 
	$T_4 \rightarrow T_3 \rightarrow T_1$, 
	then tweak the matrix for an overlapping 
	$T_4 \rightarrow T_3 \rightarrow T_2$. 
	This can undo the tweaking for $T_4 \rightarrow T_3 \rightarrow T_1$.
	
	The issue is not just for overlapping linear joins but, in general,
	applies to any pair of tweaking algorithms.
	For example, running $\T^{\prime\prime}$ (e.g. $\Tpairwise$) 
	after $\T^\prime$ (e.g. $\Tcoappear$)
	can undo the work done by $\T^\prime$.
	We do not have a solution yet. 
	Instead, we adopt the heuristic as presented in Sec.\ref{sec:toolComponent}.
	
	\begin{figure}[t]
		\centering
		\includegraphics[height=1.1in]{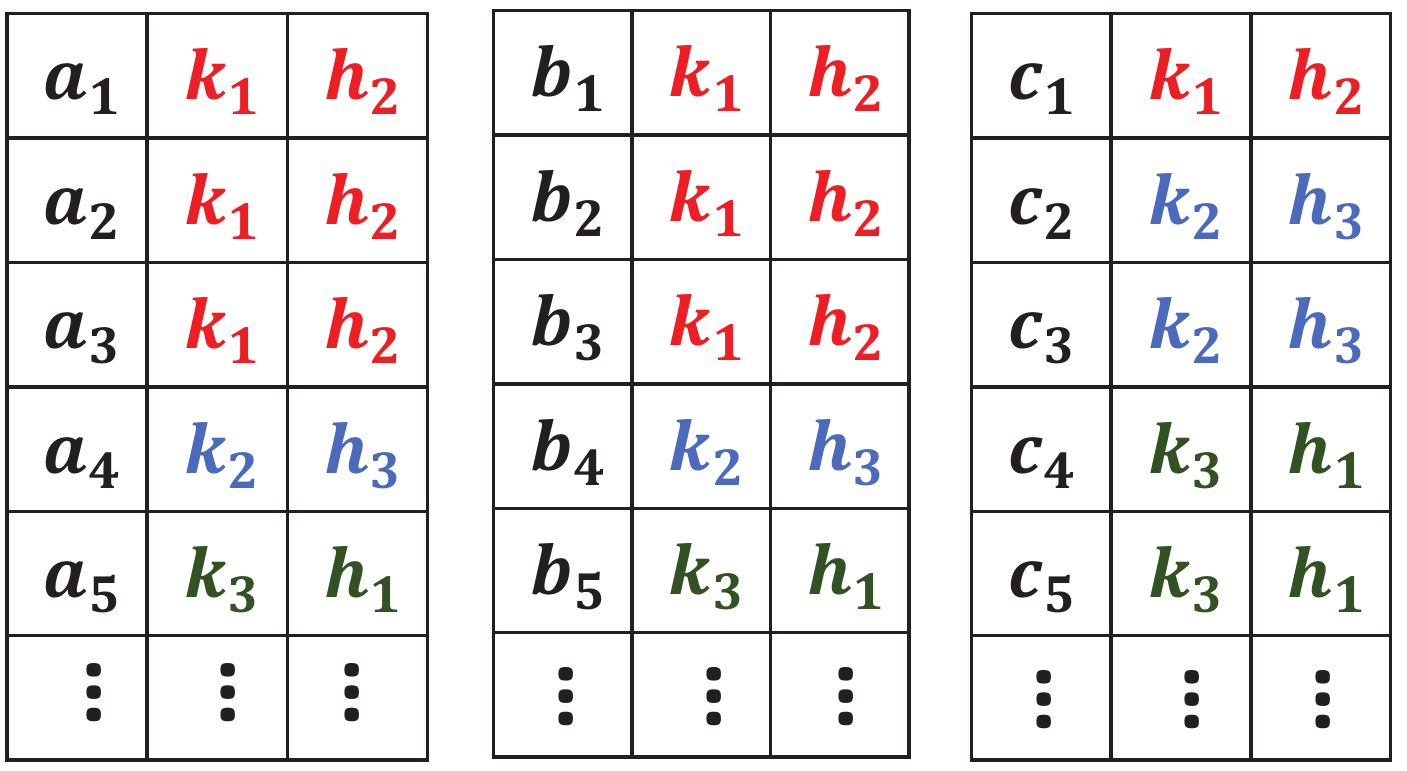}
		\caption{\textbf{Coappear feature}.
			There are 3 tables $T_A, T_B, T_C$ referencing to the same tables $T_K, T_H$. 
			$ \langle k_1, h_2 \rangle$ appeared together in $T_A$, $T_B$ for 3 times, 
			in $T_C$ for 1 time. 
		}
		\label{fig:coappeardemo}
		\vspace{-4mm}
	\end{figure}
	
	\subsection{Coappear feature}	
	\label{sec:Tcoappear}
	
	Fig.\ref{fig:coappeardemo} illustrates the concept of a $\coappear$ feature.
	The tables $T_A$, $T_B$ and $T_C$ may be for {\tt comment}, {\tt share}
	and {\tt like} in a social network service,
	referencing tables $T_K$ and $T_H$ for {\tt post} and {\tt users}.
	Thus, the same $\langle postID,userID \rangle$ may appear 
	multiple times in the same table and in multiple tables.
	This $\coappear$ feature can be used for user profiling;
	e.g. if Alice comments, shares and likes a post about volunteerism many times,
	it is more likely that Alice is interested in volunteer work~\cite{volunteer}.
	Other examples include group theme prediction~\cite{grouptheme} and 
	on-line recommendation~\cite{recommendation}. 
	
	Tweaking is done via tweaking a frequency distribution that
	captures the correlation in foreign key appearances:
	\newline
	
	\begin{definition}
		Suppose $T_1, \dots,T_k$ reference the same tables 
		$T_1^\prime$, $\dots$, $T_m^\prime$,
		and $b_1,\ldots,b_m$ coappear as foreign keys
		$v_1$ times in $T_1$, $\ldots$, $v_k$ times in $T_k$.
		If there are $n$ such $\langle b_1,\ldots, b_m\rangle$,
		then $\xi_{T_1, \dots, T_k}(v_1$, $\dots$,$v_k)=n$. 
		We call $\langle v_1, \dots,v_k\rangle$ a {\bf coappear vector}
		and $\xi_{T_1, \dots, T_k}$ the {\bf coappear distribution}.
		To simplify notation, we refer to
		$\scalexi_{T_1, \dots, T_k}$ as $\scalexi$ if there is no ambiguity.
	\end{definition}
	
	In Fig.\ref{fig:coappeardemo},  
	$\langle k_1, h_2\rangle$ appears 3 times in $T_A$, 
	3 times in $T_B$, and 1 time in $T_C$,
	so $\xi(3,3,1)=1$.
	Further,
	$\langle k_2, h_3\rangle$ and $\langle k_3,h_1\rangle$ each appears
	1 time in $T_A$, 1 time in $T_B$ and 2 times in $T_C$,
	so $\xi(1,1,2)=2$.
	\newline

	Like for linear joins, 
	we  present necessary and sufficient conditions for tweaking
	the coappear distribution:
	\newline

	\begin{theorem} \text{\normalfont [{\bf necessity}]}
		Suppose tables $T_1, \dots,T_k$ reference the same tables 
		$T_1^\prime, \dots, T_m^\prime$ in some $\scaled_i$ .
		$\scalexi$ is the coappear distribution before being tweaked by $\Tcoappear$, 
		and $\tweakxi$ is the target coappear distribution.
		$\scalexi$ can be tweaked to  $\tweakxi$
		only if:
		\label{theorem:coappearnecessary}
	\end{theorem}
	
	\begin{align}
	(C1)\phantom{XXX}& \sum_{\bf v} v_i
	\tweakxi({\bf v}) = |T_i|
	\ {\rm for\ }1\leq i\leq k \nonumber \\
	(C2)\phantom{XXX}& \sum_{\bf v}
	\tweakxi({\bf v}) = \prod_{i=1}^m |T_{i}^\prime|.
	\nonumber  
	\end{align}
	\begin{proof} 
		$(C1)$ $\tweakxi({\bf v})$ is the number of different 
		foreign key tuples $\langle b_1,\ldots, b_m\rangle$
		with coappear vector $\bf v$.
		Hence, each $\langle b_1,\ldots, b_m\rangle$ 
		appears $v_i$ times in $T_i$, so
		$ \sum_{\bf v} v_i \tweakxi({\bf v}) = |T_i|$.
		
		\noindent
		$(C2)$ $ \sum_{\bf v} \tweakxi({\bf v})$
		is the total number of different
		$\langle b_1,\ldots, b_m\rangle$ foreign key combinations.
		Since each $b_i$ is unique in $T_i^\prime$,
		the total number of combinations is  $\prod_{i=1}^m |T_{i}^\prime|$.
		\newline
	\end{proof}

	Next, we explain how $\Tcoappear$ 
	tweaks $\scalexi$ to $\tweakxi$. 
	Let $\xi^*=\scalexi-\tweakxi$, $\Delta^+=\{{\bf v}|\xi^*({\bf v})>0\}$,
	$\Delta^0=\{{\bf v}|\xi^*({\bf v})=0\}$ and 
	$\Delta^-=\{{\bf v}|\xi^*({\bf v})<0\}$.
	$\Tcoappear$  works as follows:
	
	For each ${\bf v}=(v_1,\ldots,v_k)\in \Delta^-$,
	it adds $|\xi^*({\bf v})|$ more  foreign key tuples
	$\langle b_1,\ldots, b_m\rangle$, each appearing $v_i$ times in $T_i$.
	It does this by looping $|\xi^*({\bf v})|$ times, and in each iteration:
	
	{\bf CoappearVectorRetrieve}:
	Pick the closest coappear vector
	${\bf v}^\prime=\langle v_1^\prime,\ldots,v_k^\prime\rangle \in \Delta^+$,
	using Manhattan distance.
	
	{\bf TupleRetrieve}:
	There may be multiple ${\bf b}=\langle b_1,\ldots, b_m\rangle$ foreign key
	tuples with coappear vector ${\bf v}^\prime$
	(e.g. in Fig.\ref{fig:coappeardemo}, 
	$\langle k_2,h_3\rangle$ and $\langle k_3,h_3\rangle$ 
	both have ${\bf v}^\prime=\langle 1,1,2\rangle$).
	For each ${\bf v}^\prime$,
	choose one such $\bf b$. 
	
	{\bf Tuple Modification}:
	Tuples are tweaked as follows: For  $1\leq i\leq k$,
	if $v_i^\prime - v_i>0$,
	remove $v_i^\prime-v_i$ tuples with foreign key values $\bf b$ from $T_i$ ;
	if $v_i^\prime - v_i<0$,
	add $v_i-v_i^\prime$ tuples with foreign key values $\bf b$ into $T_i$.
	
	{\bf StatsUpdate}:
	Update $\xi^*({\bf v})$ by $1$ and $\xi^*({\bf v}^\prime)$ by $-1$.
	
	The job is done when the loop terminates.
	We can prove that the necessary conditions are sufficient for the tweaking.
	\newline

	\begin{theorem}  \text{\normalfont [{\bf sufficiency}]}
		Suppose tables $T_1, \dots, T_k$ reference the same tables 
		$T_1^\prime, \dots, T_m^\prime$. 
		Let $\scalexi$ be the coappear distribution in some $\scaled_i$ before tweaking and 
		$\tweakxi$ the target coappear distribution.
		If $\tweakxi$ satisfies the necessary conditions in
		Theorem~\ref{theorem:coappearnecessary},
		then $\Tcoappear$ tweaks $\scalexi$ to become $\tweakxi$. 
		\label{theorem:Tcoappear}
		\newline
	\end{theorem}

	The formal proof is provided in the appendix.
	In a dataset, we might have multiple coappear distributions.
	Suppose $T_D$ and $T_E$ reference $T_A$ and $T_B$, 
	while $T_G$ and $T_H$ reference $T_B$ and $T_C$, 
	so there are two coappear distributions:
	$\scalexi_{T_D, T_E}$ and $\scalexi_{T_G, T_H}$. 
	$\Tcoappear$ only modifies the referencing tables;
	e.g. tweaking $\scalexi_{T_D, T_E}$ only modifies $T_D$ and $T_E$,
	without affecting $\scalexi_{T_G,T_H}$, $T_G$ and $T_H$.
	We can thus tweak coappear distributions without affecting each other.

	\subsection{Pairwise feature}
	\label{sec:Tpairwise}			
	\begin{figure}[t]
		\centering
		\includegraphics[height=1.0in]{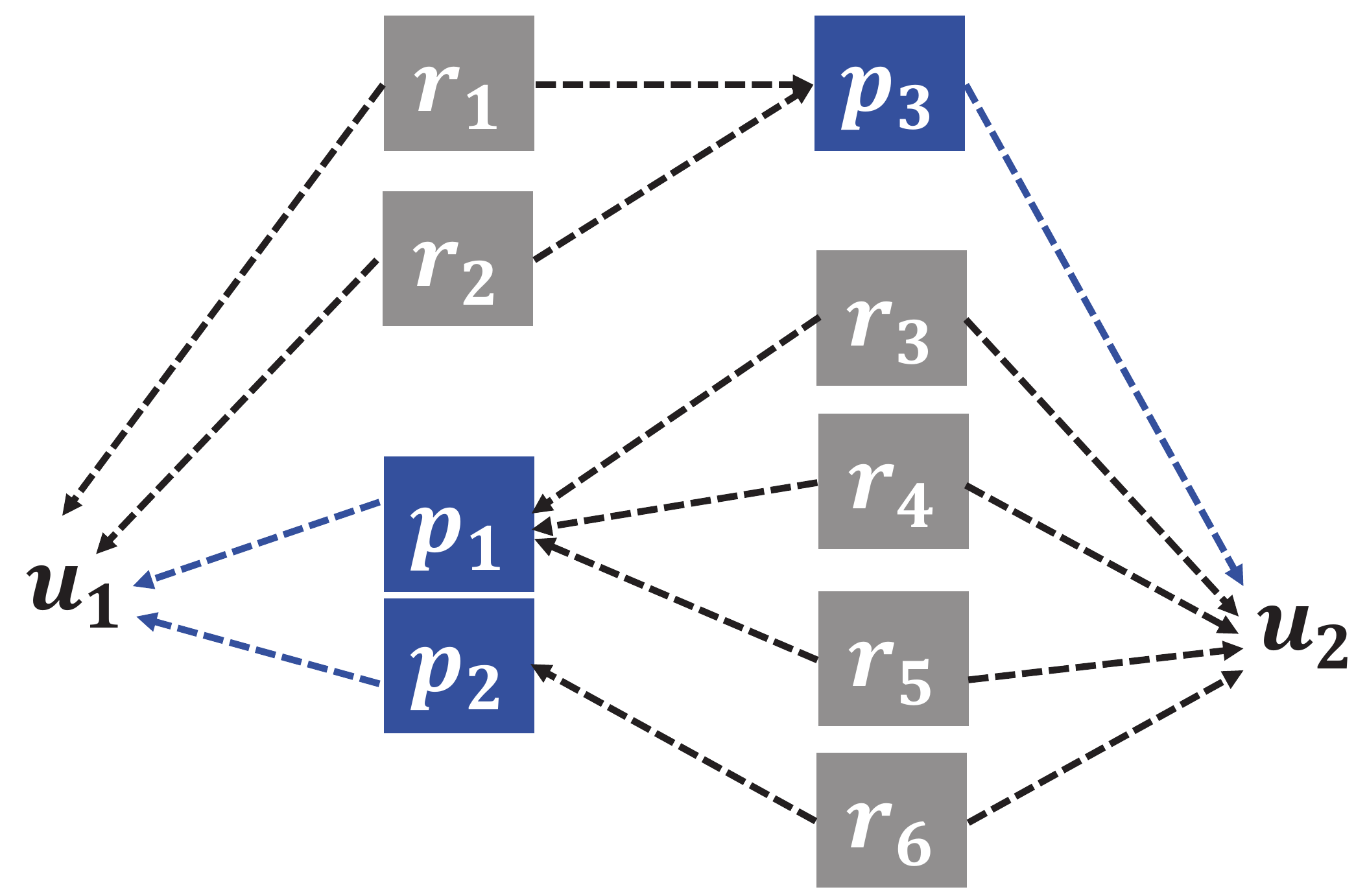}
		\vspace{-1mm}
		\caption{\textbf{Pairwise feature}.
			There are two users, $u_1$ has 2 $\post$  $p1,p2$, 
			and $u_2$ has 1 $\post$ $p_3$. 
			Moreover, $u_1$ has 2 $\response$ on $u_2$'s $\post$.
			Similarly, $u_2$ has 4 $\response$  on $u_1$'s $\post$.}
		\label{fig:pairwisedemo}
		\vspace{-5mm}
	\end{figure}
	
	Fig.\ref{fig:pairwisedemo} illustrates the concept of a $\pairwise$ feature. 
	In social networks, the social tie between users
	$u_1$ and $u_2$ may be implicit,
	instead of explicitly declared (as friends, say).
	For example, $u_1$ may respond twice $(r_1,r_2)$ to a post $p_3$ by $u_2$,
	whereas $u_2$ responds 4 times $(r_3,r_4,r_5,r_6)$ 
	to 2 posts $p_1$ and $p_2$ by $u_1$. 
	Such a feature highlighted by previous work\cite{vision} 
	involves both inter-column and inter-row correlation.
	
	For expository convenience, we use sonSchema, 
	a generic database schema for social networks~\cite{sonSchema}.
	We focus on 3 tables: {\tt users}, {\tt post}, {\tt response2post}.
	Each of these can have multiple instantiations:
	{\tt user} can represent a company, an advertiser, etc.,
	but we will only consider human users; 
	{\tt post} tables may record blogs, videos, etc. contributed by {\tt user};
	and {\tt response2post} may be a share, like, etc.
	For each type of \response \ table,
	the implicit user-to-user tie is captured by the following distribution:
	\newline
	
	\begin{definition}
		Let $R$ be a $\response$ table, and $x,y\in\{0,1,2$ $,\ldots\}$.
		Suppose there are $k$ user pairs 
		$\langle u_1,v_1\rangle$, $\ldots$, $\langle u_k,v_k\rangle$,
		where $u_i$ responds $x$ times to $v_i$'s post,
		and $v_i$ responds $y$ times to $u_i$'s post.
		We denote this as $\rho_R(x,y)=k$,
		and call $\rho_R$ the {\bf pairwise distribution}.
		\newline
	\end{definition}

	Like for $\Tlinear$, 
	we state necessary and sufficient conditions for tweaking pairwise distributions.
	To simplify the presentation,
	this section assumes a user never respond to his/her own post.
	The appendix relaxes this assumption.
	\newline

	\begin{theorem}\text{\normalfont[{\bf necessity}]}
		For a $\response$ table $R$ in some $\scaled_i$, 
		$\scalepo_R$ is the pairwise distribution before being tweaked by $\Tpairwise$,
		and $\tweakpo_R$ is the target pairwise distribution.
		$\scalepo_R$ can be tweaked to become $\tweakpo_R$ only if:
		%then $\tweakpo_R$ must satisfy the following conditions:
		\begin{align}
		(P1)&\tweakpo_R(x,y) = \tweakpo_R(y,x) {\rm\ for\ all\ }x,y \nonumber \\
		(P2)&\sum_{x,y}(x+y)\tweakpo_R(x,y) = 2|T_R|   \nonumber \\
		(P3)&\sum_{x,y} \tweakpo_R(x,y)  = |U|(|U| -1)\ 
		{\rm where\ U\ is\ the\ }{\tt user}\ {\rm table}. \nonumber
		\end{align}
		\label{theorem:pairnec}	
	\end{theorem}
	\begin{proof} 
		($P1$) $\tweakpo_R(x,y)=k$ means there are $k$ user pairs 
		$\langle u_i,v_i\rangle$, where $u_i$ responds $x$ times to $v_i$'s post
		and $v_i$ responds $y$ times to $u_i$'s post.
		This yields $x+y$ tuples in $\response$ for each $\langle u_i,v_i\rangle$.
		By symmetry, these $x+y$ tuples also represent $k$ 
		$\langle v_i,u_i\rangle$ pairs, so $\tweakpo_R(y,x)=k$.
		
		\noindent
		($P2$) As above, for each $\tweakpo_R(x,y)=k$, there are $k$ 
		$\langle u_i,v_i\rangle$ pairs, and each pair has $x+y$ tuples in $\response$.
		These $k(x+y)$ tuples are double-counted by $\tweakpo_R(x,y)$,
		so we get the equality in ($P2$).
		
		\noindent
		($P3$) Similarly, there are $|U|(|U|-1)$ user pairs,
		and each is counted once by $\tweakpo_R(x,y)$, so ($P3$) follows.
		\newline
	\end{proof}

	Next, we explain how $\Tpairwise$ tweaks $\scalepo$ to $\tweakpo$.
	Let $\rho^*_R=\scalepo_R-\tweakpo_R$,
	$\Theta^+=\{(x,y)|\rho^*_R(x,y)>0\}$ and
	$\Theta^-=\{(x,y)|\rho^*_R(x,y)<0\}$.
	$\Tpairwise$ loops through each $\response$ table $R$.
	For each $(x,y)\in\Theta^-$, it adds $|\rho^*_R(x,y)|$ pairs 
	$\langle u_i,v_i\rangle$, 
	where user $u_i/v_i$ has $x/y$ $\response$ tuples in $R$ referencing $v_i/u_i$'s post.
	It does this by looping $|\rho^*_R(x,y)|$ times, and in each iteration:
	
	{\bf PairwiseVectorRetrieve}: 
	Pick $v^\prime=(x^\prime,y^\prime)\in\Theta^+$ that is closest to $(x,y)$
	by Manhattan distance.
	
	{\bf TupleModification}:
	Choose users $u_i$ and $v_i$ with pairwise vector $(x^\prime,y^\prime)$
	and tweak $u_i$'s responses to $v_i$'s post, as follows:
	If $x<x^\prime$, 
	then
	$u_i$ has $x^\prime-x$ more responses to $v_i$'s post than desired,
	so $\Tpairwise$ randomly chooses and removes $x^\prime-x$ such responses.
	If $x>x^\prime$,
	we add $x-x^\prime$ responses from $u_i$ on $v_i$'s post.
	If $v_i$ has no post, we artificially create a post for $v_i$.
	To do this, we pick another user $w_i$ who has more than 1 post and 
	pick a post $p_w$ with minimum responses among $w_i$'s posts;
	we make $p_w$ a post by $v_i$, 
	and shift the responses to $p_w$ to other posts by $w_i$.
	If (rare case) all users have at most 1 post,
	we will make a new post $p$ for $v_i$, 
	and add $x-x^\prime$ responses to $p$.
	We similarly tweak $v_i$'s responses to $u_i$'s post.
	
	{\bf StatsUpdate}: 
	Increase $\rho^*_R(x,y)$ and $\rho^*_R(y,x)$ by 1 and 
	decrease $\rho^*_R(x^\prime,y^\prime)$ and $\rho^*_R(y^\prime,x^\prime)$ by 1.
	
	We can prove the above mentioned conditions in 
	Theorem~\ref{theorem:pairnec} are sufficient for 
	$\Tpairwise$ tweaks $\scalepo_R$ to $\tweakpo_R$ by Theorem~\ref{theorem:pairsufficient}.
	\newline
	
	\begin{theorem} \text{\normalfont [{\bf sufficiency}]}
		For each $\response$ table $R$  in some $\scaled_i$, 
		$\scalepo_R$ is the pairwise distribution before tweaking and 
		$\tweakpo_R$ is the target pairwise distribution.
		If $\tweakpo_R$ satisfies the necessary conditions in Theorem~\ref{theorem:pairnec},
		then $\Tpairwise$ tweaks $\scalepo_R$ to $\tweakpo_R$.
		Moreover, the extra tuples added to the {\tt post} table $P$ is at most
		$|U|-|P|$, where $U$ is the {\tt user} table. \newline
		\label{theorem:pairsufficient}
	\end{theorem}

	The formal proof is provided in the appendix.
	Since $\response$ can have several instantiations (e.g. share, like, etc.),
	a social network dataset can have multiple pairwise distributions,
	but they can be tweaked independently.
	For example, suppose a {\tt post} table $P$ 
	has two $\response$ tables $R_1$ and $R_2$,
	and  $\scalepo_{R_1}$ is tweaked to $\tweakpo_{R_1}$ first.
	When tweaking $\scalepo_{R_2}$, we only modify the tuples in $R_2$,
	so it does not affect the tweaked $\tweakpo_{R_1}$.
	Moreover, adding tuples in $P$ does not affect $\tweakpo_{R_1}$ as well.
	
	For the above mentioned tweaking tools, 
	they modify the dataset by calling the functions in Sec.\ref{sec:protocol}.
	$\Tlinear$ modifies the dataset through the operation {\tt replaceValues}.
	$\Tcoappear$ and $\Tpairwise$ modify the dataset 
	through the operations {\tt deleteValues}, {\tt insertValues}.

	\section{Experiment Setup}
	\label{sec:experiment}
	
	In the following, each experiment is run on a Linux machine with
	64GB memory and an Intel Xeon 2.4GHz processor. 
	We now describe the datasets and similarity measures used in our experiments. 
	ASPECT is implemented in Java.
	
	\subsection{Datasets}
	\label{sec:dataset}
	In this paper, 
	due to the space constraint, 
	we only present experiments on $\Xiami$\footnote{https://www.xiami.com}.  
	$\Xiami$ contains music-related data with 28 tables and more than 90M tuples.
	Reader can refer to the appendix for experiments 
	on three more datasets, $\DoubanBook$, $\DoubanMusic$ and $\DoubanMovie$.
	Each dataset is larger than 10GB originally. 
	However, there are columns, e.g. {\tt song\_name}, {\tt movie\_name},  
	that are irrelevant to the experiments.
	We do not want to exaggerate ASPECT's capability of handling big datasets. 
	We hence purposely filter out those irrelevant columns 
	and only conduct experiment on the relevant columns.
	
	\label{sec:partition}
	
	We take $6$ snapshots of each dataset, 
	$\oldd_1 \subset \oldd_2 \subset \oldd_3 \subset \oldd_4 \subset \oldd_5 \subset \oldd_6$. 
	For each $\oldd_i$, ASPECT takes $\oldd_1$ as input, 
	and first uses a size-scaler to scale $\oldd_1$ to $\scaled_0$,
	where $\scaled_0$ and $\oldd_i$ are of the same size. 
	We then apply tweaking tools on $\scaled_0$ to achieve the target features. 
	After the tweaking process is done, 
	ASPECT outputs $\scaled_i$ which is similar to $\oldd_i$.
	For our experiments, we use $\oldd_i$ as the ground-truth, 
	and compare the similarity between $\scaled_i$ and $\oldd_i$.

	\subsection{Size-scaler}
	\label{sec:scalers}
	Size-scalers are orthogonal to enforcing features in the final dataset $\scaled$.
	Our experiments show that ASPECT is able to generate datasets 
	with small errors for three different size-scalers,
	$\dscaler$~\cite{dscaler}, $\rex$~\cite{rex} and Rand, described below:
	
	\textbf{$\dscaler$} is the first solution to scale relational tables by different ratios. 
	It uses a \textit{correlation database} which captures fine-grained, 
	per-tuple correlations to scale the original dataset. 
	
	{\bf ReX } is an automated representative extrapolation technique~\cite{rex}.
	It scales all tables by the same ratio. 
	Since tables in the ground truth dataset do not scale uniformly, 
	the targeted features do not satisfy the  
	necessary conditions in Sec.\ref{sec:featuretweaker}
	for datasets generated by $\rex$.
	We, therefore, modify the targeted features  to enforce the necessary 
	conditions before tweaking.
	%We use the original Java implementation\footnote{https://github.com/tbuda/ReX}.  
	
	\textbf{Rand} is a randomised size-scaler. 
	The tuples are generated randomly.
	However, it satisfies two requirements: 
	(i) the number of tuples are generated as expected and
	(ii) the tuples generated satisfy the foreign key constraints. 
	%Rand is implemented in Java.

	\subsection{Similarity measure}
	\label{sec:measure}
	As stated previously, the similarity between 
	tweaked dataset $\scaled$ and ground truth dataset $\oldd$ 
	are defined through the feature set $\F$. 
	Hence, we measure how well ASPECT preserves the features.
	In the experiment, we only apply $3$ tweaking tools presented in Sec.\ref{sec:featuretweaker} 
	to preserve the corresponding features.
	Hence, we measure similarity based on these $3$ features.
	
	\subsubsection{\textbf{Feature Accuracy}} 
	we individually measure the similarity of the 3 features.
	
	\textit{Linear Feature}:
	For a target linear join matrix $H$ (ground truth) and the corresponding $\tweakH$
	in the final tweaked dataset, 
	let $\epsilon_H$ be the {\it mean} relative error among the entries.
	For example, 
	$$ 
	{\tweakH}= \begin{bmatrix}
	0	& 0	& 0 \\
	5	& 0 & 0  \\
	2	& 3	& 0  
	\end{bmatrix}
	\ \ \ \ \ \ \
	H= \begin{bmatrix}
	0	& 0	& 0 \\
	4	& 0 & 0  \\
	3	& 4	& 0  
	\end{bmatrix}
	$$
	then $\epsilon_H= \frac{1}{3}(\frac{|5-4|}{4} + \frac{|2-3|}{3} + \frac{|3-4|}{4}) = \frac{5}{18}$.
	The $\linear$ feature error of $\tweakd$ is the mean of all $\epsilon_H$, 
	thus unbounded.
	
	$\textit{Coappear Feature}$:
	For each coappear distribution,
	let $\xi$ be the target (ground truth) and 
	$\tweakxi$ the tweaked distribution.  
	The coappear distribution error $\epsilon_{\xi}$ is 
	$$\epsilon_\xi = \frac{1}{N_{\rm FK}} \sum_{\bf v} | {\xi}({\bf v}) - \tweakxi({\bf v})|,$$
	where $N_{\rm FK}$ is the number of foreign key vectors.
	$\epsilon_\xi$ is bounded by 
	$\frac{1}{N_{\rm FK}} (\sum_{\bf v} |{ \xi}({\bf v})| + |\tweakxi({\bf v})|)=2.$
	The $\coappear$ feature error of $\tweakd$ is the mean of all $\epsilon_{\xi}$.
	
	$\textit{Pairwise Feature}$:
	Similarly, for a pairwise distribution, 
	let $\rho$ be the target (ground truth) and $\tweakpo$ the tweaked distribution.
	The pairwise distribution error $\epsilon_\rho$ is
	$$\epsilon_\rho = \frac{1}{N_{\rm user-pair}} \sum_{\bf v} | {\rho}({\bf v}) - \tweakpo({\bf v})|,$$
	where $N_{\rm user-pair}$ is the number of user pairs;
	$\epsilon_\rho$ is at most 2.
	The pairwise distribution error of $\tweakd$ is the mean of all $\epsilon_\rho$.
	
	\subsubsection{\textbf
		{Query Accuracy}}
	we also measure similarity by the result of an aggregate queries 
	(\texttt{COUNT, AVERAGE}) that are related to the 3 features. 
	The query error is measured by $e_q = \frac{|q(\scaled) - q(\oldd)|}{q(\oldd)}$.

	\section{Results and Analysis}
	\label{sec:results}
	In our experiment, 
	ASPECT coordinates $\Tlinear$, $\Tcoappear$ and $\Tpairwise$
	on the scaled dataset generated by a size-scaler to realize the corresponding features.
	There are $3!=6$ ways of ordering these tweaking tools.
	We use P-L-C, say, 
	to denote the permutation where $\Tpairwise$, $\Tlinear$ and $\Tcoappear$ 
	are applied in that order.
	
	We first compare the feature similarity in Sec.\ref{sec:featureSimExp}, 
	followed by query similarity experiments in Sec.\ref{sec:querySimExp}.
	Later, we discuss the possible improvements in Sec.\ref{sec:querySimExp}.
	Lastly, we present the execution time of ASPECT in Sec.\ref{sec:time}.

	\subsection{Feature similarity}
	\label{sec:featureSimExp}
	
	For each feature, the plots are organized as follows: 
	x-axis represents the dataset snapshots;
	y-axis is the feature error.
	In each plot, we compare how the 6 permutations perform 
	against the baseline (without tweaking).

	\subsubsection{\textbf{Linear feature}}
	\label{sec:linearresult}
	In $\Xiami$, there are in total 38 linear join matrices.
	Fig.\ref{fig:linear} plots the average error of these linear join matrices.
	In general, the later $\Tlinear$ is applied,
	the smaller the $\linear$ feature error,
	i.e. C-L-P and P-L-C have smaller errors 
	than L-C-P and L-P-C, and C-P-L and P-C-L have 0 error.
	All permutations reduce the error tremendously for all size-scalers on all datasets.
	
	Different size-scalers generate a scaled dataset with different errors. 
	In Fig.\ref{fig:linear}, take $\oldd_6$ for example, 
	$\dscaler$ generates a dataset with error around 1.0, 
	while $\ReX$ generates a dataset with error around 25. 
	Regardless of the initial error difference, 
	ASPECT is able to reduce the error tremendously after applying the tweaking tools.
	
	\begin{figure}[t]
		\begin{minipage}[t]{0.32\linewidth}
			\centering
			\includegraphics[height = 1.1in]{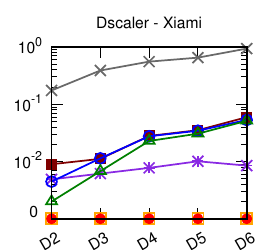}
		\end{minipage}
		\hfill
		\begin{minipage}[t]{0.32\linewidth}
			\centering
			\includegraphics[height = 1.1in]{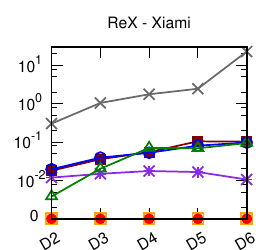}
		\end{minipage}
		\hfill
		\begin{minipage}[t]{0.32\linewidth}
			\centering
			\includegraphics[height = 1.1in]{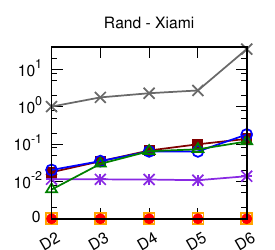}
		\end{minipage}
		\hfill
		\begin{minipage}[t]{0.99\linewidth}
			\centering
			\includegraphics[height = 0.085in]{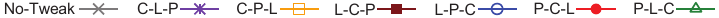}
		\end{minipage}
		\caption{Linear feature errors  for $\Xiami$ (log scale)}
		\label{fig:linear}
		\vspace{-4mm}
	\end{figure}
	
	\begin{figure}[t]
		\hfill\begin{minipage}[t]{0.32\linewidth}
			\centering
			\includegraphics[height = 1.1in]{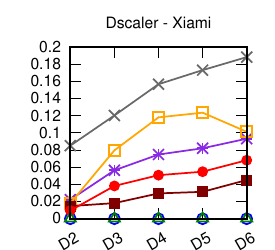}
		\end{minipage}
		\hfill
		\begin{minipage}[t]{0.32\linewidth}
			\centering
			\includegraphics[height = 1.1in]{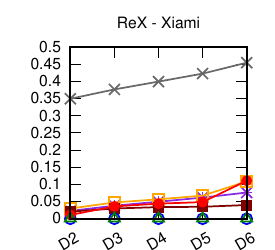}
		\end{minipage}
		\hfill
		\begin{minipage}[t]{0.32\linewidth}
			\centering
			\includegraphics[height = 1.1in]{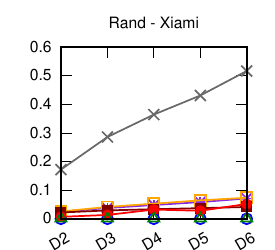}
		\end{minipage}
		\hfill
		\begin{minipage}[t]{0.99\linewidth}
			\centering
			\includegraphics[height = 0.085in]{newgnuplot/seed/xiami/legend-eps-converted-to.pdf}
		\end{minipage}
		\caption{Coappear feature errors for $\Xiami$ }
		\label{fig:coappear}
		\vspace{-4mm}
	\end{figure}
	
	\subsubsection{\textbf{Coappear feature}}
	\label{sec:coappearresult}
	There are 12 coappear distributions for $\Xiami$. 
	Fig.\ref{fig:coappear} plots the average error of these coappear distributions.
	It shows that, like for $\Tlinear$,
	the later $\Tcoappear$ is applied in the tweaking order,
	the smaller the $\coappear$ error.
	In general, we find that permutations where $\Tcoappear$ is after
	$\Tlinear$ reduces the errors more than if $\Tcoappear$ is before $\Tlinear$.
	This is expected, since $\Tlinear$ modifies the coappearing tables
	massively after $\Tcoappear$ is done.
	
	Similar to $\linear$ feature, 
	most tweaking permutations significantly reduce the $\coappear$ errors. 
	However, for the plot {\tt Dscaler-Xiami}, 
	we observe that the tweaking permutation C-L-P and C-P-L have a smaller error reduction.
	One possible reason may be small original error, 
	that gives limited room for improvement.
	The other possible reason could be the highly overlapping structure: 
	the coappear distribution involves many tables. 
	Take $\xi_{\bf T}$ for example, where 
	$\bf T$ is $\langle${\tt Listen\_Artist}, {\tt Lib\_Artist}, 
	{\tt Artist\_Fan}, {\tt Artist\_Comment}$\rangle$. 
	This $\xi_{\bf T}$ overlaps with 8 linear joins,
	so it is modified by 8 linear tweaking tools 
	if $\Tlinear$ is applied after $\Tcoappear$.
	This increases the difficulty of 
	getting a validated modification as described in Sec.\ref{sec:toolComponent}. 
	We will discuss how to improve the similarity for such highly overlapping features in Sec.\ref{sec:improvement}.
	
	Nevertheless,  {\tt ReX-Xiami} and {\tt Rand-Xiami} still have small errors 
	despite such a highly overlapping features.

	\begin{figure}[t]
		\begin{minipage}[t]{0.32\linewidth}
			\centering
			\includegraphics[height = 1.1in]{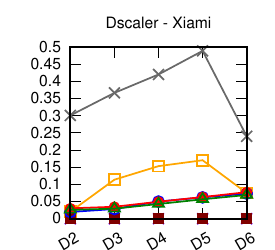}
		\end{minipage}
		\hfill\begin{minipage}[t]{0.32\linewidth}
			\centering
			\includegraphics[height = 1.1in]{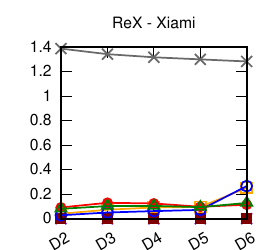}
		\end{minipage}	
		\hfill
		\begin{minipage}[t]{0.32\linewidth}
			\centering
			\includegraphics[height = 1.1in]{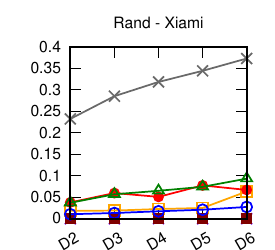}
		\end{minipage}
		\hfill
		\begin{minipage}[t]{0.99\linewidth}
			\centering
			\includegraphics[height = 0.085in]{newgnuplot/seed/xiami/legend-eps-converted-to.pdf}
		\end{minipage}
		\caption{Pairwise feature errors for $\Xiami$ }
		\label{fig:pairwise}
		\vspace{-4mm}
	\end{figure}

	\subsubsection{\textbf{Pairwise feature}}
	\label{sec:pairwiseresult}
	$\Xiami$ has 4 pairwise distributions. 
	Fig.\ref{fig:pairwise} plots the error of these pairwise distributions.
	It again shows that,
	the later $\Tpairwise$ is applied in a tweaking order,
	the smaller the pairwise feature error in the tweaked dataset.
	Moreover, 
	all tweaking permutations reduce the errors tremendously for all size-scalers.	
	
	In summary, the later a tool $\T_i$ is applied, 
	the smaller the error for the feature $\F_i$.
	Moreover, all tweaking permutations reduce the errors tremendously for most of the cases.
	If the features are highly overlapping, 
	it is possible that the error reduction is not very significant. 
	In the next section, we will discuss how to improve this.
	
	\subsection{Query similarity}
	\label{sec:querySimExp}
	\begin{figure}[b]
		\begin{minipage}[t]{0.24\linewidth}
			\centering
			\includegraphics[height = 0.95in]{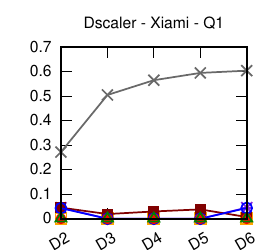}
		\end{minipage}
		\hfill
		\begin{minipage}[t]{0.24\linewidth}
			\centering
			\includegraphics[height = 0.95in]{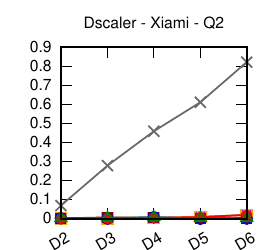}
		\end{minipage}
		\hfill
		\begin{minipage}[t]{0.24\linewidth}
			\centering
			\includegraphics[height = 0.95in]{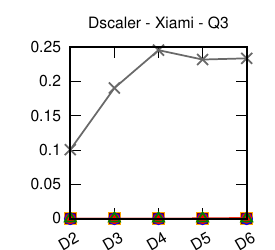}	
		\end{minipage}
		\hfill
		\begin{minipage}[t]{0.24\linewidth}
			\centering
			\includegraphics[height = 0.95in]{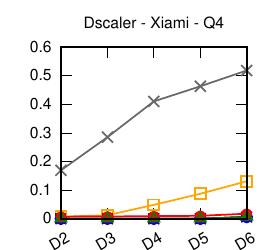}
		\end{minipage}
		\hfill
		\begin{minipage}[t]{0.24\linewidth}
			\centering
			\includegraphics[height = 0.95in]{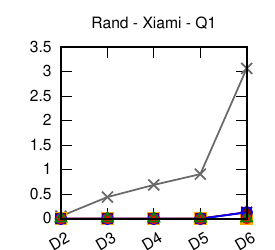}
		\end{minipage}
		\hfill
		\begin{minipage}[t]{0.24\linewidth}
			\centering
			\includegraphics[height = 0.95in]{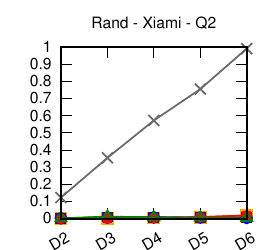}
		\end{minipage}
		\hfill
		\begin{minipage}[t]{0.24\linewidth}
			\centering
			\includegraphics[height = 0.95in]{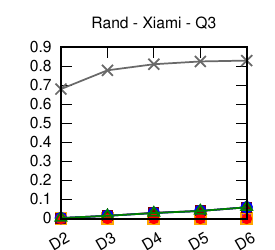}	
		\end{minipage}
		\hfill
		\begin{minipage}[t]{0.24\linewidth}
			\centering
			\includegraphics[height = 0.95in]{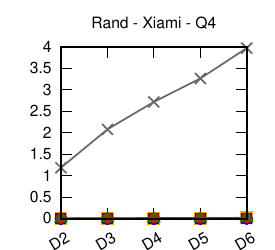}
		\end{minipage}
		\hfill
		\begin{minipage}[t]{0.99\linewidth}
			\centering
			\includegraphics[height = 0.085in]{newgnuplot/seed/xiami/legend-eps-converted-to.pdf}
		\end{minipage}
		\caption{Query similarity for $\Xiami$}
		\label{fig:queryExp}
		\vspace{-4mm}
	\end{figure}

	For each query  $q$, we compare the query results on the 
	ground-truth dataset $q(\oldd)$ and the scaled dataset $q(\scaled)$.
	As mentioned in Sec.\ref{sec:scalers}, 
	$\ReX$ cannot scale the dataset to arbitrary sizes, 
	so it cannot be used for query similarity experiments.
	The $4$ queries used are: 
	$Q_1$ computes the number of users who have uploaded a photo with commenters;
	$Q_2$ computes the number of Music Videos that have been commented on by at most 10 different users;
	$Q_3$ computes the average number of listeners per song;
	$Q_4$ computes the number of user pairs having interactions through profile page.
	
	Fig.\ref{fig:queryExp} presents results for the $4$ queries. 
	The first row uses $\dscaler$ as a size-scaler, 
	the second row uses Rand as a size-scaler.
	The x-axis represents the dataset snapshots, 
	and y-axis represents query error.
	
	As we can see from Fig.\ref{fig:queryExp}, 
	all tweaking permutations reduce the query error significantly on both size-scalers.
	The errors are reduced to $ < 0.05$ for most of the tweaking permutations.
	For $Q_2$, even though the initial error after 
	the size-scaler is relatively low for $\oldd_2$, 
	ASPECT is still able to reduce the error further.
	
	\subsection{Similarity improvement over iterations}
	\label{sec:improvement}
	Even though ASPECT significantly reduces the errors for most of the cases, 
	there are some rare exceptions, e.g. {\tt Dscaler-Xiami} in Fig.\ref{fig:coappear}, 
	where ASPECT generates a dataset with $\coappear$ error > 0.1. 
	Such cases happen when tools modify previously tweaked features.
	To improve the performance, we run ASPECT for multiple iterations. 
	Previously, we applied tools $\Tcoappear, \Tlinear, \Tpairwise$ 
	sequentially for the permutation C-L-P which results in $\scaled$.
	But now, 
	we apply tools $ \Tcoappear, \Tlinear, \Tpairwise$ on $\scaled$ 
	in the same order with another few iterations. 
	We find that by having more iterations of tweaking, 
	the error is further reduced tremendously.
	
	\begin{figure}
		\centering
		\begin{minipage}[t]{0.4\linewidth}
			\centering
			\includegraphics[height = 1.1in]{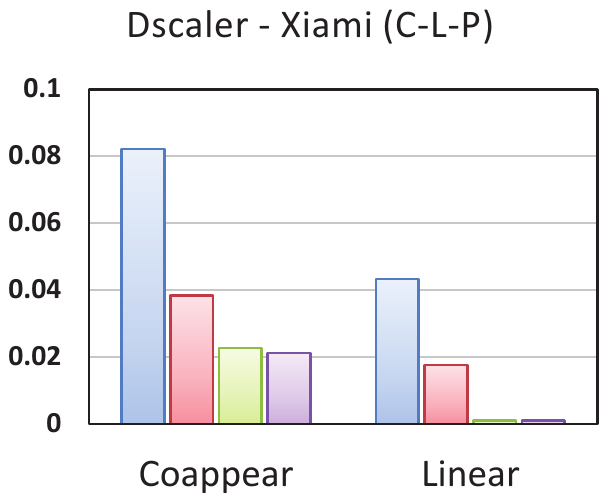}
		\end{minipage}
		\hfill
		\begin{minipage}[t]{0.4\linewidth}
			\centering
			\includegraphics[height = 1.1in]{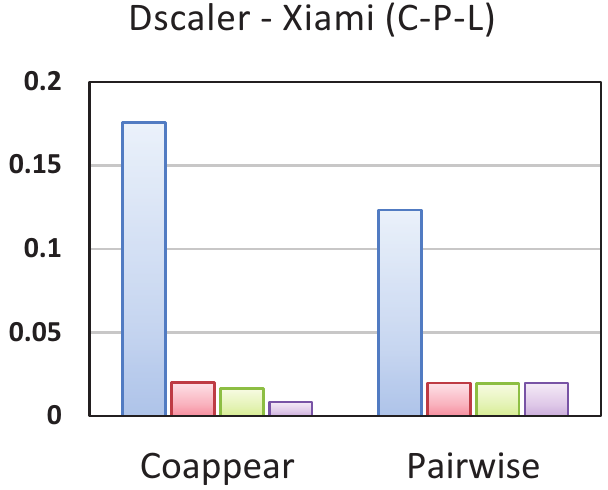}
		\end{minipage}
		\hfill
		\begin{minipage}[t]{0.15\linewidth}
			\includegraphics[height = 0.9in]{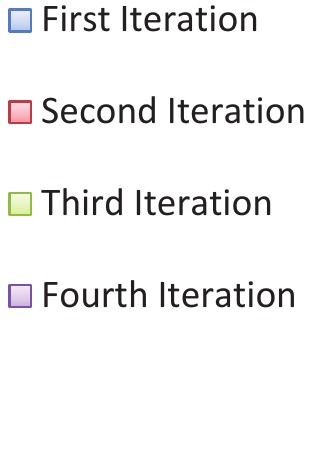}
		\end{minipage}
		\vspace{-2mm}
		\caption{Tweaking error over iterations (vertical axis is error)}
		\label{fig:iterationDemo}
	\end{figure}

	In {\tt Dscaler-Xiami}, C-L-P, C-P-L have larger errors.
	Fig.\ref{fig:iterationDemo} presents the results of C-L-P and C-P-L with more iterations. 
	The x-axis is the features; 
	the y-axis represents the errors; 
	the bar represents the iterations.
	
	For C-L-P, the $\coappear$ error is $0.08$ for the first iteration, 
	and reduced to $0.04$ in the second iteration and 
	further reduced to $0.02$ in the third iteration.
	For $\linear$ feature, the error reduction is greater, 
	from $0.05$ to $10^{-3}$ from third iteration onwards.
	For C-P-L, we observe similar phenomenon. 
	Moreover, the error reduction is faster (the error stabilises from second iteration onwards).
	
	In summary, 
	the errors are reduced significantly 
	as the number of iterations increases. 
	From the second or third iteration onwards, 
	the resulting error will be really small $\sim$0.02. 
	Hence, the room for improvement will be limited.  
	The reader can find significant error reduction for other datasets 
	in the appendix as well .

	\subsection{ASPECT execution time}
	\label{sec:time}
	So far, we have verified that ASPECT is effective in tweaking the features. 
	Next, we will show that ASPECT is efficient as well.  
	Fig.\ref{fig:runTime} presents the running time of each tweaking permutation. 
	Similar to the previous plots, 
	the x-axis represents the dataset snapshot; 
	y-axis represents the running time (minutes).
	
	In Fig.\ref{fig:runTime}, 
	the execution time increases linearly 
	with the dataset size for most of the experiments.
	All tweaking permutations finish within 100 minutes.
	Moreover, 
	different size-scalers result in different execution time.
	This is expected, 
	as different size-scalers have different feature errors.
	Hence, the amount of tweaking is different. 
	
	For the same size-scalers and the same dataset, 
	different tweaking permutations have different execution times. 
	In general, L-C-P and L-P-C 
	are more efficient than other tweaking permutations.
	
	\begin{figure}
		\begin{minipage}[t]{0.32\linewidth}
			\centering
			\includegraphics[height = 1.1in]{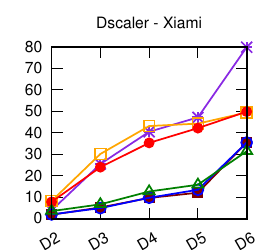}
		\end{minipage}
		\hfill
		\begin{minipage}[t]{0.32\linewidth}
			\centering
			\includegraphics[height = 1.1in]{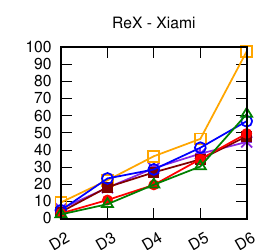}
		\end{minipage}
		\hfill
		\begin{minipage}[t]{0.32\linewidth}
			\centering
			\includegraphics[height = 1.1in]{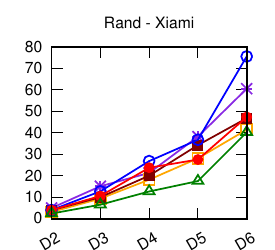}	
		\end{minipage}
		\hfill
		\begin{minipage}[t]{0.99\linewidth}
			\centering
			\includegraphics[height = 0.085in]{newgnuplot/seed/xiami/legend-eps-converted-to.pdf}
		\end{minipage}
		\caption{Execution time for $\Xiami$  (minutes)}
		\label{fig:runTime}
	\end{figure}	
	
	\section{Limitations And Observations}
	\label{sec:limit}
	ASPECT aims to tweak the features 
	$\oldf_1, \oldf_2,\dots, \oldf_n$
	so that the tweaked dataset has the corresponding target features 
	$\scalef_1$, $\scalef_2$, $\dots$, $\scalef_n$.
	Extensive experiments above show that ASPECT
	has the capability of tweaking complex features with 
	reasonably small errors within reasonable running time.
	
	\subsection{Limitations}
	While tweaking a feature, we might modify some already tweaked features.
	These already tweaked features may take various forms,
	which significantly increases the difficulty of 
	proving some error bound of previously tweaked features. 
	We state this issue as the \textbf{Feature Tweaking Bound Problem}:
	
	\begin{quote}
		\textit{Assuming a dataset $\scaled_n$ has features $\scalef_1, \dots, \scalef_n$.
			If a tweaking tool $\T_{n+1}$ is applied on $\scaled_n$, 
			how much does it affect the previous features $\scalef_1, \dots ,\scalef_n$? }
	\end{quote}
	
	Solving the Feature Tweaking Bound Problem for general features might not be possible. 
	Proving error bounds should be easier if the features satisfy certain properties.
	Consider the following trivial example: 
	if each $\scalef_i$ represents the attribute distribution of a distinct column, 
	then one can easily see that tweaking tool $\T_{n+1}$ 
	will never affect $\scalef_1, \dots ,\scalef_n$.
	An example of a non-trivial restriction would be limiting the features to just 1 join.
	
	Besides the Feature Tweaking Bound Problem, 
	there is also the \textbf{Feature Tweaking Order Problem}: 
	
	\begin{quote}
		\textit{When tweaking a dataset for $n$ features, 
			which tweaking order results in the least error?}
	\end{quote}
	
	We believe the \textit{Feature Tweaking Bound Problem} 
	and \textit{Feature Tweaking Order Problem} are issues that offer
	a rewarding challenge for research on dataset tweaking.  
	
	\subsection{Observations}
	In developing ASPECT, we arrive at the following observations:
	
	(\textbf{O1}) \textbf{Non-overlapping features}. 
	If the features are not overlapping, 
	then regardless of the modification of $\T_{n+1}$, 
	all the previous modified features will be preserved. 
	
	(\textbf{O2}) \textbf{Determination of non-overlapping features}. 
	Given (\textbf{O1}), 
	we would want to determine which features among 
	$\scalef_1, \dots ,\scalef_n$ do not overlap. 
	Then, the user will clearly know which features 
	do not affect each other.
	This can be achieved in ASPECT through monitoring 
	the dataset access by each tweaker.
	In ASPECT, 
	each tweaker can only access the dataset via 
	the functions similar to the ones provided in Fig.\ref{fig:interface}.
	Hence, 
	ASPECT knows if any two tweaking tools have accessed the same tuples. 
	Then the problem is reduced to finding independent sets in graph theory, 
	where the nodes are the tweaking tools. 
	If two tweaking tools access the same tuples, 
	then there will be an edge linking the two nodes. 
	Even though, it is NP-Hard to find a maximum independent set, 
	Robson~\cite{robson} has proven that it can done in $O(1.22^n)$ time, 
	which is a reasonable complexity for a small number of tools $n$.

	(\textbf{O3}) 
	\textbf{Conflicting overlapping features}. 
	Overlapping features are called conflicting if no dataset can satisfy all of them. 
	A simple example of conflicting
	features of a social network dataset is: 
	$F_1$ = {more than half of the customers are men} 
	and 
	$F_2$ = {more than half of the customers are women}. 
	Such features must be modified to resolve the conflict,
	and ASPECT always modifies the features that are applied earlier.

	(\textbf{O4}) 
	\textbf{Non-conflicting overlapping features}. 
	For non-conflicting features, 
	it is not always feasible to synthesize a dataset 
	that satisfies all of them, even if one exists. 
	For example, it is already NP-Hard to decide whether there exists a
	graph that satisfies certain degree distributions\cite{gmark}, 
	so there is no polynomial algorithm that generates a graph for such distributions.
	For the sake of efficiency, 
	we may have to sacrifice some feature accuracy.
	Even so, for the features in this paper, 
	ASPECT maintains the features accurately.
	As we can see from Sec.\ref{sec:improvement}, 
	the error is reduced to 0.02 after $2$ to $3$ iterations in the experiments..

	\section{Related Work}
	\label{sec:related}
	The Dataset Scaling Problem (DSP) was first advocated by Tay~\cite{vision}.
	There have been several solutions to this problem in the field of relational database.
	$\upsizer$~\cite{upsizer} is the first solution to DSP, 
	which uses attribute correlation extracted from an empirical dataset to
	generate a synthetic dataset.
	$\rex$~\cite{rex} is a later work that scales up the original dataset by
	an integer factor \textit{s}, using an automated representative extrapolation
	technique.
	Chronos~\cite{chronos} scales the streaming data by
	focusing on capturing and simulating streaming data
	with both column correlation and temporal correlation.
	Recently, DSP was extended to non-uniform DSP (nuDSP)~\cite{dscaler}. 
	As a solution to nuDSP, $\dscaler$ uses a {\it correlation database} 
	which captures fine-grained, per-tuple correlations for scaling.
	
	Data scaling is extended into other fields as well.
	In~\cite{rbench}, the authors propose a dataset scaling problem for RDF data and
	provide a solution RBench that scales the original input dataset by preserving  4 features:
	resource identity (resource name, resource type, resource degree), 
	relationship patterns (subgraphs with only relationship edges), 
	predicate dictionary (frequency counts of the words)
	and attribute stars (frequency counts of the star structure).
	In~\cite{vig}, the authors lift the scaling approach from
	the pure database level to the OBDA level, 
	where the domain information of ontologies and mappings are also taken into account as well.
	VIG~\cite{vig} maintains the similarity for OBDA data by preserving the following features:
	size of columns clusters and disjointness, schema dependencies and 
	column-based duplicates and NULL Ratios. 
	However, VIG only supports dataset where each table has at most one foreign key only.
	In most storage systems, compression time and compression ratio are important issues.
	Hence, these two criteria should also be used for similarity measurement. 
	In~\cite{SDGen}, SDGen is proposed to scale an input dataset to 
	an arbitrary size which preserves these two similarities.
	In~\cite{gscaler}, 
	the authors extend DSP to the Graph Scaling Problem for directed graphs.
	$\gscaler$ is proposed as a solution that maintains not only local graph properties, 
	e.g. degree distribution,
	but also global graph properties, e.g. effective diameter. 
	
	In the broader field, much work~\cite{flexible, quickBillion, sig2015,mudd} 
	have been done on query-independent 
	application-specific database generation (not scaling).
	The query-independent  
	database generation refers to generating the database in terms of a
	given real data set (data-driven) or a set of its character descriptions (character-driven), which is initiated by Jim Gray~\cite{quickBillion}.

	For character-driven data generation,
	a parallel algorithm is first proposed to generate a dataset with a
	predefined schema and a predefined distribution~\cite{quickBillion}.
	However, no correlations between attributes are preserved.
	Later, Bruno and Chaudhuri introduce DGL~\cite{flexible}, 
	a simple specification language to generate datasets with
	complex synthetic distributions and inter-table correlations.
	In~\cite{simpleGen}, Houkjaer proposes a table-wise graph model
	which holds various statistical information about the foreign key and column content.
	By using such a graph model, a more realistic dataset can be generated.
	PSDG~\cite{psdg} is another parallel solution to generate ``industrial sized'' dataset.
	PSDG supports easy parallelism, using a construct for specifying foreign keys.
	However, PSDG does not allow independent generation of dependent tables. 
	Referenced tables have to be created for generating dependent tables. 
	This is very inefficient if the referenced tables are not needed, and are huge in size.
	To overcome this, PDGF~\cite{pdgf} is proposed.
	
	Character-driven data generation also appears in graphs as well~\cite{chunglu,erdos60,jennifer2015,gmark}.
	For example, the Erd{\"{o}}s-R{\'{e}}nyi model generates a graph of
	any size $n$ with a specified edge probability $p$.
	gMark~\cite{gmark} is another tool that 
	generates a graph database based on users' specifications, 
	e.g. degree distribution and node type occurrence.
	Recently, TrillionG~\cite{trillionG} is proposed to generate large scale graphs in a short time. 
	It can generate trillion-node graphs within two hours by using 10PCs.
	
	For data-driven synthetic database generation, 
	the input is always a real dataset.
	MUDD~\cite{mudd} is the first tool which uses a real dataset for database generation.
	However, MUDD uses very little information (name and address) from the real dataset.
	Similarly, TEXTURE~\cite{texture} is another micro-benchmark for text query workloads.
	However, it only extracts simple properties,
	e.g. word distribution, document length.
	As an extension to PDGF,
	DBSynth~\cite{datasynth} scales the dataset by utilizing the meta-data
	(e.g. min/max constraints)
	and the statistics from the input dataset. 
	Strictly speaking, DSP is a sub-problem of query-independent synthetic database generation (data-driven).
	
	There are also domain-specific benchmarks that
	generate datasets for specific domains.
	They include TPC~\footnote{http://www.tpc.org/}, YCSB~\cite{ycsb}, 
	LinkBench~\cite{linkbench}, LDBC~\cite{ldbc}, 
	BigDataBench~\cite{biddatabench} and MWGen~\cite{MWGen}.
	For example, MWGen~\cite{MWGen} uses road and floor plans as input
	and generates a set of real world infrastructure together with moving objects in different transportation modes.
	Such domain-specific data generation usually generates 
	only one fixed dataset with fixed schema for 
	all applications under the same domain.
	Hence, domain-specific data generation is different from application-specific data scaling.
	
	Nevertheless, previous works generate datasets with pre-defined features, instead of flexible features.
	To the best of our knowledge, 
	ASPECT is the first framework which allows 
	dataset scaling with flexible features 
	through coordinating the tweaking tools.

	\section{Conclusion and future work}
	\label{sec:conclusion}
	This paper introduces ASPECT, 
	a framework for flexible application of tweaking tools to 
	enforce target features in synthetic dataset.
	To generate a scaled dataset with greater similarity 
	comparing to the original dataset,
	one just needs to apply more tweaking tools.
	We demonstrate ASPECT by coordinating 3 highly overlapping and complex tweaking tools 
	on real datasets to realize the target features.
	Extensive experiments show that ASPECT 
	effectively reduces the errors by orders of magnitudes
	in the synthetic data without sacrificing efficiency.
	
	ASPECT is a step towards the vision for application-specific benchmark data generation.
	It facilitates bottom-up collaboration among developers in contributing tools
	for tweaking synthetic datasets to enforce similarity with empirical data.
	
	Our current work is on the Feature Tweaking Bound and Order Problems. 
	We hope to make some progress by restricting the feature types (e.g. single joins).
	
	%\balance
	\bibliographystyle{abbrv}
	\bibliography{tweak}{}

\begin{thebibliography}{10}

\bibitem{joinsyn}
S.~Acharya, P.~B. Gibbons, V.~Poosala, and S.~Ramaswamy.
\newblock Join synopses for approximate query answering.
\newblock In {\em SIGMOD}, pages 275--286, 1999.

\bibitem{chunglu}
W.~Aiello, F.~Chung, and L.~Lu.
\newblock A random graph model for power law graphs.
\newblock {\em Experimental Mathematics}, 10(1):53--66, 2001.

\bibitem{datasynth}
A.~Arasu, R.~Kaushik, and J.~Li.
\newblock Data generation using declarative constraints.
\newblock In {\em SIGMOD}, pages 685--696, 2011.

\bibitem{linkbench}
T.~G. Armstrong, V.~Ponnekanti, D.~Borthakur, and M.~Callaghan.
\newblock {LinkBench}: A database benchmark based on the facebook social graph.
\newblock In {\em SIGMOD}, pages 1185--1196, 2013.

\bibitem{gmark}
G.~Bagan, A.~Bonifati, R.~Ciucanu, G.~H. Fletcher, A.~Lemay, and N.~Advokaat.
\newblock {gMark}: schema-driven generation of graphs and queries.
\newblock {\em IEEE TKDE}, 2016.

\bibitem{sonSchema}
Z.~Bao, Y.~C. Tay, and J.~Zhou.
\newblock {sonSchema}: A conceptual schema for social networks.
\newblock In {\em Int. Conf. Conceptual Modeling (ER)}, pages 197--211, 2013.

\bibitem{flexible}
N.~Bruno and S.~Chaudhuri.
\newblock Flexible database generators.
\newblock In {\em VLDB}, pages 1097--1107, 2005.

\bibitem{rex}
T.~Buda, T.~Cerqueus, et~al.
\newblock {ReX}: Extrapolating relational data in a representative way.
\newblock In {\em Data Science}, LNCS 9147, pages 95--107. Springer, 2015.

\bibitem{ycsb}
B.~F. Cooper, A.~Silberstein, E.~Tam, R.~Ramakrishnan, and R.~Sears.
\newblock Benchmarking cloud serving systems with {YCSB}.
\newblock In {\em ACM Symp. Cloud Computing}, pages 143--154, 2010.

\bibitem{grouptheme}
P.~Cui, T.~Zhang, F.~Wang, and P.~He.
\newblock Perceiving group themes from collective social and behavioral
  information.
\newblock In {\em AAAI}, pages 65--71, 2015.

\bibitem{texture}
V.~Ercegovac, D.~J. DeWitt, and R.~Ramakrishnan.
\newblock The {TEXTURE} benchmark: measuring performance of text queries on a
  relational {DBMS}.
\newblock In {\em VLDB}, pages 313--324, 2005.

\bibitem{erdos60}
P.~Erd{\"{o}}s and A.~R{\'{e}}nyi.
\newblock On the evolution of random graphs.
\newblock In {\em Publication of the Mathematical Institute of the Hungarian
  Academy of Science}, pages 17--61, 1960.

\bibitem{ldbc}
O.~Erling, A.~Averbuch, J.~Larriba-Pey, et~al.
\newblock The {LDBC} social network benchmark: interactive workload.
\newblock In {\em SIGMOD}, pages 619--630, 2015.

\bibitem{linkedSyn}
R.~Gemulla, P.~R{\"o}sch, and W.~Lehner.
\newblock Linked {Bernoulli} synopses: Sampling along foreign keys.
\newblock In {\em Scientific and Statistical Database Management}, pages 6--23,
  2008.

\bibitem{SDGen}
R.~Gracia-Tinedo, D.~Harnik, D.~Naor, et~al.
\newblock {SDGen}: Mimicking datasets for content generation in storage
  benchmarks.
\newblock In {\em USENIX Conf. File and Storage Technologies (FAST)}, pages
  317--330, Santa Clara, CA, 2015.

\bibitem{quickBillion}
J.~Gray, P.~Sundaresan, et~al.
\newblock Quickly generating billion-record synthetic databases.
\newblock In {\em SIGMOD}, pages 243--252, 1994.

\bibitem{chronos}
L.~Gu, M.~Zhou, Z.~Zhang, et~al.
\newblock Chronos: An elastic parallel framework for stream benchmark
  generation and simulation.
\newblock In {\em ICDE}, pages 101--112, 2015.

\bibitem{recommendation}
X.~He, H.~Zhang, M.-Y. Kan, and T.-S. Chua.
\newblock Fast matrix factorization for online recommendation with implicit
  feedback.
\newblock In {\em SIGIR}, volume~16, 2016.

\bibitem{psdg}
J.~E. Hoag and C.~W. Thompson.
\newblock A parallel general-purpose synthetic data generator.
\newblock {\em SIGMOD Rec.}, 36(1):19--24, Mar. 2007.

\bibitem{simpleGen}
K.~Houkj{\ae}r, K.~Torp, and R.~Wind.
\newblock Simple and realistic data generation.
\newblock In {\em VLDB}, pages 1243--1246, 2006.

\bibitem{vig}
D.~Lanti, G.~Xiao, and D.~Calvanese.
\newblock Fast and simple data scaling for {OBDA} benchmarks.
\newblock {\em Proc. BLINK}, 2016.

\bibitem{jennifer2015}
S.~Mussmann, J.~Moore, J.~J. Pfeiffer, and J.~Neville.
\newblock Incorporating assortativity and degree dependence into scalable
  network models.
\newblock In {\em AAAI}, 2015.

\bibitem{trillionG}
H.~Park and M.-S. Kim.
\newblock {TrillionG}: A trillion-scale synthetic graph generator using a
  recursive vector model.
\newblock In {\em SIGMOD}, pages 913--928, 2017.

\bibitem{sdv}
N.~Patki, R.~Wedge, and K.~Veeramachaneni.
\newblock The synthetic data vault.
\newblock In {\em DSAA}, pages 399--410, Oct 2016.

\bibitem{rbench}
S.~Qiao and Z.~M. {\"O}zsoyo{\u{g}}lu.
\newblock {RBench}: Application-specific {RDF} benchmarking.
\newblock In {\em SIGMOD}, pages 1825--1838, 2015.

\bibitem{sig2015}
T.~Rabl, M.~Danisch, et~al.
\newblock Just can't get enough: Synthesizing big data.
\newblock In {\em SIGMOD}, pages 1457--1462, 2015.

\bibitem{pdgf}
T.~Rabl, M.~Frank, H.~M. Sergieh, and H.~Kosch.
\newblock A data generator for cloud-scale benchmarking.
\newblock In {\em TPC Tech. Conf. (TPCTC)}, pages 41--56, 2010.

\bibitem{robson}
J.~M. Robson.
\newblock Algorithms for maximum independent sets.
\newblock {\em {Journal of Algorithms}}, 7(3):425--440, 1986.

\bibitem{volunteer}
X.~Song, Z.-Y. Ming, L.~Nie, Y.-L. Zhao, and T.-S. Chua.
\newblock Volunteerism tendency prediction via harvesting multiple social
  networks.
\newblock {\em ACM Trans. Inf. Syst.}, 34(2):10:1--10:27, Feb. 2016.

\bibitem{mudd}
J.~M. Stephens and M.~Poess.
\newblock {MUDD}: a multi-dimensional data generator.
\newblock In {\em SIGSOFT Software Engineering Notes}, pages 104--109, 2004.

\bibitem{stonebreaker}
M.~Stonebraker.
\newblock A new direction for {TPC}?
\newblock In {\em TPCTC}, pages 11--17, 2009.

\bibitem{vision}
Y.~C. Tay.
\newblock Data generation for application-specific benchmarking.
\newblock {\em {PVLDB}}, 4(12):1470--1473, 2011.

\bibitem{upsizer}
Y.~C. Tay, B.~T. Dai, et~al.
\newblock {UpSizeR}: Synthetically scaling an empirical relational database.
\newblock {\em Inf. Syst.}, 38(8):1168--1183, 2013.

\bibitem{biddatabench}
L.~Wang, J.~Zhan, C.~Luo, et~al.
\newblock {BigDataBench}: A big data benchmark suite from internet services.
\newblock In {\em High Performance Computer Architecture (HPCA)}, pages
  488--499, 2014.

\bibitem{MWGen}
J.~Xu and R.~H. G\"{u}ting.
\newblock {MWGen}: A mini world generator.
\newblock In {\em Mobile Data Management (MDM)}, pages 258--267, July 2012.

\bibitem{dscaler}
J.~W. Zhang and Y.~C. Tay.
\newblock Dscaler: Synthetically scaling a given relational database.
\newblock In {\em VLDB}, pages 1671--1682, 2016.

\bibitem{gscaler}
J.~W. Zhang and Y.~C. Tay.
\newblock {GSCALER}: Synthetically scaling a given graph.
\newblock In {\em EDBT}, pages 53--64, 2016.

\bibitem{aspect}
J.~W. Zhang and Y.~C. Tay.
\newblock Synthetic dataset scaling with flexible features.
\newblock \tt http://www.comp.nus.edu.sg/$\sim$upsizer/, 2017.

\end{thebibliography}

	\section{Appendix}
	\subsection{Algorithm pseudocode}
	In this section, we present the pseudocode for three tweaking tools: $\Tlinear$, $\Tcoappear$, $\Tpairwise$.
	\begin{algorithm}[h]
		\caption{\textbf{Linear Feature Tweaking}}
		\label{algo:Tlinear}
		\For{each row of $H$}{
			leadingElementAdjust()     //Lemma ~\ref{lemma:base}\\
			\For{ 2 to row\_length}{
				nonLeadingElementAdjust()    //Lemma~\ref{lemma:induction}
			}
		} 
	\end{algorithm}
	
	\begin{algorithm}[h]
		\caption{\textbf{Coappear Feature Tweaking}}
		\label{algo:Tcoappear}
		\For{each $(v_1,\dots,v_k) \in \Delta^-$}{
			\While{$\xi^*(v_1,\dots,v_k)\neq0$}{
				$(v'_1,\dots,v'_k) \leftarrow$ closest$(\Delta^+, (v_1,\dots,v_k))$\\
				\For{$i \in [1,k]$}{
					$t_i \leftarrow$ 
					tupleRetrieve$(T_i, (v'_1,\dots,v'_k))$\\
					tupleModification$(T_i, v'_i-v_i, t_i)$		
				}
				statsUpdate$(\xi^*(v_1,\dots,v_k), \xi^*(v'_1,\dots,v'_k))$\\
			}
		} 
		
	\end{algorithm}

	\begin{algorithm}[h!]
		\caption{\textbf{Pairwise Feature Tweaking}}
		\label{algo:Tpairwise}
		\For {each $\po^*_R$} {
			\For{each $(x,y) \in \Theta^-$}{
				\While{$\po_R^*(x,y)\neq0$}{
					$(x',y') \leftarrow$ closest$(\Theta^+, (x,y))$\\
					$u, v \leftarrow$ 
					tupleRetrieve$(R, (x',y'))$\\
					tupleModification$(R, x'-x, y'-y, u,v)$\\	
					statsUpdate$(\po^*_R(x,y))$\\
					statsUpdate$(\po^*_R(x',y'))$
				}
			} 
		}
	\end{algorithm}

	\subsection{Theorems and proofs for $\linear$ feature}
	In this section, we present the the formal proofs that are related to the $\linear$ feature.
	Define $H^{*} = \scaleH - \tweakH$ and 
	$h^{*}_{j,i} = \scaleh_{j,i} - \tweakh_{j,i}$ for all $1\leq i,j \leq k$.\newline
	
	\begin{lemma}
		[{\bf \tt leadingElementAdjust}]
		If the first $n$ rows of $\scaleH$ and $\tweakH$ are the same,
		then $(h^{*}_{n+1,1},h^{*}_{n+1,2},\dots, $ $h^{*}_{n+1,n})$
		can be tweaked to 
		$(0,h^{*\prime}_{n+1,2}, \dots, h^{*\prime}_{n+1,n})$,
		where $h^{*\prime}_{n+1,i}\geq0$ for $2\leq i \leq n$.
		\label{lemma:base}
	\end{lemma}
	\begin{proof}
		There are two cases: $h^*_{n+1,1} > 0 $ and $h^*_{n+1,1} < 0 $.
		
		\indent \textbf{Case $h^*_{n+1,1} > 0 $}: 
		There are $h^*_{n+1,1}$ more tuples in $T_1$ having descendants in $T_{n+1}$.
		Hence, tweaking is needed to make these $h^*_{n+1,1}$ tuples have no descendants in $T_{n+1}$.
		It takes two steps: \textit{Leaf Tuple Plucking} and {\it Leaf Tuple Attaching}.
		
		{\textbf{\textit{Leaf Tuple Plucking}}}: 
		Consider $S_{n+1,1}$, the tuples in $T_1$ which have descendants in $T_{n+1}$.
		Let $R_{n+1,1} $  be the $h^*_{n+1,1}$ tuples from $S_{n+1,1}$ 
		with least number of descendants in $T_{n+1}$.
		Let $Q_{n+1,1}$ be the tuples in $T_{n+1}$ which are descendants of tuples in 
		$R_{n+1,1}$.
		
		Pluck all tuples in $Q_{n+1,1}$ by removing their foreign key reference to $T_{n}$.
		Then, all tuples in $R_{n+1,1} $ have no descendants in $T_{n+1}$,
		so, $h^{*}_{n+1,1} = 0$ after the tuple plucking. 
		Next, we will do \textit{Leaf Tuple Attaching}.

		\textbf{\textit{Leaf Tuple Attaching}}:
		Let $V_{n+1,1}$ be the tuples in $T_n$ which are descendants of $S_{n+1,1} - R_{n+1,1} $.
		All tuples in $S_{n+1,1} - R_{n+1,1} $ already have descendants in $T_{n+1}$. 
		Hence, attaching more tuples to $V_{n+1,1}$ will not change $h^{*}_{n+1,1}$.
		Therefore, 
		attach back all the tuples in  $Q_{n+1,1}$ by 
		setting their foreign key reference randomly to tuples in $V_{n+1,1}$. \\

		\indent \textbf{Case $h^*_{n+1,1}<0$}:
		There are $h^*_{n+1,1}$ more tuples in $T_1$ which do not have descendants in $T_{n+1}$.
		The tweaking takes two steps by doing tuple plucking first, and then tuple attaching. 
		
		\textbf{\textit{Leaf Tuple Plucking:}}
		$|h^*_{n+1,1}|$ tuples in $T_{n+1}$ need to be found first. 
		Once these $|h^*_{n+1,1}|$ tuples are plucked, 
		we can attach them to the tuples in $T_{n}$ to increase $h^*_{n+1,1}$.
		For each tuple in $S_{n+1, 1}$ 
		(the tuples in $T_1$ which have descendants in $T_{n+1}$), 
		pick $1$ descendant in $T_{n+1}$ to form a leaf set $Leaf_{n+1,1}$.
		It is obvious that plucking any tuples from $T_{n+1} - Leaf_{n+1,1}$ never modify $h^*_{n+1,1}$.
		Then,
		\begin{equation*}
		\begin{split}
		& |T_{n+1} - Leaf_{n+1,1}| - |h^*_{n+1,1}| \\
		& =|T_{n+1}| - \scaleh_{n+1,1} - (\tweakh_{n+1,1} -{\scaleh}_{n+1,1} ) \\
		& = 
		|T_{n+1}| - \tweakh_{n+1,i} \geq 0\ \ \ \ \ \ \ \ \  \ \  \  (By\ L1)
		\end{split}
		\end{equation*} 
		Hence,  $|h^*_{n+1,1}|$ tuples in $T_{n+1} - Leaf_{n+1,1}$ can be randomly plucked.
		Next, we will attach these tuples back.
		
		\textbf{\textit{Leaf Tuple Attaching:}}
		Consider $S_{n,1} - S_{n+1,1}$,  
		the set of tuples in $T_{1}$ having descendants in $T_{n}$ but no descendants in $T_{n+1}$.
		Since
		\begin{equation*}
		\begin{split}
		& |S_{n,1} - S_{n+1,1}| - |h^*_{n+1,i}| \\
		& = |\scaleh_{n,1} - \scaleh_{n+1,1}| - | \scaleh_{n+1,1} - \tweakh_{n+1,1} |  \\
		&	=\scaleh_{n,1} - \scaleh_{n+1,1} - (- (  \scaleh_{n+1,1} - \tweakh_{n+1,1} )) \\
		&	= \scaleh_{n,1}  - \tweakh_{n+1,1} \\
		&	= \tweakh_{n,1}  -\tweakh_{n+1,1}  \geq 0\ \ \ \ \ \  \ \  \  (By\ L2)
		\end{split}	
		\end{equation*}
		
		Then, there are $|h^*_{n+1,1}|$ tuples from $S_{n,1} - S_{n+1,1}$, 
		denoted as $Sub_{n+1,1}$. 
		For each tuple in $Sub_{n+1,1}$, it must have a descendant in $T_n$. 
		Hence, randomly attach the  $|h^*_{n+1,1}|$ tuples from \textit{Leaf Tuple Plucking:} 
		to the decent in $T_n$.
		Then all tuples in $Sub_{n+1,1}$ have descendants in $T_{n+1}$ now.
		Hence, $h^*_{n+1,i}=0$ after attachment.
	\end{proof}
	
	\begin{figure}
		\centering
		\includegraphics[height=1.1in]{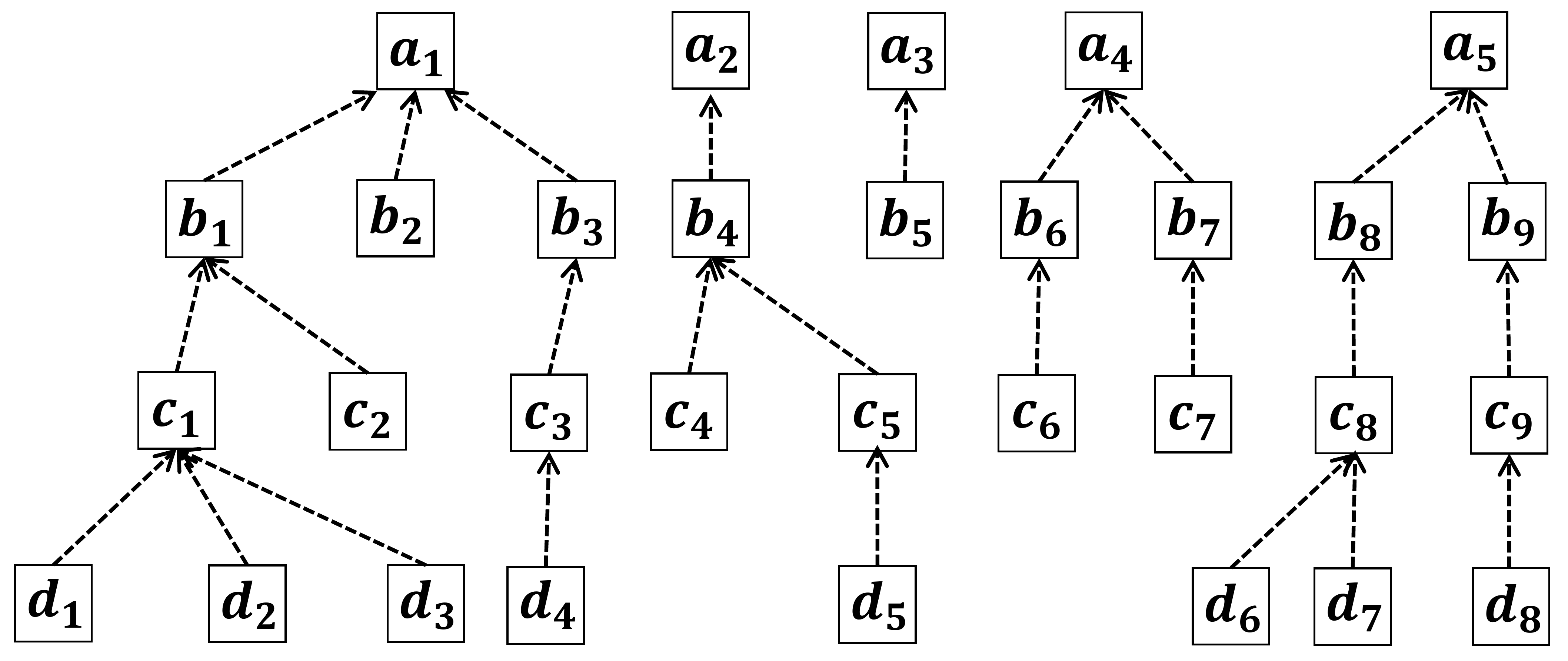}
		\caption{Non-Leading Element Tweaking Demonstration}
		\label{fig:nonLead}
	\end{figure}
	
	We are done with the leading element tweaking; 
	next, we will prove the correctness non-leading element tweaking.
	Some care is needed to tweak $h^{*\prime}_{n+1,2}$.
	
	Suppose $h^{*\prime}_{4,2}=1$ for Fig.\ref{fig:nonLead},
	so we want to remove descendants in  $T_D$ for 1 tuple in $T_B$.
	If we do this by plucking $d_5$ from $c_5$,
	then $h^*_{4,1}$ is decreased by 1 as well,
	so we should instead pluck $d_6$ or $d_7$.
	Therefore, when tweaking $h^*_{n+1,i}$,
	we should avoid affecting $h^*_{n+1,i-1}$.
	
	If $h^{*\prime}_{4,2}=-1$, 
	we need to add  descendants in $T_D$ for 1 more tuple
	in $T_B$.
	We can pluck a leaf, say $d_1$, and attach it to $c_6$;
	this will increase $\widetilde{h}_{4,2}$,
	but it will also increase $\widetilde{h}_{4,1}$.
	However, if we first pluck the subtree rooted at $b_6$ and attach it to $a_2$,
	then pluck $d_1$ and attach it to $c_6$,
	$\widetilde{h}_{4,2}$ will increase 
	without affecting other $\widetilde{h}_{j,i}$ values.
	This leads us to the following definition:
	\newline

	\begin{definition}
		For a reference chain $T_k\rightarrow\dots\rightarrow T_1$,
		suppose we pluck $t\in T_i$ from $t^\prime\in T_{i-1}$
		and attach $t$ to some other $t^{\prime\prime}\in T_{i-1}$.
		We call this an {\bf isomorphic adjustment} if the linear join matrix
		is unchanged.\newline
	\end{definition}
	
	\begin{lemma}
		For a reference chain $T_k\rightarrow\dots\rightarrow T_1$,
		we can make $|S_{k-1,i} - S_{k,i}| - |S_{k-1,i-1} - S_{k,i-1}|$
		isomorphic adjustments to $T_i$.
		\label{lemma:maxiso}
	\end{lemma}
	\begin{proof}
		($S_{k-1,i-1} - S_{k,i-1}$) are the tuples in $T_{i-1}$ having descendants in $T_{k-1}$ 
		but no descendants in $T_{k}$. 
		Similarly, ($S_{k-1,i} - S_{k,i}$) are the tuples in $T_i$ having descendants in $T_{k-1}$ 
		but no descendants in $T_{k}$. 
		For each tuple in $S_{k-1,i-1} - S_{k,i-1}$ , 
		pick $1$ descendant in $T_i$ with descendants in $T_{k-1}$ to form a set $S_y$.
		Then, $S_y \subseteq S_{k-1,i} - S_{k,i}$. 
		Let's consider $S_{k-1,i} - S_{k,i} - S_y$. 
		For any tuple in $S_{k-1,i} - S_{k,i} - S_y$, 
		we can pluck it and attach it back  randomly to tuples in 
		$S_{k,i-1}$. 
		Theses adjustments are isomorphic.
		Hence, the maximum number is $|S_{k-1,i} - S_{k,i} - S_y| = |S_{k-1,i} - S_{k,i}| - (|S_{k-1,i-1} - S_{k,i-1}|)$.
	\end{proof}

	\begin{lemma}
		[{\bf \tt nonLeadingElementAdjust}]
		Suppose the first $n$ rows of $\scaleH$ and $\tweakH$ are the same.
		Then $(0,\dots, 0, h^{*}_{n+1,i},\ldots, h^{*}_{n+1,n})$ can be tweaked to
		$(0,\dots, 0, h^{*\prime}_{n+1,i+1},$ $\ldots, h^{*\prime}_{n+1,n})$.
		\label{lemma:induction}
	\end{lemma}
	\begin{proof}
		The proof and tweaking steps are similar to Lemma~\ref{lemma:base}. 
		There are two cases: 	
		
		\indent \textbf{Case $h^*_{n+1,i} > 0 $}: Similar to Lemma~\ref{lemma:base}.
		There are two steps:
		
		\textbf{	\textit{Leaf Tuple Plucking }}: 
		For each tuple in $S_{n+1, i-1}$, 
		pick one descendant $t_y$ in $T_{i}$,  
		where $t_y$ has descendants in $T_{n+1}$.
		Use $R_{y}$ to denote the set of all such $t_y$.
		Therefore, $|R_y| = |S_{n+1,i-1}|$.
		For any tuple in $S_{n+1,i}-R_{y}$, $\scaleh_{n+1,i}$ will be decreased 
		if all its descendants in $T_{n+1}$ are detached.
		Moreover, such detachment will not affect $\scaleh_{n+1,i-1}$.
		Thus,
		\begin{equation*}
		\begin{split}
		|S_{n+1,i}-R_{y}|  &=\scaleh_{n+1,i} - \scaleh_{n+1,i-1}\\
		&=\scaleh_{n+1,i} - \tweakh_{n+1,i-1}\\
		& \geq  \scaleh_{n+1,i} - \tweakh_{n+1,i} \ \ \ \ \ (By\  L3)\\
		&= h^*_{n+1,i}
		\end{split}
		\end{equation*}
		
		Hence, $h^*_{n+1,i}$ tuples in $S_{n+1,i}-R_{y}$
		can be randomly plucked.
		
		\textbf{\textit{Leaf Tuple Attaching}:}
		Similar to Lemma~\ref{lemma:base}.\\

		\indent \textbf{Case $h^*_{n+1,i}<0$}:
		Similar to Lemma~\ref{lemma:base}, there are two steps: Leaf Tuple Plucking and Leaf Tuple Attaching.
		
		\textbf{\textit{Leaf Tuple Plucking:}} Similar to Lemma~\ref{lemma:base}.

		\textbf{\textit{Leaf Tuple Attaching:}}
		If no isomorphic adjustment is needed, then it is the same as Lemma~\ref{lemma:base}.
		Otherwise, based on Lemma~\ref{lemma:maxiso}, the maximum isomorphic adjustment is  
		$$|S_{n,i}| - |S_{n+1,i}| - (|S_{n,i-1}| -|S_{n+1,i-1}|)$$
		Moreover, 
		\begin{equation*}
		\begin{split}
		& |S_{n,i}| - |S_{n+1,i}| - (|S_{n,i-1}| -|S_{n+1,i-1}|) - |h^*_{n+1,i}|  \\
		= &  
		\scaleh_{n,i} - \scaleh_{n+1,i} + \scaleh_{n+1,i-1} - \scaleh_{n,i-1} - \tweakh_{n+1,i} + \scaleh_{n+1,i} \\
		= & \scaleh_{n,i}  + \scaleh_{n+1,i-1} - \scaleh_{n,i-1} - \tweakh_{n+1,i}  \\
		= &\tweakh_{n,i}  + \tweakh_{n+1,i-1} - \tweakh_{n,i-1}-\tweakh_{n+1,i} \geq 0 \ \ \ \ \ (By \ \ L4)
		\end{split}
		\end{equation*}	
		Therefore, at least $|h^*_{n+1,i}|$ isomorphic adjustments in $T_i$ can be made. 
		Let $Sub_{n+1,i}$ be the
		set of tuples that undergo isomorphic adjustments.
		For each tuple in the subset $Sub_{n+1,i}$, 
		pick $1$ descendant in $T_{n}$ and randomly attach a tuple plucked from the previous step.
		Hence, the leaf tuple attaching can be done.
		\newline
	\end{proof}

	\begin{theorem}
		For a reference chain $T_k\rightarrow\dots\rightarrow T_1$ in some $\scaled_i$,
		let $\scaleH$ be the linear join matrix before tweaking
		and $\tweakH$ the target linear join matrix.
		If $\tweakH$ satisfies the necessary conditions in 
		Theorem~\ref{theorem:linearnec}, 
		then Algorithm~\ref{algo:Tlinear} tweaks $\scaleH$ to give $\tweakH$.
		\label{theorem:Tlinear}
	\end{theorem}
	\begin{proof}
		Algorithm~\ref{algo:Tlinear} iterates over the rows of $H^*=\scaleH-\tweakH$;
		for each row, the first entry is tweaked to 0 with Lemma~\ref{lemma:base},
		and the following entries are tweaked to 0 with Lemma~\ref{lemma:induction}.
	\end{proof}

	\subsection{Theorem and proofs for $\coappear$ feature}
	In this section, we present formal proofs that are related to $\coappear$ feature.\newline
	
	\textit{Theorem 3}  \text{\normalfont [{\bf sufficiency}]}
	Suppose tables $T_1, \dots, T_k$ reference the same tables 
	$T_1^\prime, \dots, T_m^\prime$. 
	Let $\scalexi$ be the coappear distribution in some $\scaled_i$ before tweaking and 
	$\tweakxi$ the target coappear distribution.
	If $\tweakxi$ satisfies the necessary conditions in
	Theorem~\ref{theorem:coappearnecessary},
	then $\Tcoappear$ tweaks $\scalexi$ to become $\tweakxi$. 
	\begin{proof}
		Both $\scalexi$ and $\tweakxi$ satisfy $C1$,
		so $|T_i|$ is unaffected by the tweaking.
		Similarly, $C2$ ensures
		$ \sum_{\bf v}\scalexi= \sum_{\bf v}\tweakxi$,
		so $\sum_{\bf v}\xi^*=0$.
		Therefore
		\[
		\sum_{{\bf v}\in\Delta^+}\xi^*({\bf v})+
		\sum_{{\bf v}\in\Delta^0}\xi^*({\bf v})+
		\sum_{{\bf v}\in\Delta^-}\xi^*({\bf v})=0,
		\]
		so
		$ \sum_{{\bf v}\in\Delta^+}\xi^*({\bf v})=
		\sum_{{\bf v}\in\Delta^-}-\xi^*({\bf v})$,
		i.e.
		$ \sum_{{\bf v}^\prime\in\Delta^+}\xi^*({\bf v}^\prime)=
		\sum_{{\bf v}\in\Delta^-}|\xi^*({\bf v})|$.
		Each tweak decreases $\xi^*({\bf v}^\prime)$ by 1 and increases 
		$\xi^*({\bf v})$ by 1 for some ${\bf v}^\prime\in\Delta^+$
		and ${\bf v}\in\Delta^-$.
		After $\sum_{{\bf v}^\prime\in \Delta^+}\xi^*({\bf v}^\prime)$ iterations,
		we get
		$ \sum_{{\bf v}^\prime\in\Delta^+}\xi^*({\bf v}^\prime)=0=
		\sum_{{\bf v}\in\Delta^-}|\xi^*({\bf v})|$, so $\xi^*=0$;
		i.e. $\scalexi=\tweakxi$.
	\end{proof}

	\subsection{Theorem and proofs for $\pairwise$ feature}
	In this section, we present formal proofs that are related to $\linear$ feature.
	We first prove Theorem~\ref{theorem:pairsufficient}.\newline
	
	\textit{Theorem 5} \text{\normalfont [{\bf sufficiency}]}
	For each $\response$ table $R$, 
	let $\scalepo_R$ be the pairwise distribution in $\scaled$ and 
	$\tweakpo_R$ the target pairwise distribution.
	If $\tweakpo_R$ satisfies ($P3$) in Theorem~\ref{theorem:pairnec},
	then Algorithm~\ref{algo:Tpairwise} tweaks $\scalepo_R$ to become $\tweakpo_R$.
	Moreover, the extra tuples added to the {\tt post} table $P$ is at most
	$|U|-|P|$, where $U$ is the {\tt user} table.
	
	\begin{proof}
		Since $\scalepo_R$ and $\tweakpo_R$ both satisfy ($P3$),
		and $\Tpairwise$ does not affect $|U|$,
		we have $\sum_{x,y}\scalepo_R(x,y)= \sum_{x,y}\tweakpo_R(x,y)$,
		and so $\sum_{x,y}\rho^*_R(x,y)=0$.
		As in the proof of Theorem~\ref{theorem:Tcoappear},
		this implies
		\[
		\sum_{(x,y)\in\Theta^+}\rho^*_R(x,y)=
		\sum_{(x,y)\in\Theta^-}|\rho^*_R(x,y)|.
		\]
		Each tweak by Algorithm~\ref{algo:Tpairwise} 
		increases $\rho^*_R(x,y)$ and $\rho^*_R(y,x)$ by 1, and
		decreases $\rho^*_R(x^\prime,y^\prime)$ and $\rho^*_R(y^\prime,x^\prime)$ by 1
		for $(x,y)\in\Theta^-$ and $(x^\prime,y^\prime)\in\Theta^+$.
		After $\frac{1}{2}\sum_{(x,y)\in\Theta^+}\rho^*_R(x,y)$ loops,
		we get  
		$\sum_{(x,y)\in\Theta^+}\rho^*_R(x,y)=0
		= \sum_{(x,y)\in\Theta^-}|\rho^*_R(x,y)|$,
		so $\rho^*_R=0$, i.e. $\scalepo=\tweakpo$.
		
		Moreover, 
		if a new post needs to be added to $P$,
		then each user has at most 1 post.
		Thus, there are $|U|-|P|$ users who do not have posts,
		so that many posts need to be added to ensure each user has 1 post.
		\newline
	\end{proof}
	
	Previously, we assume a user does not respond to his/her own post.
	Now, we remove the assumption. 
	However, we separate the distribution $\rho_R$ into 
	2 distributions: $\rho_S$ and $\rho_N$, 
	where $\rho_S$ is the distribution generated by self-responding behavior, and
	$\rho_N$ does not contain any pairwise vector generated by self-responding.
	Sec.\ref{sec:Tpairwise} has discussed the case for $\rho_N$, 
	so we now discuss  tweaking for $\rho_S$.\newline

	\begin{theorem}
		For a $\response$ table $R$ in some $\scaled_i$
		let $\scalepo_S$ be the pairwise distribution generated 
		by user self responding before tweaking.
		and $\tweakpo_S$ the target pairwise distribution.
		If $\scalepo_S$ can be tweaked to become $\tweakpo_S$,
		then $\tweakpo_S$ satisfies the following conditions:
		\begin{align}
		(SP1)&\sum_{x}2x\tweakpo_S(x,x) + \sum_{x,y}(x+y)\tweakpo_N(x,y) = 2|T_R|   \nonumber \\
		(SP2)&\sum_{x}\tweakpo_S(x,x)  = |U|\ 
		{\rm where\ U\ is\ the\ }{\tt user}\ {\rm table}. \nonumber
		\end{align}
		\label{theorem:pairnecall}	
	\end{theorem}
	\begin{proof}
		\noindent
		($SP1$) For each response tuple $t$, made from $u_i$ to $v_i$. 
		If $u_i \neq v_i$, assuming their pairwise vector is (x,y), 
		then it is double-counted by $\tweakpo_R(x,y)$ and  $\tweakpo_R(x,y)$.
		If $u_i = v_i$, then there are only $x$ tuples, which are double-counted by $(x+y)$.
		So we get the equality in ($SP1$).
		
		\noindent
		($SP2$) There are $|U|$ users, each user can respond to himself. Hence counted only once.
		\newline
	\end{proof}

	For tweaking $\rho_S$, it is similar to $\rho_N$.

	Let $\rho^*_S=\scalepo_S-\tweakpo_S$,
	$\Theta_S^+=\{(x,x)|\rho^*_S(x,x)>0\}$ and
	$\Theta_S^-=\{(x,x)|\rho^*_S(x,x)<0\}$.
	For each $(x,x)\in\Theta_S^-$, it adds $|\rho^*_S(x,x)|$ pairs 
	$\langle u_i,u_i\rangle$, 
	where user $u_i$ has $x$ $\response$ tuples in $R$ referencing $u_i$'s post.
	It does this by looping $|\rho^*_S(x,x)|$ times, and in each iteration:
	
	{\bf PairwiseVectorRetrieve:}
	Pick $v^\prime=(x^\prime,x^\prime)\in\Theta_S^+$ that is closest to $(x,x)$
	by Manhattan distance.
	
	{\bf TupleModification:}
	Choose users $u$ with self-respond  pairwise vector $(x^\prime,x^\prime)$
	and tweak $u$'s responses to $u$'s post, as follows:
	If $x<x^\prime$:
	this means $u$ has $x^\prime-x$ more responses to $u$'s post than desired,
	so $\Tpairwise$ randomly chooses and removes $x^\prime-x$ such responses.
	If $x>x^\prime$:
	$\Tpairwise$ adds $x-x^\prime$ responses from $u$ on $u$'s post.
	If $u$ has no post, we artificially create a post for $u$.
	To do this, we pick another user $w$ who has more than 1 post and,
	among $w$'s posts, pick a post $p_w$ with minimum responses;
	we make $p_w$ a post by $u$, 
	and shift the responses to $p_w$ to other posts by $w$.
	If (in the worst case) all other users have at most 1 post,
	then we create a new post $p$ for $u$, and add $x-x^\prime$ responses to $p$.
	
	{\bf StatsUpdate:}
	Increase $\rho^*_S(x,x)$ by 1 and 
	decrease $\rho^*_S(x^\prime,x^\prime)$  by 1.

	The following theorem says that conditions in Theorem~\ref{theorem:pairnecall} suffices to ensure that
	$\Tpairwise$ tweaks $\scalepo_S$ to become $\tweakpo_S$.
	\newline
	
	\begin{theorem}
		For each $\response$ table $R$ in some $\scaled_i$, 
		let $\scalepo_S$ be the self-responded pairwise distribution before tweaking and 
		$\tweakpo_S$ the target pairwise distribution.
		If $\tweakpo_S$ satisfies the conditions in Theorem~\ref{theorem:pairnecall},
		then $\scalepo_S$ can be tweaked to  $\tweakpo_S$.
		Moreover, the extra tuples added to the {\tt post} table $P$ is at most
		$|U|-|P|$, where $U$ is the {\tt user} table.
	\end{theorem}
	\begin{proof}
		The proof is similar to the proof in Sec.\ref{sec:Tpairwise}
	\end{proof}

	\subsection{Dataset summary}
	In this section, we summarize the datasets used in the experiments.
	We used $4$ datasets from $\Douban$\footnote{https://www.douban.com} 
	and $\Xiami$\footnote{https://www.xiami.com}.  
	$\Douban$ is a Chinese social network website that allows 
	the creation and sharing of content related 
	to movies, books, music, recent events and activities in Chinese cities.
	$\Xiami$ is a Chinese online music website 
	that provides recommendations of music services,
	offline music activities, and other interactive content.
	The short summary of the $4$ datasets are the following:
	\begin{enumerate}
		\item 
		$\DoubanMovie$ contains movie-related data in 
		17 tables, with table sizes ranging from 10856 to 36747342 tuples. 
		
		\item 
		$\DoubanBook$ contains book-related data in 
		12 tables, with table sizes ranging from 686605 to 12891598 tuples.  
		
		\item 
		$\DoubanMusic$ contains music-related data in 
		10 tables, with table sizes ranging from 52078 to 7086936 tuples. 
		
		\item 
		$\Xiami$ also contains music-related data, but is larger:
		It has 26 tables and more than 90millions tuples.
	\end{enumerate}

	Fig.\ref{fig:datasize} presents the dataset size for each partition. 
	For example, the 6th partition of {\tt DoubanMovie} $\oldd_6$ is 2.5 Gigabytes.
	\begin{figure}[h!]
		\centering
		\includegraphics[height = 0.88in]{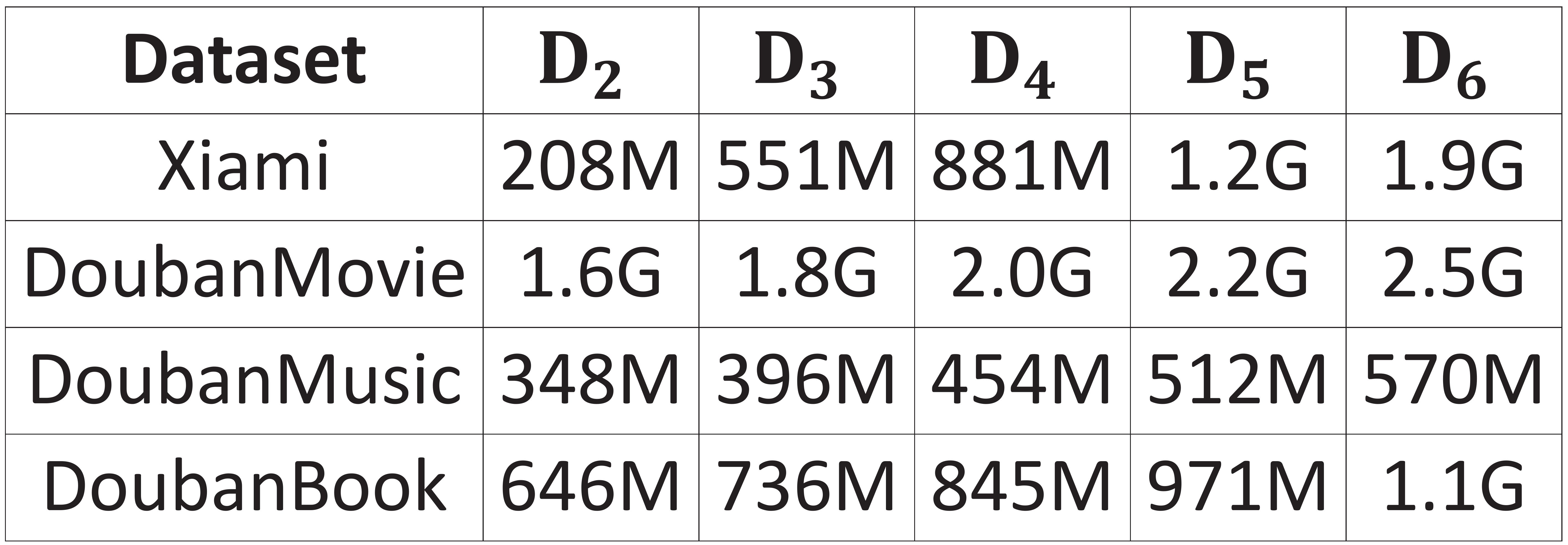}
		\caption{Dataset Size Summary}
		\label{fig:datasize}
	\end{figure}
	
	The schema of each dataset is presented as follows.
	\begin{figure}[h!]
		\centering
		\includegraphics[height = 2.1in,trim=4 4 4 4,clip]{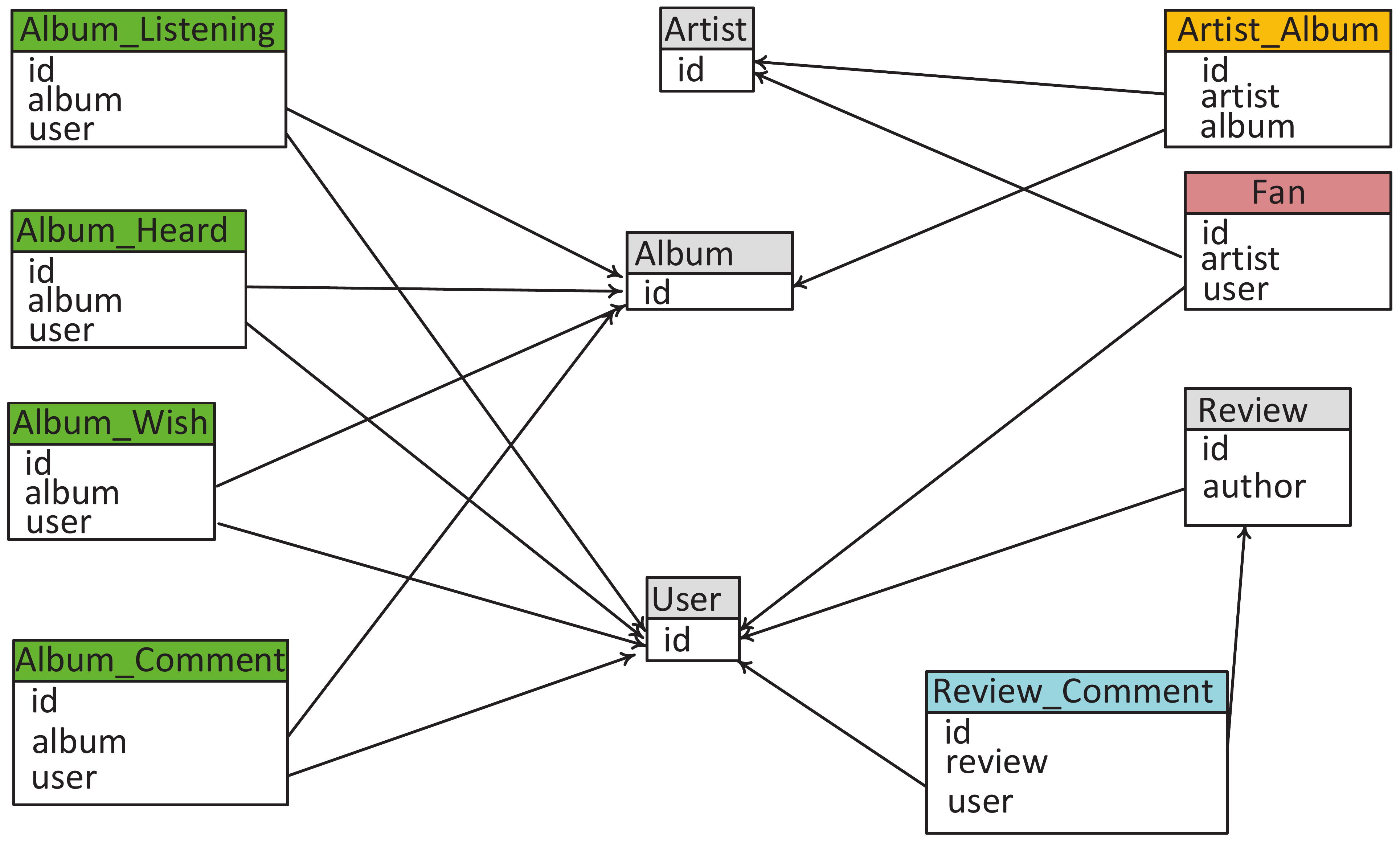}
		\caption{Schema For $\DoubanMusic$: 
			There are 11 tables;
			The tables with the same color (except grey color) share the same coappear distribution.  
			For example, {\tt Album\_Comment, Album\_Listening, Album\_Heard, Album\_Wish} reference to both {\tt Album} and {\tt User } tables; {\tt Review} is the {\tt post} table; {\tt Review\_Comment} is the $\response$ table.}
		\label{fig:schemamusic}
	\end{figure}

	\begin{figure}[h!]
		\includegraphics[height = 2.5in]{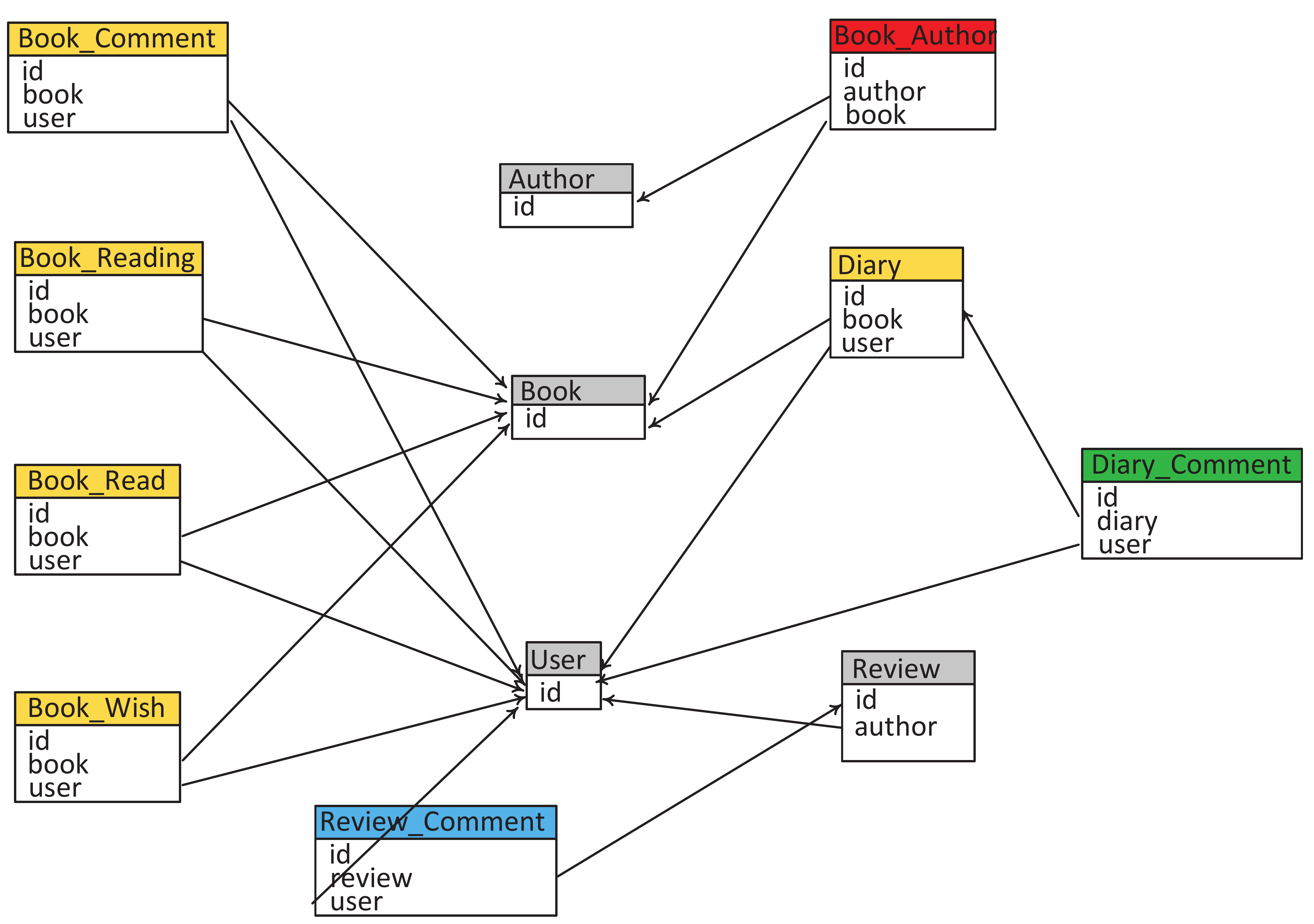}
		\caption{Schema For $\DoubanBook$: 
			There are 12 tables;
			The tables with the same color (except grey color) share the same coappear distribution.  
			For example, {\tt Book\_Comment}, {\tt Book\_Reading}, {\tt Book\_Read}, {\tt Book\_Wish, Diary} 
			reference to both {\tt Book} and {\tt User} tables; 
			{\tt Diary} and {\tt Review} are the {\tt post} tables; {\tt Diary\_Comment} {\tt Review\_Comment} are the $\response$ tables.}
		\label{fig:schemabook}
	\end{figure}

	\begin{figure}[h!]
		\includegraphics[height = 2.8in]{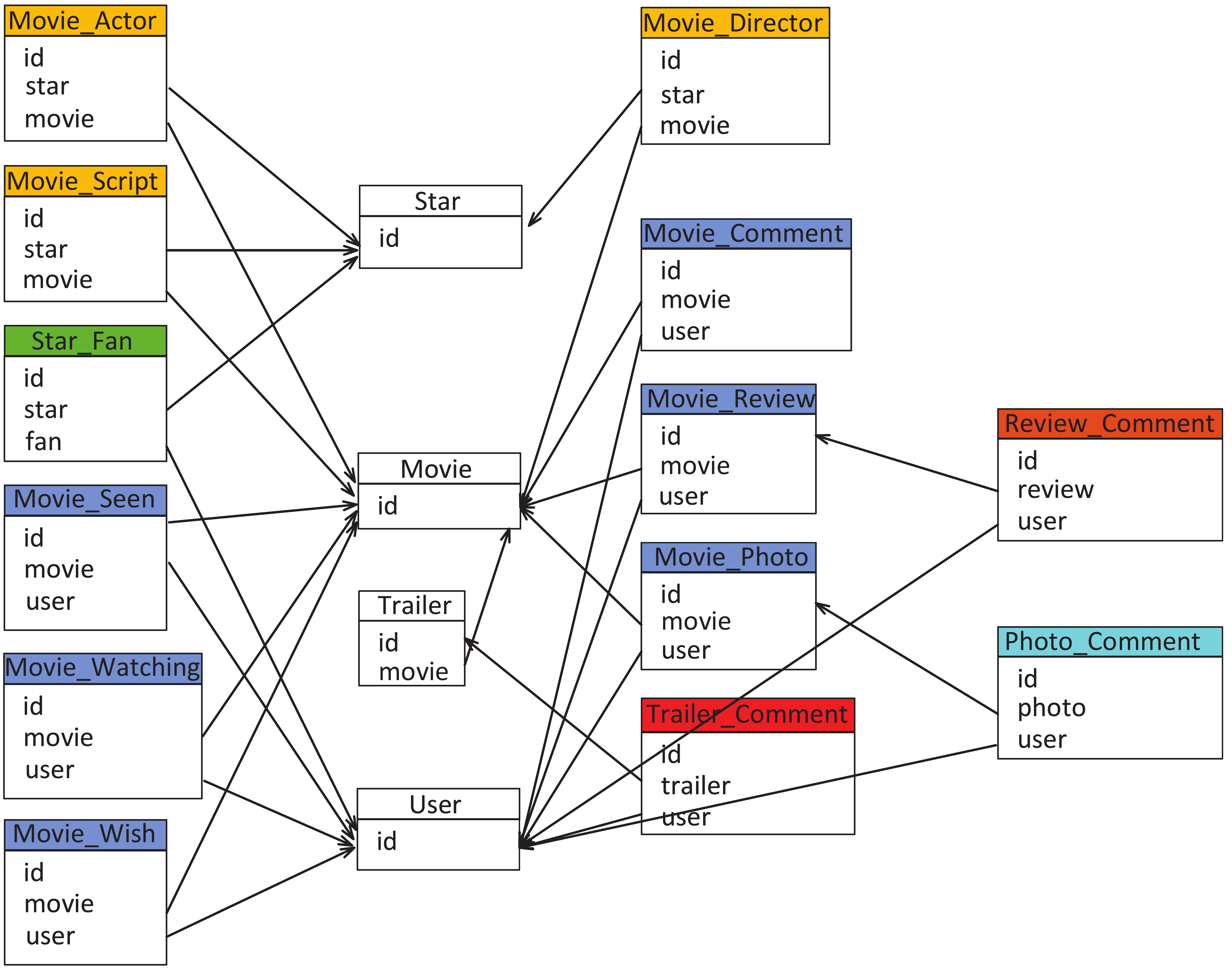}
		\caption{Schema For $\DoubanMovie$: 
			There are 17 tables;
			The tables with the same color (except grey color) share the same coappear distribution. 
			For example, {\tt Movie\_Actor}, {\tt Movie\_Script}, {\tt Movie\_Director} reference to both {\tt Star} and {\tt Movie} tables;
			{\tt Movie\_Review} and {\tt Movie\_Photo} are the {\tt post} tables; {\tt Review\_Comment} and {\tt Photo\_Comment} are the $\response$ tables.}
		\label{fig:schemamovie}
	\end{figure}
	
	\begin{figure}[h!]
		\includegraphics[height = 3.1in,trim=4 4 4 4,clip]{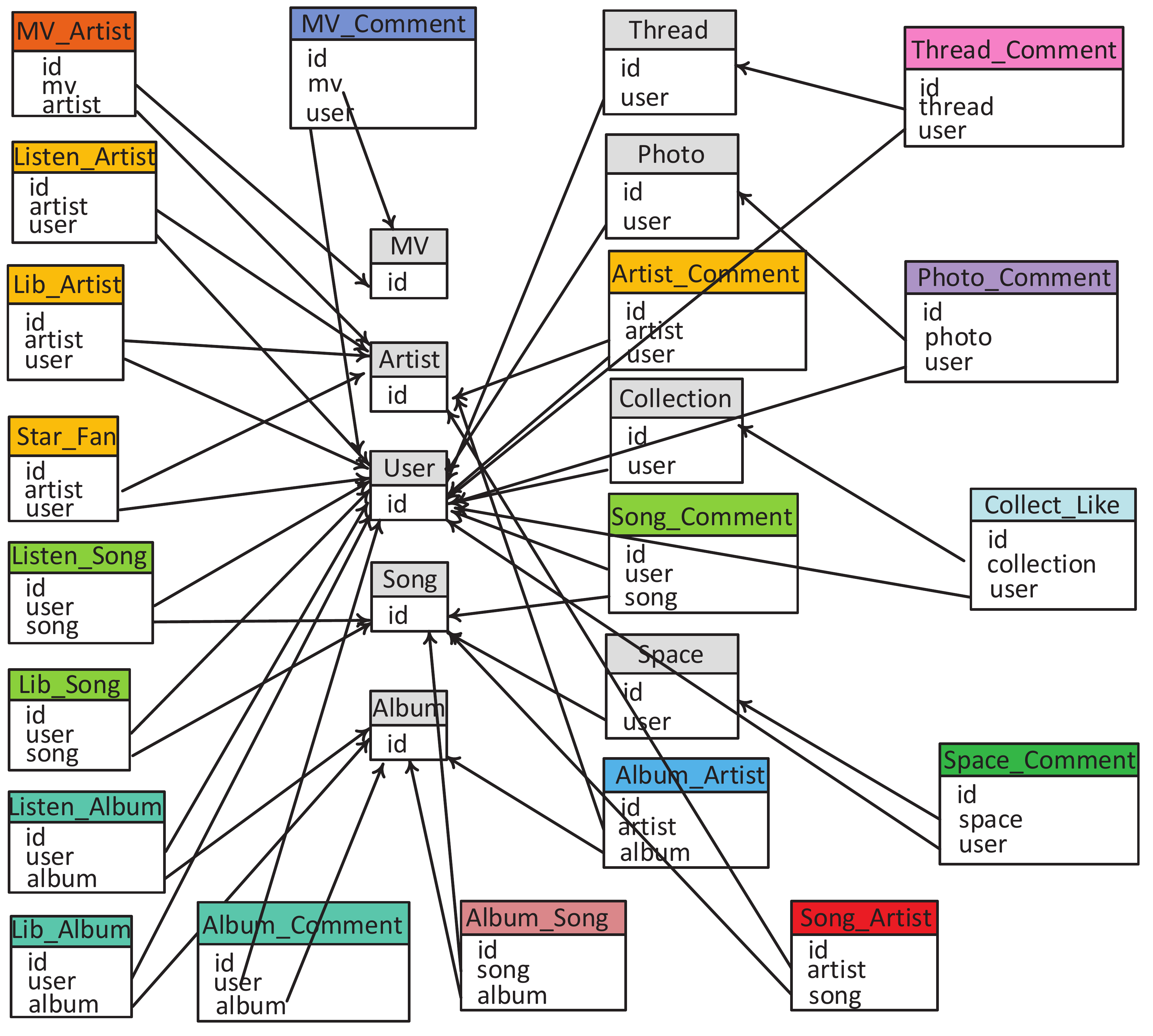}
		\caption{Schema For $\Xiami$: 
			There are 28 tables;
			The tables with the same color (except grey color) share the same coappear distribution.
			For example, {\tt Listen\_Song}, {\tt Lib\_Song} 
			reference to both {\tt Song} and {\tt User} tables;
			{\tt Collection}, {\tt Photo}, {\tt Space} and {\tt Thread}  are the {\tt post} tables;
			{\tt Photo\_Comment}, {\tt Space\_Comment}, {\tt Collect\_Like} and {\tt Thread\_Comment} are the $\response$ tables.}
		\label{fig:schemaxiami}
	\end{figure}
	
	\subsection{Feature similarity for $\DoubanMovie$, $\DoubanMusic$, $\DoubanBook$}
	In this section, we presents the feature similarity for $\DoubanMovie$, $\DoubanMusic$, $\DoubanBook$.
	\begin{figure}[t!]
		\begin{minipage}[t]{0.32\linewidth}
			\centering
			\includegraphics[height = 1.1in]{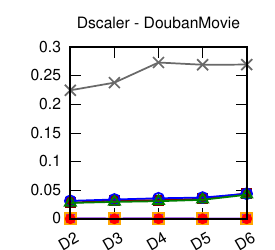}
		\end{minipage}
		\hfill
		\begin{minipage}[t]{0.32\linewidth}
			\centering
			\includegraphics[height = 1.1in]{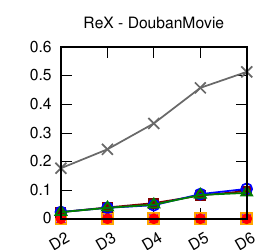}
		\end{minipage}
		\hfill
		\begin{minipage}[t]{0.32\linewidth}
			\centering
			\includegraphics[height = 1.1in]{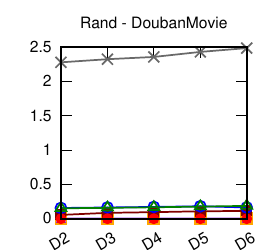}
		\end{minipage}
		\begin{minipage}[t]{0.32\linewidth}
			\centering
			\includegraphics[height = 1.1in]{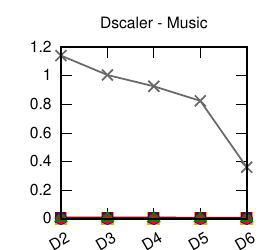}
		\end{minipage}
		\hfill\begin{minipage}[t]{0.32\linewidth}
			\centering
			\includegraphics[height = 1.1in]{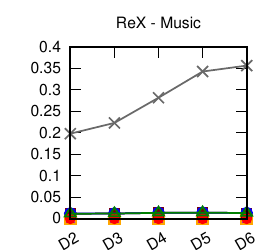}
		\end{minipage}	
		\hfill
		\begin{minipage}[t]{0.32\linewidth}
			\centering
			\includegraphics[height = 1.1in]{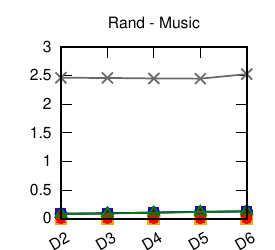}
		\end{minipage}
		\begin{minipage}[t]{0.32\linewidth}
			\centering
			\includegraphics[height = 1.1in]{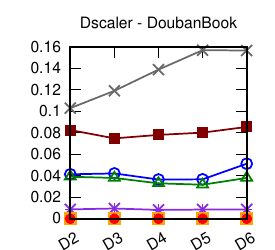}
		\end{minipage}
		\hfill
		\begin{minipage}[t]{0.32\linewidth}
			\centering
			\includegraphics[height = 1.1in]{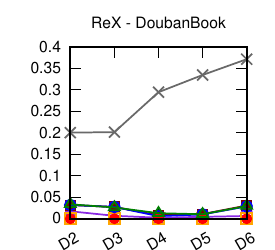}
		\end{minipage}
		\hfill\begin{minipage}[t]{0.32\linewidth}
			\centering
			\includegraphics[height = 1.1in]{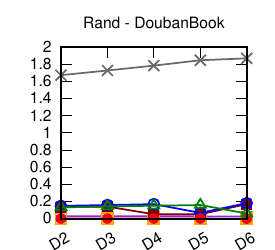}
		\end{minipage}
		\hfill
		\begin{minipage}[t]{0.99\linewidth}
			\centering
			\includegraphics[height = 0.085in]{newgnuplot/seed/xiami/legend-eps-converted-to.pdf}
		\end{minipage}
		\caption{Linear feature errors }
		\label{fig:linearMore}
		\vspace{-4mm}
	\end{figure}
	
	\subsubsection{\textbf{Linear feature similarity}}
	Fig.\ref{fig:linearMore} presents the $\linear$ feature similarity results.
	All tables are involved in at least one linear join matrix for all datasets.
	For $\DoubanMovie$, each of the 17 tables is involved in one of 24 linear
	join matrices.
	For example, {\tt Movie\_Comment } $\rightarrow$ {\tt Movie} and 
	{\tt Trailer\_Comment} $\rightarrow$ {\tt Trailer} $\rightarrow$ 
	{\tt Movie} are maximal linear joins.
	Similarly, the 12 tables in $\DoubanBook$ have 15 linear join matrices,
	the 11 tables in $\DoubanMusic$ have 14 linear join matrices,
	
	In general, the later $\Tlinear$ is applied,
	the smaller the linear feature error,
	i.e. C-L-P and P-L-C have smaller errors 
	than L-C-P and L-P-C, and C-P-L and P-C-L have 0 error.
	All permutations reduce the error tremendously for all size-scalers on all datasets.
	
	Even through the error reduction is huge, there are still some cases that the error is > 0.1.
	For example,  {\tt Rand-DoubanBook} for L-P-C. 
	It reduces $\linear$ feature error from 2 to 0.2. 
	We further investigate this issue, 
	the largest error occurs on the join {\tt Book\_Comment}  $\rightarrow$ {\tt User}. 
	For L-P-C, while tweaking the coappear distribution for $\xi_{\bf T}$, where 
	$\bf T$ is $\langle$ {\tt Book\_Comment}, {\tt Book\_Read}, 
	{\tt Book\_Reading}, {\tt Book\_Wish}, {\tt Book\_Review} $\rangle$,
	it overlaps with 1 pairwise distribution --- 
	{\tt Book\_Review} as a {\tt post} table, 
	and {\tt Review\_Comment} as a {\tt response2post} table. 
	Moreover, it overlaps with 12 linear joins 
	(e.g. {\tt Book\_Comment  $\rightarrow$ User}, {\tt Book\_Comment  $\rightarrow$ Book}).
	As stated in Sec.\ref{sec:featureSimExp},
	such highly overlapped features increase the difficulty of 
	getting a validated modification as described in Section~\ref{sec:framework}. 
	Hence, this could be a potential reason that error is > 0.1.
	
	\begin{figure}[t!]
		\begin{minipage}[t]{0.32\linewidth}
			\centering
			\includegraphics[height = 1.1in]{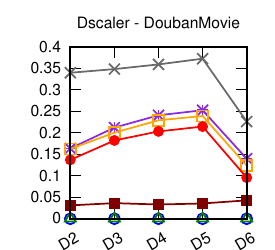}
		\end{minipage}
		\hfill
		\begin{minipage}[t]{0.32\linewidth}
			\centering
			\includegraphics[height = 1.1in]{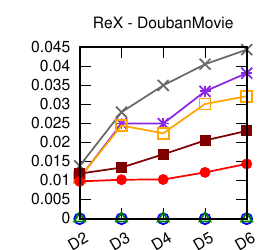}
		\end{minipage}
		\hfill
		\begin{minipage}[t]{0.32\linewidth}
			\centering
			\includegraphics[height = 1.1in]{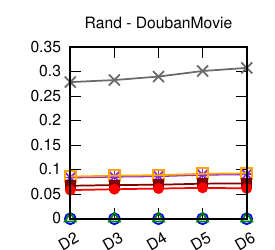}
		\end{minipage}
		\begin{minipage}[t]{0.32\linewidth}
			\centering
			\includegraphics[height = 1.1in]{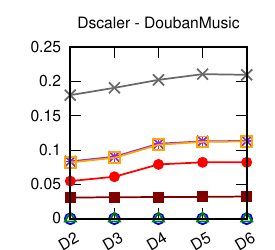}
		\end{minipage}
		\hfill\begin{minipage}[t]{0.32\linewidth}
			\centering
			\includegraphics[height = 1.1in]{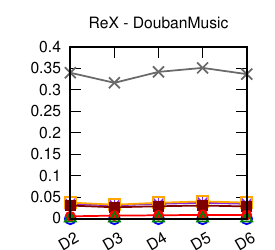}
		\end{minipage}	
		\hfill
		\begin{minipage}[t]{0.32\linewidth}
			\centering
			\includegraphics[height = 1.1in]{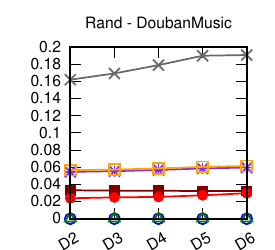}
		\end{minipage}
		\begin{minipage}[t]{0.32\linewidth}
			\centering
			\includegraphics[height = 1.1in]{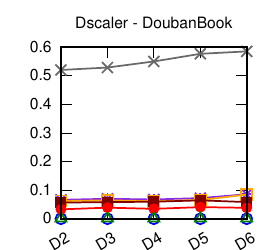}
		\end{minipage}
		\hfill
		\begin{minipage}[t]{0.32\linewidth}
			\centering
			\includegraphics[height = 1.1in]{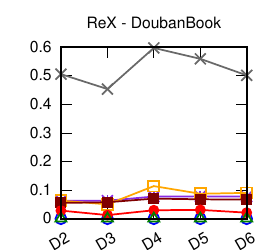}
		\end{minipage}
		\hfill\begin{minipage}[t]{0.32\linewidth}
			\centering
			\includegraphics[height = 1.1in]{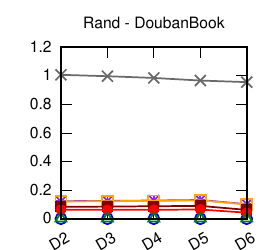}
		\end{minipage}
		\hfill
		\begin{minipage}[t]{0.99\linewidth}
			\centering
			\includegraphics[height = 0.085in]{newgnuplot/seed/xiami/legend-eps-converted-to.pdf}
		\end{minipage}
		\caption{Coappear feature errors }
		\label{fig:coappearMore}
		\vspace{-4mm}
	\end{figure}
	
	\subsubsection{\textbf{Coappear feature similarity}}
	Fig.\ref{fig:coappearMore} presents the $\coappear$ feature similarity results.
	There are 6 coappear distributions for $\DoubanMovie$. 
	For example, 
	The 6 tables $\langle $ {\tt  Movie\_Seen}, {\tt Movie\_Watching}, 
	{\tt Movie\_Wish}, {\tt Movie\_Photo}, {\tt Movie\_Review}, 
	{\tt Movie\_Comment} $\rangle$ reference {\tt Movie} and {\tt User}.
	Similarly, $\DoubanMusic$ has 4 coappear distributions,
	$\DoubanBook$ has 4 and $\Xiami$ has 12.
	In each case, each table is involved in one or more coappear distributions.
	
	Fig.\ref{fig:coappearMore} shows that, like for $\Tlinear$,
	the later $\Tcoappear$ is applied in the tweaking order,
	the smaller the coappear error.
	In general, we find that permutations where $\Tcoappear$ is after
	$\Tlinear$ reduces the errors more than if $\Tcoappear$ is before $\Tlinear$.
	This is expected, since $\Tlinear$ modifies the coappearing tables
	massively after $\Tcoappear$ is done.
	
	For average error, all permutations of tweaking significantly 
	reduce the error for all datasets for all size-scalers.
	It is below 0.1 for all most tweaking.
	
	For the plot {\tt Dscaler-DoubanMovie}, 
	we observe that the tweaking permutations ($\Tlinear$ applied after $\Tcoappear$) have an error around $0.2$.
	By looking at the details, 
	we find that this happens for the coappear distribution involving many tables. 
	Take $\xi_{\bf T}$ for example, where 
	$\bf T$ is $\langle$ {\tt Movie\_Comment}, {\tt Movie\_Seen}, 
	{\tt Movie\_Watching}, {\tt Movie\_Wish}, {\tt Movie\_Review},{\tt Movie\_Photo} $\rangle$. 
	This coappear distribution overlaps with 12 linear join matrices and 2 pairwise distribution.
	This coappear distribution will be modified by 12 linear tweaking tools if $\Tlinear$ applied after $\Tcoappear$.
	Hence, increase the difficulty of 
	getting a validated modification as described in Section~\ref{sec:framework}. 
	Nevertheless, we still have a small error for {\tt ReX-DoubanMovie, Rand-DoubanMovie} for such a highly overlapped structure.
	
	For  {\tt ReX-DoubanMovie}, 
	even through the error without tweaking is as low as 0.01. 
	All tweaking permutations are still able to reduce the error.

	\begin{figure}[t!]
		\begin{minipage}[t]{0.32\linewidth}
			\centering
			\includegraphics[height = 1.1in]{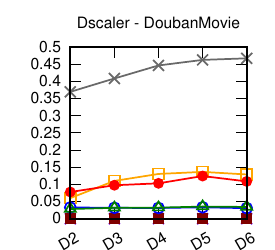}
		\end{minipage}
		\hfill
		\begin{minipage}[t]{0.32\linewidth}
			\centering
			\includegraphics[height = 1.1in]{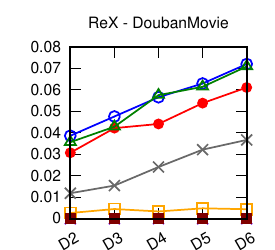}
		\end{minipage}
		\hfill
		\begin{minipage}[t]{0.32\linewidth}
			\centering
			\includegraphics[height = 1.1in]{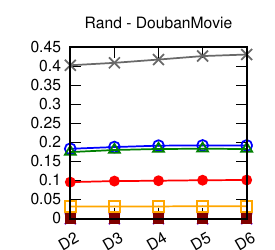}
		\end{minipage}
		\begin{minipage}[t]{0.32\linewidth}
			\centering
			\includegraphics[height = 1.1in]{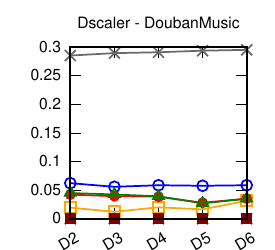}
		\end{minipage}
		\hfill\begin{minipage}[t]{0.32\linewidth}
			\centering
			\includegraphics[height = 1.1in]{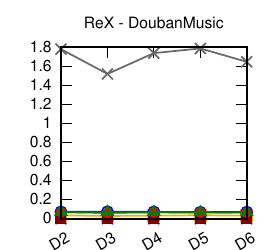}
		\end{minipage}	
		\hfill
		\begin{minipage}[t]{0.32\linewidth}
			\centering
			\includegraphics[height = 1.1in]{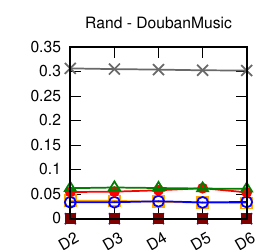}
		\end{minipage}
		\begin{minipage}[t]{0.32\linewidth}
			\centering
			\includegraphics[height = 1.1in]{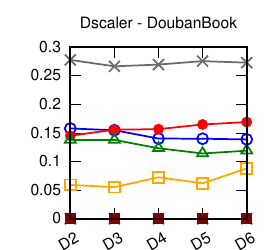}
		\end{minipage}
		\hfill
		\begin{minipage}[t]{0.32\linewidth}
			\centering
			\includegraphics[height = 1.1in]{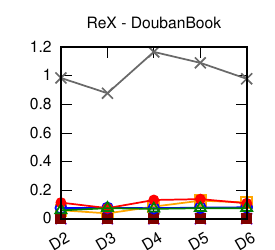}
		\end{minipage}
		\hfill\begin{minipage}[t]{0.32\linewidth}
			\centering
			\includegraphics[height = 1.1in]{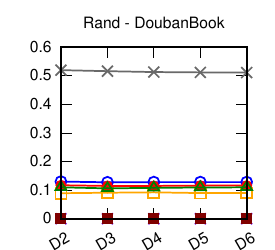}
		\end{minipage}
		\hfill
		\begin{minipage}[t]{0.99\linewidth}
			\centering
			\includegraphics[height = 0.085in]{newgnuplot/seed/xiami/legend-eps-converted-to.pdf}
		\end{minipage}
		\caption{Pairwise feature errors }
		\label{fig:pairwiseMore}
		\vspace{-4mm}
	\end{figure}
	
	\subsubsection{\textbf{Pairwise feature similarity}}
	Fig.\ref{fig:pairwiseMore} presents the $\pairwise$ feature similarity results.
	$\DoubanMovie$ has 2 pairwise distributions:
	(i) {\tt Review} as {\tt post} table and {\tt Review\_Comment} as
	$\response$ table; and
	(ii) {\tt Photo} as {\tt post} and {\tt Photo\_Comment} as $\response$. 
	$\DoubanMusic$ has 1 pairwise distribution, 
	and $\DoubanBook$ has 2 pairwise distributions.
	
	Fig.\ref{fig:pairwiseMore} again shows that,
	the later $\Tpairwise$ is applied in a tweaking order,
	the smaller the pairwise feature error in the tweaked dataset.
	For $\DoubanMusic$ and $\Xiami$,
	all tweaking permutations reduce the errors tremendously for all size-scalers.	
	For $\DoubanMovie$, 
	all tweaking permutations on data generated by $\dscaler$ significantly 
	reduce the pairwise feature error; 
	most tweaking permutations on data generated by Rand significantly reduce the error, except L-P-C and P-L-C.
	For {\tt Dscaler-DoubanMovie},  
	the error without tweaking is small (< 0.05),
	some tweaking permutations increase the errors.
	For $\DoubanBook$, 
	all tweaking permutations reduce the errors tremendously for all size-scalers 
	except three tweaking permutations on {\tt Dscaler-DoubanBook}.
	
	\subsection{Query similarity for $\DoubanMovie$, $\DoubanMusic$, $\DoubanBook$ }
	In this section, 
	we similarly run queries on  $\DoubanMovie$, $\DoubanMusic$, $\DoubanBook$, 
	and compare the query results on ground-truth dataset and scaled dataset.
	
	\subsubsection{\textbf{Query similarity for $\DoubanMovie$}} 	
	Fig.\ref{fig:movieQuery} presents the query results on $\DoubanMovie$.
	The $4$ queries used are: 
	$Q_1$ computes the number of movies that have video clips with commenters;
	$Q_2$ computes the number of movies that have been commented on by at most 10 different users;
	$Q_3$ computes the average number of stars per movie;
	$Q_4$ computes the number of user pairs having interactions through a movie review.
	
	As we can see from Fig.\ref{fig:movieQuery}, 
	all tweaking permutations reduce the query error significantly on both size-scalers.
	The errors are reduced to $ < 0.05$ for most of the tweaking permutations.

	\begin{figure}[t]
		\begin{minipage}[t]{0.24\linewidth}
			\centering
			\includegraphics[height = 0.95in]{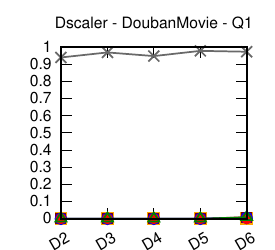}
		\end{minipage}
		\hfill
		\begin{minipage}[t]{0.24\linewidth}
			\centering
			\includegraphics[height = 0.95in]{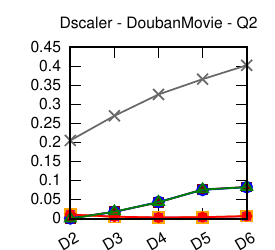}
		\end{minipage}
		\hfill
		\begin{minipage}[t]{0.24\linewidth}
			\centering
			\includegraphics[height = 0.95in]{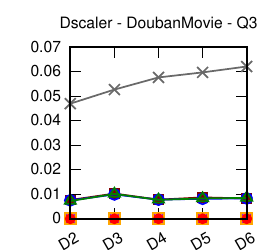}	
		\end{minipage}
		\hfill
		\begin{minipage}[t]{0.24\linewidth}
			\centering
			\includegraphics[height = 0.95in]{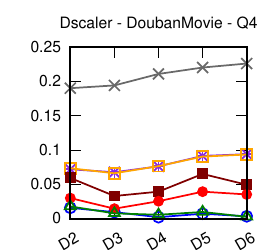}
		\end{minipage}
		\hfill
		\begin{minipage}[t]{0.24\linewidth}
			\centering
			\includegraphics[height = 0.95in]{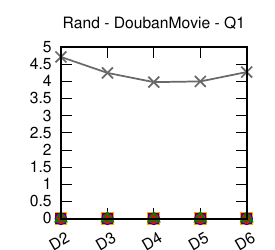}
		\end{minipage}
		\hfill
		\begin{minipage}[t]{0.24\linewidth}
			\centering
			\includegraphics[height = 0.95in]{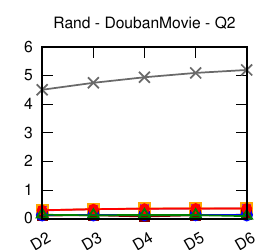}
		\end{minipage}
		\hfill
		\begin{minipage}[t]{0.24\linewidth}
			\centering
			\includegraphics[height = 0.95in]{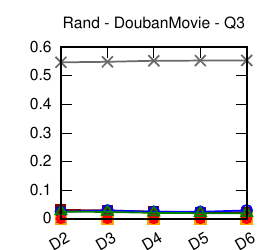}	
		\end{minipage}
		\hfill
		\begin{minipage}[t]{0.24\linewidth}
			\centering
			\includegraphics[height = 0.95in]{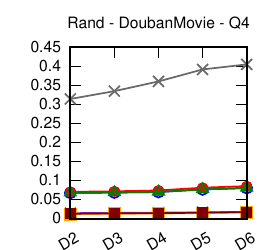}
		\end{minipage}
		\hfill
		\begin{minipage}[t]{0.99\linewidth}
			\centering
			\includegraphics[height = 0.085in]{newgnuplot/seed/xiami/legend-eps-converted-to.pdf}
		\end{minipage}
		\caption{Query similarity for $\DoubanMovie$}
		\vspace{-4mm}
		\label{fig:movieQuery}
	\end{figure}

	\subsubsection{\textbf{Query similarity for $\DoubanMusic$}} 	
	Fig.\ref{fig:musicQuery} presents the query results on $\DoubanMusic$.
	The $4$ queries used are: 
	$Q_1$ computes the number of users that have written a album-view with commenters;
	$Q_2$ computes the number of stars that have at most 10 different fans;
	$Q_3$ computes the average number of interested listeners of a album;
	$Q_4$ computes the number of user pairs having interactions through a album review.
	
	Similar to $\DoubanMovie$, all permutations reduce the errors tremendously.

	\begin{figure}[t]
		\begin{minipage}[t]{0.24\linewidth}
			\centering
			\includegraphics[height = 0.95in]{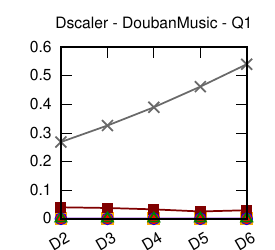}
		\end{minipage}
		\hfill
		\begin{minipage}[t]{0.24\linewidth}
			\centering
			\includegraphics[height = 0.95in]{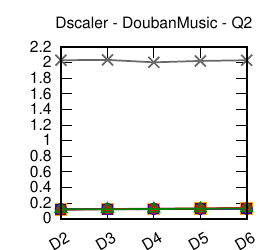}
		\end{minipage}
		\hfill
		\begin{minipage}[t]{0.24\linewidth}
			\centering
			\includegraphics[height = 0.95in]{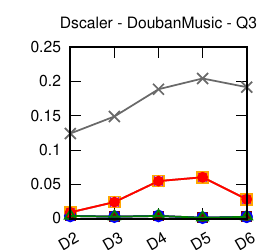}	
		\end{minipage}
		\hfill
		\begin{minipage}[t]{0.24\linewidth}
			\centering
			\includegraphics[height = 0.95in]{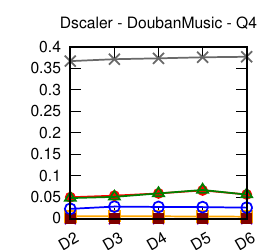}
		\end{minipage}
		\hfill
		\begin{minipage}[t]{0.24\linewidth}
			\centering
			\includegraphics[height = 0.95in]{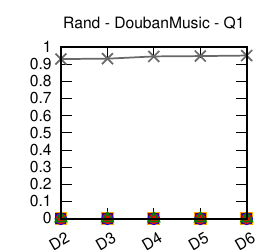}
		\end{minipage}
		\hfill
		\begin{minipage}[t]{0.24\linewidth}
			\centering
			\includegraphics[height = 0.95in]{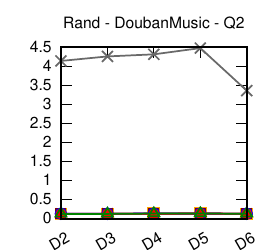}
		\end{minipage}
		\hfill
		\begin{minipage}[t]{0.24\linewidth}
			\centering
			\includegraphics[height = 0.95in]{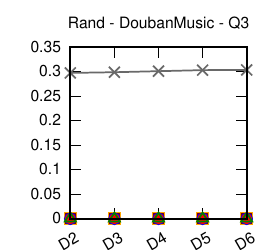}	
		\end{minipage}
		\hfill
		\begin{minipage}[t]{0.24\linewidth}
			\centering
			\includegraphics[height = 0.95in]{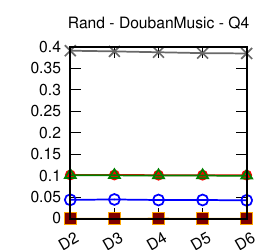}
		\end{minipage}
		\hfill
		\begin{minipage}[t]{0.99\linewidth}
			\centering
			\includegraphics[height = 0.085in]{newgnuplot/seed/xiami/legend-eps-converted-to.pdf}
		\end{minipage}
		\caption{Query similarity for $\DoubanMusic$}
		\label{fig:musicQuery}
		\vspace{-6mm}
	\end{figure}

	\subsubsection{\textbf{Query similarity for $\DoubanBook$}} 	
	Fig.\ref{fig:bookQuery} presents the query results on $\DoubanBook$.
	The $4$ queries used are: 
	$Q_1$ computes the number of users that have written a book-view with commenters;
	$Q_2$ computes the number of diaries that have at most 10 different commenters;
	$Q_3$ computes the average number of interested readers of a book;
	$Q_4$ computes the number of user pairs having interactions through a book review.
	
	As we can see from Fig.\ref{fig:bookQuery}, 
	most of the tweaking permutations reduce the errors tremendously except for few rare cases, 
	e.g. \texttt{Dscaler--DoubanBook-Q1}.
	For the L-C-P permutation, we can see that it has a larger error than the baseline. 
	This is expected, since $Q_1$ is a $\linear$ feature related query, 
	and the $\linear$ feature that were tweaked by $\Tlinear$ is subsequently modified by $\Tcoappear$ and $\Tpairwise$.
	Such a scenario can be improved by having more iterations.
	In Fig.\ref{fig:queryimprove}, we run L-C-P on Dscaler-DoubanBook with more iterations.
	We can see that from second iteration onwards, the Q1 error is reduced to less than 0.001.

	\begin{figure}[t]
		\begin{minipage}[t]{0.24\linewidth}
			\centering
			\includegraphics[height = 0.95in]{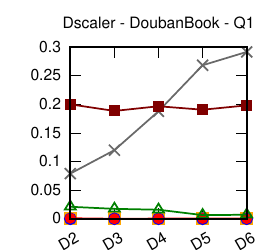}
		\end{minipage}
		\hfill
		\begin{minipage}[t]{0.24\linewidth}
			\centering
			\includegraphics[height = 0.95in]{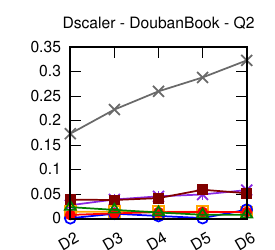}
		\end{minipage}
		\hfill
		\begin{minipage}[t]{0.24\linewidth}
			\centering
			\includegraphics[height = 0.95in]{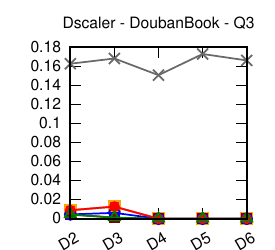}	
		\end{minipage}
		\hfill
		\begin{minipage}[t]{0.24\linewidth}
			\centering
			\includegraphics[height = 0.95in]{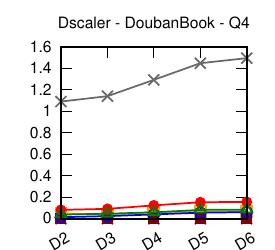}
		\end{minipage}
		\hfill
		\begin{minipage}[t]{0.24\linewidth}
			\centering
			\includegraphics[height = 0.95in]{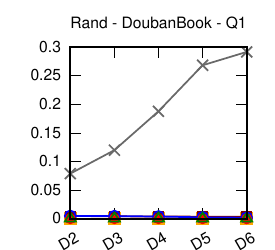}
		\end{minipage}
		\hfill
		\begin{minipage}[t]{0.24\linewidth}
			\centering
			\includegraphics[height = 0.95in]{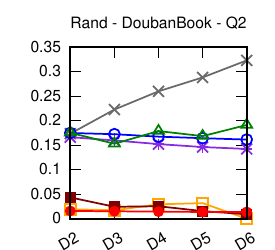}
		\end{minipage}
		\hfill
		\begin{minipage}[t]{0.24\linewidth}
			\centering
			\includegraphics[height = 0.95in]{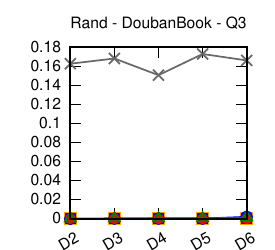}	
		\end{minipage}
		\hfill
		\begin{minipage}[t]{0.24\linewidth}
			\centering
			\includegraphics[height = 0.95in]{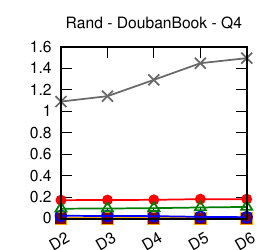}
		\end{minipage}
		\hfill
		\begin{minipage}[t]{0.99\linewidth}
			\centering
			\includegraphics[height = 0.085in]{newgnuplot/seed/xiami/legend-eps-converted-to.pdf}
		\end{minipage}
		\caption{Query similarity for $\DoubanBook$}
		\label{fig:bookQuery}
		\vspace{-4mm}
	\end{figure}

	\begin{figure}
		\centering
		\includegraphics[height=0.9in]{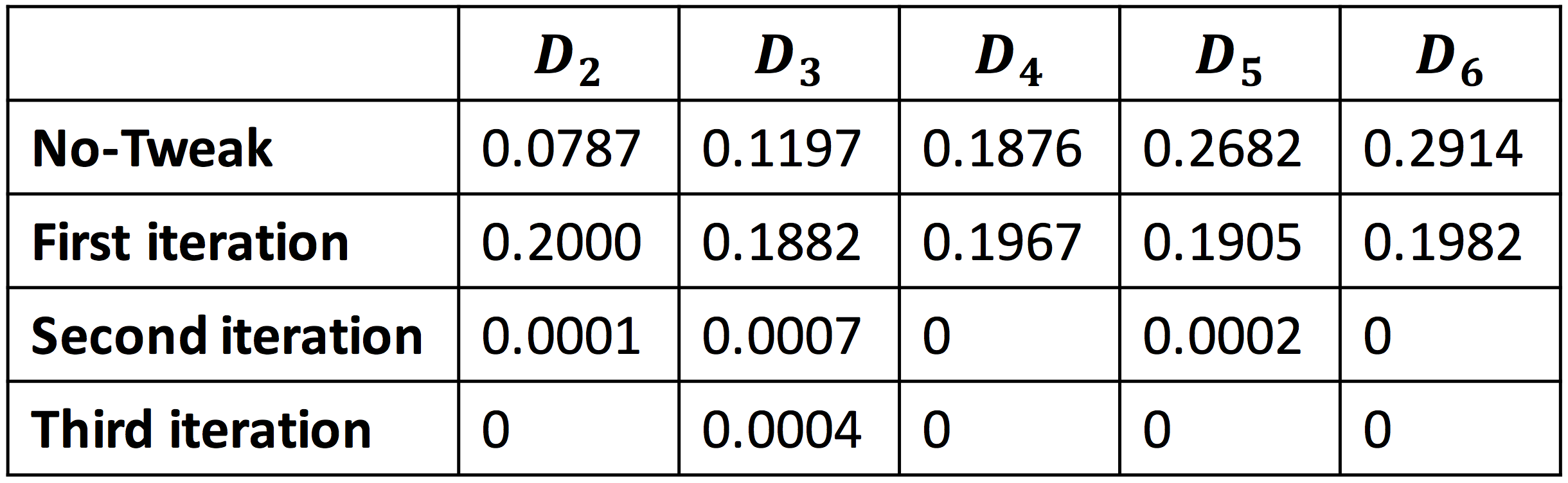}
		\caption{Query similarity improvement with more iterations. We run L-C-P for more iterations on Dscaler--DoubanBook, and test Q1. }
		\label{fig:queryimprove}
		\vspace{-4mm}
	\end{figure}

	\begin{figure}[h!]
		\centering
		\includegraphics[height=1.2in]{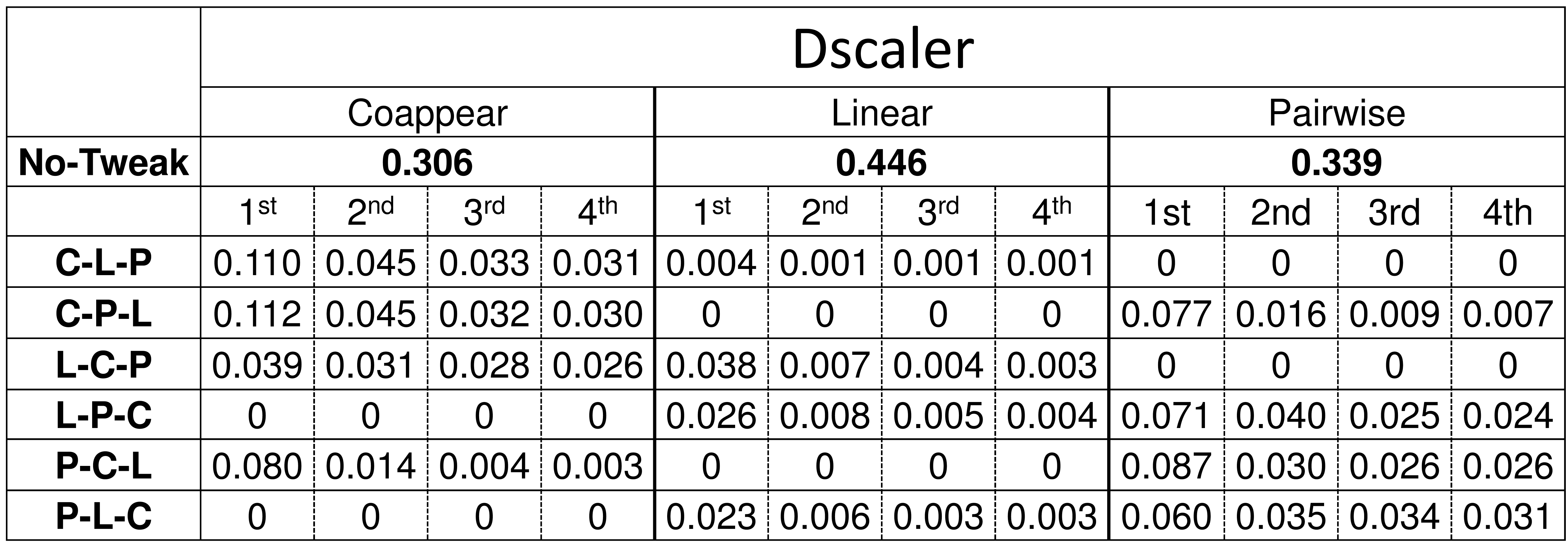}
		\caption{Feature errors of using $\dscaler$  as a size-scaler}
		\label{fig:dscalerIter}
		\vspace{-6mm}
	\end{figure}
	\begin{figure}[h!]
		\centering
		\includegraphics[height=1.2in]{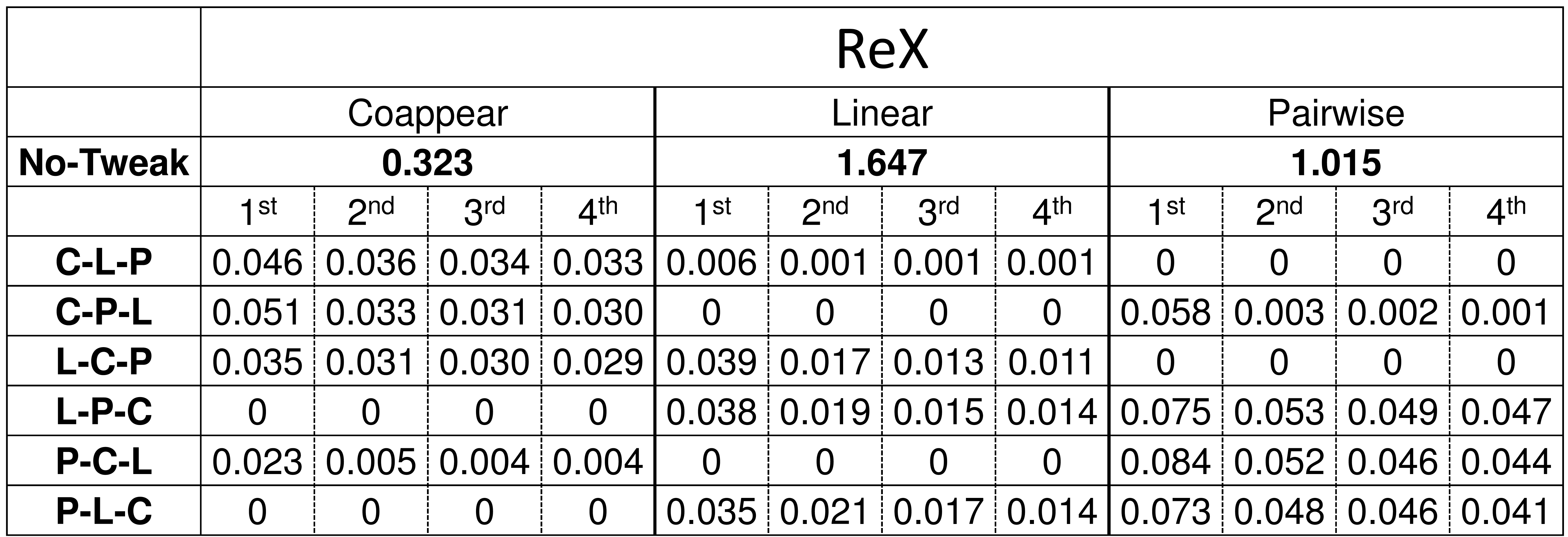}
		\caption{Feature errors of using ReX  as a size-scaler}
		\label{fig:rexIter}
		\vspace{-5mm}
	\end{figure}
	\begin{figure}[h!]
		\centering
		\includegraphics[height=1.2in]{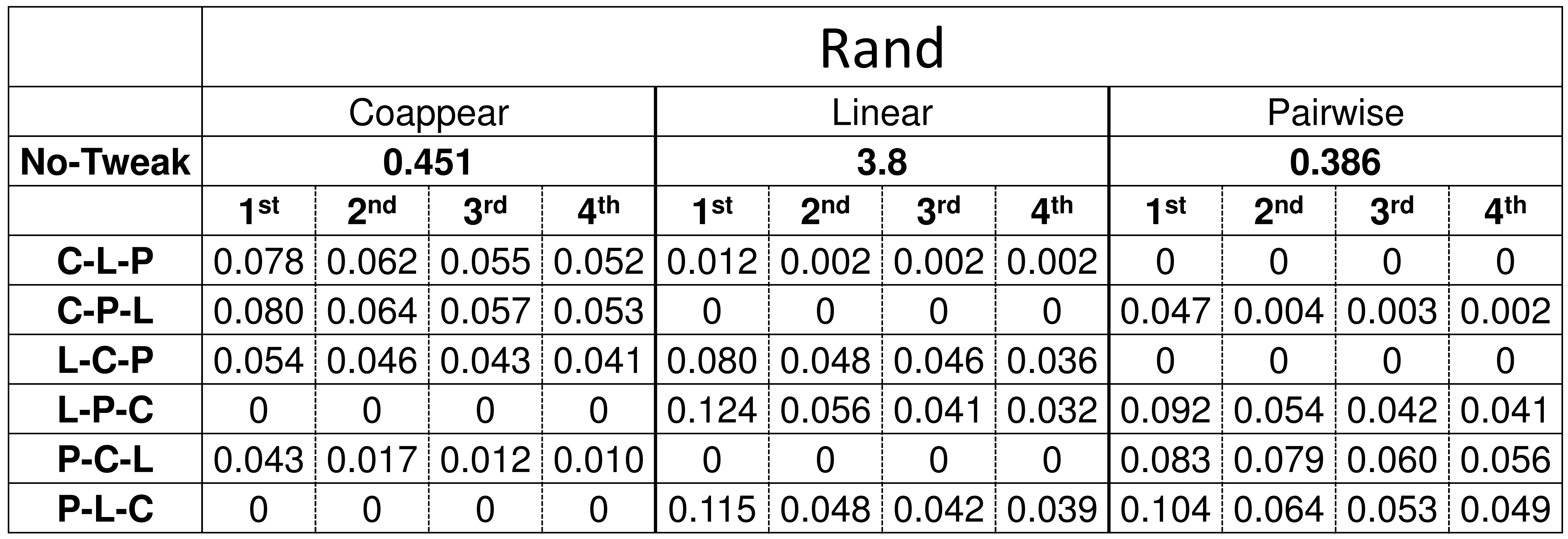}
		\caption{Feature errors of using Rand  as a size-scaler}
		\label{randIter}
		\vspace{-5mm}
	\end{figure}
	
	\subsection{Similarity improvement over iterations }
	In this section, we present feature similarity results for different iterations.
	Fig.\ref{fig:dscalerIter}, Fig.\ref{fig:rexIter} and Fig.\ref{randIter} 
	present the results of running 6 tweaking permutations for up to 4 iterations 
	on the dataset generated by  $\dscaler$, $\rex$ and $\rand$.
	Take Fig.\ref{fig:dscalerIter} for example, 
	for $\coappear$ feature, 4th column for C-L-P is 0.031. 
	It means that after running C-L-P permutation on the data generated by $\dscaler$ for 4 times, 
	the $\coappear$ feature error is 0.031. It is a 10-fold decrease from 0.306 (the No-Tweak baseline).
	
	For all the three figures, 
	we can see that the more iterations of tweaking, the less error we will have.
	On average, ASPECT can achieve an error of around 0.02 after 2 or 3 iterations.

	\subsection{Execution time for $\DoubanMovie$, $\DoubanMusic$, $\DoubanBook$}
	\begin{figure}[t!]
		\begin{minipage}[t]{0.32\linewidth}
			\centering
			\includegraphics[height = 1.1in]{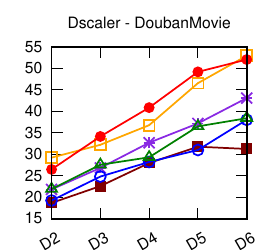}
		\end{minipage}
		\hfill
		\begin{minipage}[t]{0.32\linewidth}
			\centering
			\includegraphics[height = 1.1in]{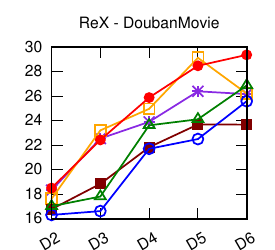}
		\end{minipage}
		\hfill
		\begin{minipage}[t]{0.32\linewidth}
			\centering
			\includegraphics[height = 1.1in]{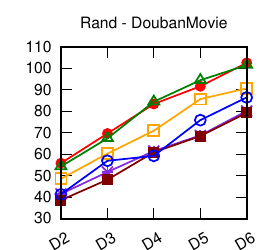}
		\end{minipage}
		\begin{minipage}[t]{0.32\linewidth}
			\centering
			\includegraphics[height = 1.1in]{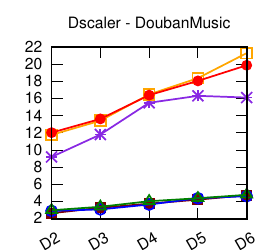}
		\end{minipage}
		\hfill\begin{minipage}[t]{0.32\linewidth}
			\centering
			\includegraphics[height = 1.1in]{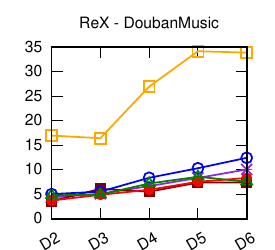}
		\end{minipage}	
		\hfill
		\begin{minipage}[t]{0.32\linewidth}
			\centering
			\includegraphics[height = 1.1in]{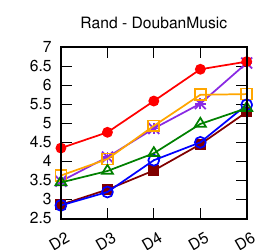}
		\end{minipage}
		\begin{minipage}[t]{0.32\linewidth}
			\centering
			\includegraphics[height = 1.1in]{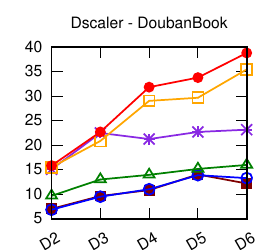}
		\end{minipage}
		\hfill
		\begin{minipage}[t]{0.32\linewidth}
			\centering
			\includegraphics[height = 1.1in]{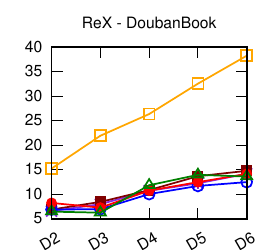}
		\end{minipage}
		\hfill\begin{minipage}[t]{0.32\linewidth}
			\centering
			\includegraphics[height = 1.1in]{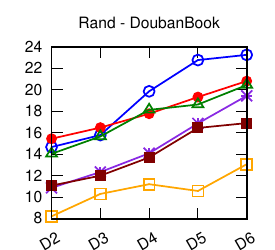}
		\end{minipage}
		\hfill
		\begin{minipage}[t]{0.99\linewidth}
			\centering
			\includegraphics[height = 0.085in]{newgnuplot/seed/xiami/legend-eps-converted-to.pdf}
		\end{minipage}
		\caption{Execution time }
		\label{fig:timeMore}
		\vspace{-4mm}
	\end{figure}

	We can see that, 
	the execution time increases linearly 
	with the dataset size for most of the experiments.
	{\tt DoubanMovie} is the largest dataset, 
	it takes more time.
	Nevertheless, most experiments finishes with 60 minutes 
	for the largest snapshot of {\tt DoubanMovie}.  
	{\tt DoubanMusic} and {\tt DoubanBook} are the smaller datasets,
	hence, it takes less time, within 60 minutes,
	for the worst tweaking permutation.
	
	For the same dataset, 
	different size-scaler will result in different execution time.
	This is understandable, 
	the data generated by the size-scalers have different feature errors.
	Hence, the amount of tweaking is different. 
	Take {\tt DoubanMovie} for example,  
	the execution time for each permutation varies among the different size-scaler.
	Moreover, for the same size-scaler and the same dataset, 
	different tweaking permutation has different execution time. 
	In general we find that L-C-P and L-P-C 
	are more efficient than other tweaking permutations.

\end{document}